\documentclass[11pt]{article}
\usepackage{hyperref}
\pdfoutput=1

\usepackage{chngcntr}

\usepackage[final]{graphicx}
\usepackage{subfigure}
\usepackage{amssymb,amsthm}
\usepackage{amsmath}
\usepackage{wrapfig}
\usepackage{algorithm}
\usepackage{algorithmic}
\usepackage{appendix}
\usepackage{bm}
\usepackage{cite}
\usepackage{enumerate}
\usepackage[papersize={8.5in,11in},top=1in,left=1in,right=1in,bottom=1in]{geometry}
\usepackage{xfrac}
\newenvironment{lessspaceitemize*}%
  {\begin{itemize}%
  \vspace{-1mm}
    \setlength{\itemsep}{0pt}%
    \setlength{\parskip}{0pt}}%
    {\vspace{-1mm} \end{itemize}}

\newenvironment{lessspaceenum*}%
  {\begin{enumerate}%
  \vspace{-1mm}
    \setlength{\itemsep}{0pt}%
    \setlength{\parskip}{0pt}}%
  {\end{enumerate}}

\counterwithout{figure}{section}

\newenvironment{definition}[1][Definition]{\begin{trivlist}
\item[\hskip \labelsep {\bfseries #1}]}{\end{trivlist}}
\newenvironment{example}[1][Example]{\begin{trivlist}
\item[\hskip \labelsep {\bfseries #1}]}{\end{trivlist}}

\newtheorem{theorem}{Theorem}[section]

\newtheorem{lemma}[theorem]{Lemma}
\newtheorem{claim}[theorem]{Claim}
\newtheorem{proposition}[theorem]{Proposition}
\newtheorem{corollary}[theorem]{Corollary}

\newenvironment{thm_app}[1]{\noindent\textbf{Theorem~\ref{#1}.}}{\par\addvspace{\baselineskip}}

\newenvironment{clm_app}[1]{\noindent\textbf{Claim~\ref{#1}.}}{\par\addvspace{\baselineskip}}

\newenvironment{cor_app}[1]{\noindent\textbf{Corollary~\ref{#1}.}}{\par\addvspace{\baselineskip}}

\newcommand{\lx}{\ensuremath{\lambda(x)}}%
\newcommand{\ldx}{\ensuremath{\lambda'(x)}}%
\newcommand{\pe}{\ensuremath{\overline{p}_e}}%
\title{Price Competition in Networked Markets: How do Monopolies Impact Social Welfare?}
\author{Elliot Anshelevich\thanks{Computer Science Dept, Rensselaer Polytechnic Institute, Troy, NY. \texttt{eanshel@cs.rpi.edu}}
\and Shreyas Sekar\thanks{Computer Science Dept, Rensselaer Polytechnic Institute, Troy, NY. \texttt{sekars@rpi.edu}}}

\begin{document}

% Page heads
%\markboth{ANSHELEVICH AND SEKAR}{Market price competition in networks: How do monopolies impact social welfare?}

% Title portion
% NOTE! Affiliations placed here should be for the institution where the
%       BULK of the research was done. If the author has gone to a new
%       institution, before publication, the (above) affiliation should NOT be changed.
%       The authors 'current' address may be given in the "Author's addresses:" block (below).
%       So for example, Mr. Abdelzaher, the bulk of the research was done at UIUC, and he is
%       currently affiliated with NASA.
\maketitle

\begin{abstract}
We study the efficiency of allocations in large markets with a network structure where every seller owns an edge in a graph and every buyer desires a path connecting some nodes. While it is known that stable allocations in such settings can be very inefficient, the exact properties of equilibria in markets with multiple sellers are not fully understood even in single-source single-sink networks. In this work, we show that for a large class of natural buyer demand functions, we are guaranteed the existence of an equilibrium with several desirable properties. The crucial insight that we gain into the equilibrium structure allows us to obtain tight bounds on efficiency in terms of the various parameters governing the market, especially the number of monopolies $M$. All of our efficiency results extend to markets without the network structure.

While it is known that monopolies can cause large inefficiencies in general, our main results for single-source single-sink networks indicate that for several natural demand functions the efficiency only drops linearly with $M$. For example,  for concave demand we prove that the efficiency loss is at most a factor $1+\frac{M}{2}$ from the optimum,  for demand with monotone hazard rate it is at most $1+M$, and for polynomial demand the efficiency decreases logarithmically with $M$. 
In contrast to previous work that showed that monopolies may adversely affect welfare, our main contribution is showing that monopolies may not be as `evil' as they are made out to be; the loss in efficiency is bounded in many natural markets. Finally, we consider more general, multiple-source networks and show that in the absence of monopolies, mild assumptions on the network topology guarantee an equilibrium that maximizes social welfare. \end{abstract}

%\category{J.4}{Social and Behavioral Sciences}{Economics}
%
%\terms{Theory, Economics}
%
%\keywords{Market Equilibria, Quality of Equilibrium, Bertrand Networks}
%
%\acmformat{Anshelevich, Elliot and Sekar, Shreyas 2015. Market price competition in networks: How do monopolies impact social welfare?}
%% At a minimum you need to supply the author names, year and a title.
%% IMPORTANT:
%% Full first names whenever they are known, surname last, followed by a period.
%% In the case of two authors, 'and' is placed between them.
%% In the case of three or more authors, the serial comma is used, that is, all author names
%% except the last one but including the penultimate author's name are followed by a comma,
%% and then 'and' is placed before the final author's name.
%% If only first and middle initials are known, then each initial
%% is followed by a period and they are separated by a space.
%% The remaining information (journal title, volume, article number, date, etc.) is 'auto-generated'.
%
%\begin{bottomstuff}
%Author's addresses: Department of Computer Science, Rensselaer Polytechnic Institute, Troy, NY 12180; email: enshel@cs.rpi.edu, sekars@rpi.edu.
%\end{bottomstuff}

\section{Introduction}
The mechanism governing large decentralized markets is often straightforward: sellers post prices for their goods and buyers buy bundles that meet their requirements. Given this framework, the challenge faced by researchers has been to characterize the equilibrium states at which these markets operate. More concretely, consider a market with multiple sellers that can be represented by a network as follows:
\begin{itemize}
\item Every seller owns an item, which is a link in the network.
\item Every infinitesimal buyer seeks to purchase a path in the network (set of items) connecting some pair of nodes.
\end{itemize}

In addition to actual bandwidth markets where users purchase capacity on links for routing traffic, networks are commonly used in the literature to model combinatorial markets where the items are a mix of substitutes and complements~\cite{acemoglu2007competition,chawla2008bertrand,kuleshov2012efficiency,melo2014price}. As an illustrative example, consider the (sample) sandwich market of substitutes and complements in Figure~\ref{fig_examplemarket}; the prices of different sandwich ingredients are controlled by different sellers, and every buyer wants to purchase a set of ingredients to make a full sandwich. Our goal in this paper is to analyze the effects of \emph{price competition} in such networked markets, i.e., the pricing strategies employed by competing sellers and their effect on equilibrium welfare. 

In this paper, we are interested in the following question: ``How efficient are the equilibrium allocations in such markets and how do they depend on \emph{buyer demand} and \emph{network structure}?". Our interest stems in part, from the fact that a vast majority of literature has focused exclusively on Walrasian Equilibrium as the market operating point. There is limited understanding of efficiency beyond the Walrasian framework. Walrasian Equilibria are indeed attractive: they always exist in large markets~\cite{azevedo2013walrasian} and they are guaranteed to be optimal. However, the idea that prices are just `handed out' to buyers and sellers so that the market clears may not always be applicable in a decentralized market. As~\cite{gale2000strategic} remarks,
\begin{quotation}
``...an embarrassing lacuna in the theory of walrasian equilibrium is the failure to explain where the prices come from."
\end{quotation}

In contrast, we take the view that the sellers control their own prices in order to maximize revenue and not just clear the market, i.e., they act as \emph{price-setters} and not \emph{price-takers}. In such a setting, it does not make sense to consider Walrasian Equilibrium as the notion of stability. Our objective therefore, is to quantify the inefficiency as compared to Walrasian Equilibrium due to the strategic behavior of agents. A large body of work in Computer Science has attempted to bound this inefficiency~\cite{babaioff2014efficiency, chen2012incentive, hassidim2011non}, albeit with a focus on the strategic behavior of buyers who can misreport their preferences, in settings with a single central seller. On the other hand, we consider a market with many strategic sellers and a continuum of buyers. In such a model, it is reasonable to assume that buyers behave as price-taking agents since their individual demand is small compared to the market size. Moreover, in large markets, the seller does not need to know every buyer's value as long as they can anticipate the aggregate demand at any given price. 

\subsection*{Our Model and Equilibrium Concept}
We model the interaction between buyers and sellers as a two-stage pricing game in a networked market. Each seller $e$ controls a single good or link in the network; he can produce any quantity $x$ of this good incurring a production cost of $C_e(x)$. Every non-atomic buyer $i$ in the market wants to purchase infinitesimal amount of some path (bundle of edges) connecting a source and a sink node for which she receives a value $v_i$. For the majority of this work, we will focus on \emph{single-source, single-sink networks}, i.e., markets where every buyer wants to purchase a path between the same source and sink node although they may hold different values for the same.

We consider a full information game where sellers can estimate the aggregate demand at any given price. In the first stage of the game, sellers set prices on the edges and in the second stage buyers buy edges along a path.  For any seller $e$, if at a price of $p_e$ per unit amount of the good, a population $x_e$ of buyers purchase the good, then the profit is $p_ex_e - C_e(x_e)$. The buyer's utility is $v_i$ minus the total price paid. A solution is said to be a Nash Equilibrium if (i) Every buyer receives a utility maximizing bundle, i.e, she purchases the cheapest available path as long as its price is at most $v_i$, (ii) No seller can unilaterally change his price and improve his profit at the new allocation (which depends on how many buyers purchase the good at the new prices).

\subsection*{Bertrand Competition with Monopolies}
Our two-stage game is essentially a generalization of the classic model of competition proposed by Bertrand where sellers fix prices and buyers choose quantities to purchase. Bertrand's model has been extensively studied in settings where all the items are substitutes\cite{baye1999folk, dastidar1995existence,guzman2011price}), a common theme being that perfect competition leads to socially optimal outcomes.

In combinatorial markets such as ours, however, inefficiencies arise mainly due to the monopolizing power held by some sellers, i.e., their items are not substitutable. In fact, as we show in Example~\ref{ex:1}, even a monopoly in a single-link market can improve its profit from the Walrasian outcome by raising prices, leading to a loss in social welfare. Our contribution is breaking down the dependence of efficiency on network topology into a single parameter $M$: the number of monopolies simultaneously operating in the market.

Our work is most closely related to the model of Bertrand Competition in Networks with supply limited sellers studied in~\cite{chawla2008bertrand} and later in~\cite{chawla2009price}. Our model is more general as the convex production costs that we consider strictly generalize limited supply. The behavior of Bertrand Networks with seller costs was posed as an open question in~\cite{chawla2008bertrand}. We address this question by extensively applying techniques from the theory of minimum-cost flows. The above paper also considered the efficiency of such markets and showed that in the worst case the equilibrium solution can be arbitrarily worse than the social optimum. We provide a more nuanced understanding of efficiency and characterize the loss in welfare as a function of the number of monopolies ($M$) for a wide spectrum of demand functions. Our main result is that for a large class of reasonable demand functions, the loss in efficiency is at most a factor $(1+M)$ from the optimum solution. We interpret this as a positive result for two reasons,
\begin{enumerate}
\item Given the previous results~\cite{chawla2008bertrand, chawla2009price} that welfare drops exponentially as the number of monopolies increases, a linear loss in welfare indicates that monopolies may not be as `evil' as they are made out to be in many markets.

\item Although a market may consist of a large number of distinct goods, it is reasonable to expect that the number of sellers independently monopolizing the market may be limited. Our bound depends only on such sellers; it is one of very few results interpolating between perfect competition and complete monopoly structure.
\end{enumerate}

\vskip3pt \noindent{\bf The Inverse Demand Function~}
In this work, our primary focus will be on single-source single-sink networks where every buyer has a different value $v_i$, although we do look at more general models in Sections~\ref{sec:generalizations} and~\ref{sec:pos1}. In large markets, it is reasonable to assume that sellers know exactly how many buyers value the $s$-$t$ path at $v$ or more. Formally, we define an inverse demand function $\lambda(x)$ such that for any $v$, $\lambda(x) = v$ implies that exactly $x$ amount of buyers value the path at $v$ or larger. Inverse demand functions are extremely common in Economics literature and provide a direct method to relate buyer demand and welfare in large markets, i.e., total value derived when $t$ buyers receive the good is $\int_{x=0}^{t} \lx dx$.

\subsection{Our contributions}
Our objective in this paper is to characterize the quality of equilibrium in terms of the demand and network structure, and specifically show the effect of monopolies on efficiency. Therefore, all our efficiency bounds depend only on the number of monopolies $M$ which is equivalent to the number of edges present in \emph{all} $s$-$t$ paths. Note that we define efficiency to be the ratio of the optimum social welfare of the market to that at equilibrium. We make no assumptions on the graph structure and sellers' cost functions other than convexity, which is the standard way to model production costs in literature.

\vskip5pt \noindent{\bf Single-Source Single-Sink Networks~}

Our first results concern existence and uniqueness. We show that:
\begin{enumerate}
\item There exists a Nash Equilibrium Pricing in every market where the inverse demand has a monotone `price elasticity'. (MPE functions, see Appendix~\ref{app:inverse_demand} for details). This is a very natural assumption, which is obeyed by most of the demand functions considered in the literature.

\item Our existence proof is constructive: we characterize both the equilibrium prices and the allocation and show how to compute it efficiently.

\item The equilibrium that we construct satisfies several desirable properties including fairness and Pareto-optimality, and it is reasonable to expect this equilibrium to arise in practice. In fact, although there may exist multiple equilibria, ours is the unique equilibrium that is robust or resilient to small perturbations, i.e., it strictly dominates all neighboring solutions.

\end{enumerate}
Therefore, we bound the efficiency of this solution.

\subsubsection*{Efficiency for general classes of demand functions}

\begin{figure}[ht]
\includegraphics[width=\textwidth]{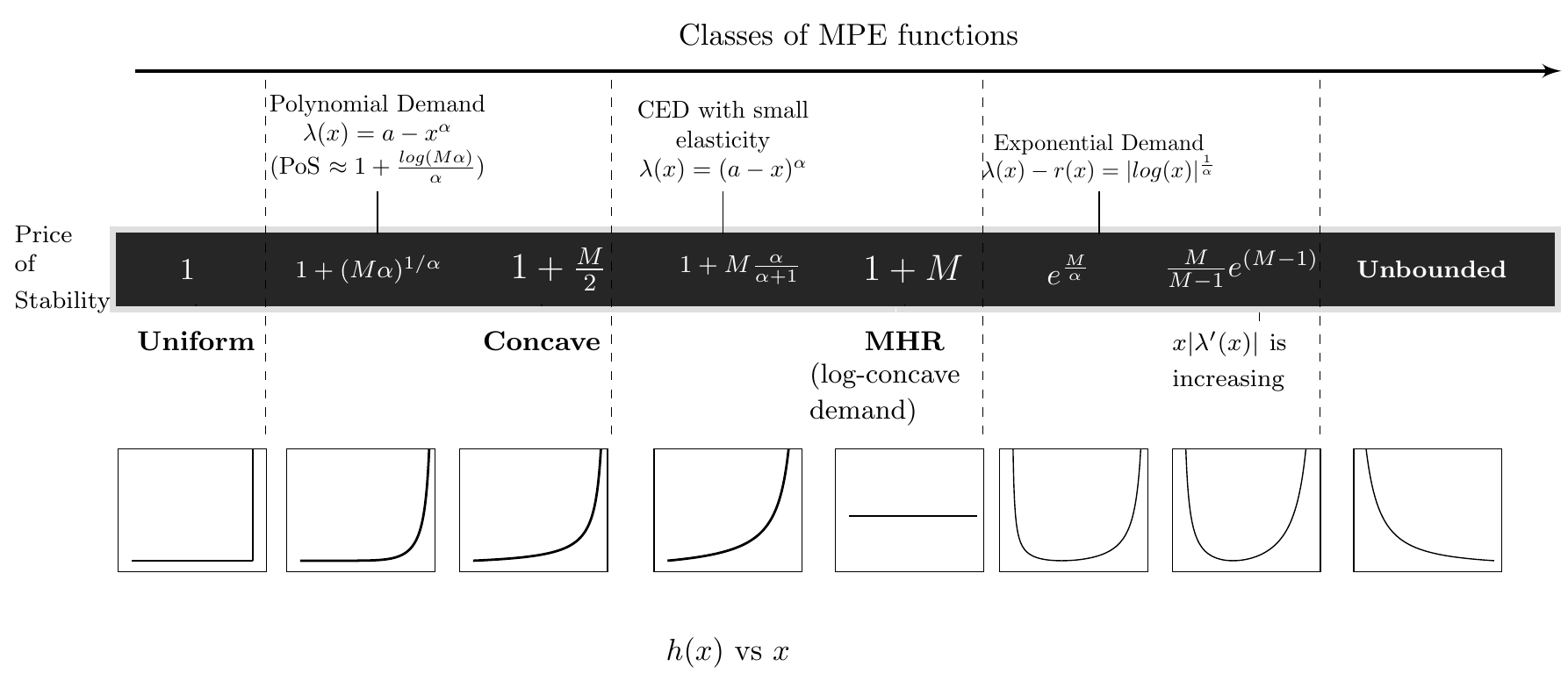}

\caption{Our bounds on the efficiency for different types of inverse demand functions $\lx$. As the hazard rate $h(x)=\frac{|\ldx|}{\lx}$ of the function changes from highly increasing (left) to quickly decreasing (right), the quality of equilibrium solutions changes from perfectly efficient to infinitely worse than the optimum. Equilibrium exists for all functions in the figure, since they all satisfy the MPE property, and it is easy to show that it is perfectly efficient when $M=0$.}
\label{fig:results}
\end{figure}

Our main contribution is showing that for a large class of demand functions, the efficiency drops only linearly as the number of monopolies increases. In particular, we prove the following
\begin{itemize}
\item When all buyers have symmetric valuations (uniform demand), there exists Nash Equilibrium Pricing maximizing social welfare.
\item If the inverse demand function $\lx$ has a \textbf{monotone hazard rate} (MHR), the loss in efficiency at equilibrium is bounded by a factor of $1+M$.
\item When $\lx$ is \textbf{concave} (a subset class of MHR functions), the efficiency loss is $1+\frac{M}{2}$.

\end{itemize}
These efficiency bounds are all tight. Both concave and MHR inverse demand assumptions are quite general and include many popular demand functions (see Section~\ref{app:inverse_demand} for details and examples).

\subsubsection*{Filling the gaps: Efficiency for specific demand functions}
We also come close to characterizing the efficiency for all demand functions obeying the MPE condition. This characterization is summarized in Figure~\ref{fig:results}. This result is also useful because some of the special classes of functions considered here capture buyer demand in some specific applications in the real world. For instance, in highly inelastic markets (eg. electricity market), the efficiency is close to $1+\frac{M}{2}$ and in settings with concave but polynomial demand, the efficiency only drops logarithmically with $M$.

The main conclusion to draw from this is that the presence of monopolies does not completely destroy efficiency: it crucially depends on the properties of the demand curve and the {\em number} of these monopolies. We supplement this result by showing that in several settings without any monopolies, there exist welfare-maximizing Nash Equilibrium. This generalizes Bertrand's classic maxim of `competition causes efficiency' to combinatorial markets. Finally, we also generalize our results to markets that do not have a graph structure, i.e, buyers may desire arbitrary bundles.  All our results extend to this setting with one difference: $M$ is the number of sellers with `monopoly-like' power at equilibrium. This may include sellers who are not monopolies in the traditional sense as their items do not belong to all bundles.

\vskip5pt \noindent{\bf Multiple-Source Networks~}
We provide a first step towards understanding efficiency in multiple-source networked markets by tackling a question of special interest: what conditions cause equilibrium to be fully efficient in such markets? We show that even when buyers desire different paths, as long as the network has a series-parallel structure and no monopolies, equilibrium is efficient. In contrast, without the series-parallel topology, even simple networks with no monopolies may have inefficient equilibria. We also show that the equilibrium is optimal if the buyer demand is fully elastic, or if every buyer has a `last-mile monopoly' (See Section~\ref{sec:pos1}).

\subsection{Related Work}
\noindent{\bf Price Competition with Multiple Sellers.}
As mentioned earlier, our model generalizes previous papers on Bertrand competition in networks, especially~\cite{chawla2008bertrand} and~\cite{chawla2009price}, which showed worst case bounds on efficiency. As our model is more general than theirs, we cannot hope to do better over all instances, however we show that for many important demand functions this inefficiency is bounded. Nevertheless some of our results (specifically Theorem~\ref{thm_mpepos}) are essentially generalizations of results from~\cite{chawla2009price}. Moreover, the behavior in markets with production costs may be quite different from that in supply-limited markets as illustrated in Claim~\ref{thm_specposcap}.

Another closely related model that also considers price competition was studied in~\cite{babaioff2014price}. They show that when there are multiple sellers but a single buyer, equilibrium allocations are efficient. Although we consider large markets, our Theorem~\ref{thm_gen_uniformdemand} is similar in spirit to their result. Our main results, however, are for more complex demand functions, which are not considered in that paper.

The negative results in~\cite{chawla2008bertrand} and~\cite{chawla2009price} have led some researchers to consider more sophisticated pricing mechanisms and other notions of equilibrium: see \cite{kuleshov2012efficiency, correa2008pricing, melo2014price}. For example,~\cite{kuleshov2012efficiency} and ~\cite{correa2008pricing} consider non-linear pricing, where the unit price of a good increases with demand. While complex pricing mechanisms do sometimes lead to improvement in efficiency, it imposes additional complications on the buyers as they now have to anticipate the change in price due to the behavior of other buyers. Thus even though fixed pricing (as we consider) is more natural, its effects are not really well-understood beyond the fact that competition leads to better outcomes and monopolies cause a loss in welfare. Our work captures both these maxims, but attempts to provide additional insight on the structure and quality of equilibrium.

Work such as \cite{acemoglu2007competition, acemoglu2007competitionorig} has also considered two-sided markets where buyers pay the price on each edge, but also incur a cost due to the congestion on the edge; such settings are essentially a combination of the type of market we consider and the classic selfish routing games~\cite{roughgarden2005selfish}. Unfortunately, most of the results in this settings are only known for simple structures such as parallel links or parallel paths. One exception is ~\cite{papadimitriou2010new}, which considers a unique one-sided model where the routing decisions are taken locally by sellers and not buyers as in our paper. They show that in the absence of monopolies, local decisions by sellers can result in efficient solutions.

%[However, the prices that they consider are not directly applicable in market settings where buyers actually purchase links. In their paper, prices are used to propagate information down the chain.]

%\subsection*{Combinatorial Auctions}
\noindent{\bf Combinatorial Auctions with a Single Seller.~} Algorithmic Game Theory has also focused on Nash Equilibria of games derived from market mechanisms with strategic buyers and a single centralized seller~\cite{babaioff2014efficiency, chen2012incentive, hassidim2011non}. The motivation in such papers is complementary to ours, to bound the efficiency loss due to non-price taking buyers. In addition, there has been a surge in the field of envy-free monopoly pricing~\cite{briest2011buying, guruswami2005profit}, two-stage games where all the items are owned by a single seller.

Finally, our two-stage game bears some similarity to first price procurement or path auctions (see \cite{immorlica2005first, moulin2013price} and the references therein). Such auctions are very useful for modeling competition between sellers, but usually ignore the buyer side of the market by assuming that there is a single buyer who wants to purchase exactly one bundle (instead of having some price-dependent demand). Contrary to our setting, path auctions become uninteresting in the presence of monopolies; existing work has mostly focused on concepts like frugality~\cite{immorlica2005first, moulin2013price} and not social welfare.

\section{Definitions and Preliminaries}
\label{sec:prelim}
An instance of our two-stage game is specified by a directed graph $G=(V,E)$, a source and a sink $(s,t)$, an inverse demand function $\lx$ and a cost $C_e(x)$ on each edge. There is a population $T$ of infinitesimal buyers; every buyer wants to purchase edges on some $s$-$t$ path and $x$ amount of buyers hold a value of $\lx$ or more for these paths. A buyer is satisfied if she purchases all the edges on some path connecting $s$ and $t$ and is indifferent among the different paths.

 Equivalently, we could consider a single atomic buyer with demand for $T$ units such that $\Lambda(x) = \int_{x=0}^{y} \lx dx$ is her value for a total of $y$ units of the bundles corresponding to $s$-$t$ paths. We define $M$ to be the number of monopolies in the market: an edge $e$ is a monopoly if removing it disconnects the source and sink. We make the following standard assumptions on the demand and the cost functions.

\begin{enumerate}
\item The inverse-demand function $\lx$ is continuous on $[0,T]$ and non-increasing. The latter assumption simply means that the total demand should not increase if sellers increase their price.

%
%Buyers' utility $\Lambda(x)$ is non-decreasing and concave; equivalently $\lx$ is non-increasing. This is a very standard assumption whether modeling a single buyer or a continuum of buyers: for a single buyer with utility $\Lambda(x)$ this means that the buyer's utility experiences diminishing returns; for a continuum of buyer this means that the demand should not increase if we increase the price.

\item $C_e(x)$ is non-decreasing and convex $\forall e$, a standard assumption for production or congestion costs. Moreover, $C_e(x)$ is continuous, twice differentiable, and its derivative $c_e(x)=\frac{d}{dx}C_e(x)$ satisfies $c_e(0)=0$. %The final assumption means that it costs the seller very little to produce the first quantity of his item.
We relax these assumptions in Section~\ref{sec:generalizations}. %, with our results staying essentially the same.
\end{enumerate}

\vskip 3pt \noindent{\bf Nash Equilibrium Pricing.~}
A solution of our two-stage game is a vector of prices on each item $\vec{p}$ and an allocation or flow $\vec{x}$ of the amount of each $s$-$t$ path purchased, representing the strategies of the sellers and buyers respectively. The total flow or market demand is equal to the number of buyers with non-zero allocation $x = \sum_{P \in \mathbb{P}}x_P$, where $\mathbb{P}$ is the set of $s$-$t$ paths. We can also decompose this flow $\vec{x}$ into the amount of each edge purchased by the buyers ($x_e$). Given such a solution, the total utility of the sellers is $\sum_{e \in E}(p_e x_e - C_e(x_e))$ and the aggregate utility of the buyers is $\int_{t=0}^x \lambda(t) dt - \sum_{e \in E}p_e x_e$. The total social welfare is simply $\int_{t=0}^x \lambda(t)dt - \sum_e C_e(x_e)$, i.e., prices are intrinsic to the system and do not appear in the welfare.

We now formally define the equilibrium states of our game. An allocation $\vec{x}$ is said to be a \textbf{best-response} by the buyers to prices $\vec{p}$ if buyers only buy the cheapest paths and for any cheapest path $P$, $ \lambda(x) = \sum_{e \in P}p_e$. That is, buyers act as price-takers and any buyer whose value is at least the price of the cheapest path will purchase some such path. A solution $(\vec{p},\vec{x})$ is a Nash Equilibrium if $\vec{x}$ is a best-response allocation to the prices and, $\forall e$ if the seller unilaterally changes his price from $p_e$ to $p'_e$, then for {\em every} feasible best-response flow $(x'_e)$ for the new prices, seller $e$'s profit cannot increase, i.e., $p_ex_e - C_e(x_e) \geq p'_ex'_e - C_e(x'_e)$. Our notion of equilibrium is quite strong as the seller does not have to predict exactly the resulting flow: for every best-response by the buyers, the seller's profit should not increase.

\vspace{-1mm}
\subsubsection*{Classes of inverse demand functions that we are interested in}

We now define some classes of demand functions that we consider in this paper. For the sake of notational convenience, we assume that the inverse demand function is continuously differentiable, and thus $\ldx$ is well-defined (and not positive as $\lx$ is non-increasing). However, all our results hold exactly even without this assumption (see Section~\ref{sec:generalizations}). The reader is asked to refer to Appendix~\ref{app:inverse_demand} for additional discussion and interpretation of each class of functions and other related economic concepts.

\begin{description}
\setlength{\itemsep}{-1pt}
\item[\textbf{Uniform buyers}:] $\lx=\lambda_0 > 0$ for $x \leq T$. In other words, a population of $T$ buyers all have the same value $\lambda_0$ for the bundles.

\item[\textbf{Concave Demand}:] $\ldx$ is a non-increasing function of $x$. This includes the popular linear inverse demand case ($\lx = a-x$)~\cite{abolhassani2014network, tsitsiklis2012efficiency} where the demand drops linearly as price increases.

\item[\textbf{Monotone Hazard Rate} (MHR) Demand:] $\frac{\ldx}{\lx}$ is non-increasing or $h(x)=\frac{|\ldx|}{\lx}$ is non-decreasing in $x$. This is equivalent to the class of {\em log-concave} functions~\cite{amir1996cournot} where $\log(\lx)$ is concave, and essentially captures inverse demand functions without a heavy tail. Example function: $\lx = e^{-x}$.

\item[\textbf{Monotone Price Elasticity} (MPE):] $xh(x) = \frac{x |\ldx|}{\lx}$ is a non-decreasing function of $x$ which tends to zero as $x \to 0$. This is equivalent to functions where the price elasticity of demand is non-decreasing as the price increases. Price elasticity measures the responsiveness of the market demand over its sensitivity to price, and exactly equals $\lx/x\ldx$ using our notation.
\end{description}

Each class of demand function defined above strictly contains all the classes defined previously, i.e., Uniform Demand $\subset$ Concave $\subset$ MHR $\subset$ MPE.

\vskip 3pt\noindent\textbf{Min-Cost Flows and the Social Optimum:} Since an allocation vector on a graph is equivalent to a $s$-$t$ flow, we briefly dwell upon minimum cost flows. Formally, we define $R(x)$ to be the cost $\sum_e C_e(x_e)$ of the min-cost flow of magnitude $x\geq 0$ and $r(x)$, its derivative, i.e., $r(x) = \frac{d}{dx}R(x)$. Both the flow and its cost can be computed via a simple convex program given a graph and cost functions. Clearly, $R(x)$ is non-decreasing since increasing the amount of  flow can only lead to an increase in cost, as the production costs $C_e$ are non-decreasing. We prove in the Appendix that:

\begin{proposition}\label{prop:Rstuff}
\label{prop_mincostfunction}
$R(x)$ is continuous, differentiable, and convex for all $x\geq 0$.
\end{proposition}

It is also easy to see from the KKT conditions that for a min-cost flow $\vec{x}$, we have $r(x)=\sum_{e \in P}c_e(x_e)$ for any path $P$ with non-zero flow (for a full proof see the Appendix). Given an instance of our game, the optimum solution is an allocation or flow which maximizes the social welfare $\Lambda(x)-\sum_e C_e(x_e)$. Since the buyers' utility depends only on the magnitude of the flow, welfare is maximized
when the flow is of minimal cost. The optimum solution therefore maximizes $\Lambda(x)-R(x)$ and must satisfy the following condition:

\begin{proposition}
\label{prop_optflow}
The solution maximizing social welfare is a min-cost flow of magnitude $x^*$ satisfying $\lambda(x^*) \geq r(x^*)$. Moreover, $\lambda(x^*)=r(x^*)$ unless $x^* = T$. \end{proposition}

\section{Existence, Uniqueness, and Computation of Equilibrium Prices}
\label{sec:existenceBody}
In this section, we show that for a large class of demand functions, we are always guaranteed the existence of Nash Equilibrium. Although there may be multiple equilibria in general, we show that some of these are highly unrealistic. The equilibrium that we consider, on the other hand, satisfies several desirable properties and is the unique solution that is resilient to small perturbations. Finally, we show how to efficiently compute this unique equilibrium.

We now show our first result that reinforces the fact that even in arbitrarily large networks (not necessarily parallel links), competition results in efficiency. This result is only a starting point for us since it is the addition of monopolies that leads to interesting behavior. We only sketch our proofs here, full proofs are located in the Appendix.

\begin{claim}
\label{clm_monopolypos1}
In any network with no monopolies (i.e, you cannot disconnect $s$, $t$ by removing any one edge), there exists a Nash Equilibrium maximizing social welfare.
\end{claim}
\emph{(Proof Sketch)} Consider the optimum solution $\vec{x^*}$ and price every edge at $p_e = c_e(x^*_e)$. Notice that if any edge increases its price, it stands to lose all of its flow since buyers will have other cheaper paths. No edge can improve profits by pricing below marginal cost because $p_e \Delta x \leq C_e(x^*_e+\Delta x) - C_e(x^*_e) \leq c_e(x^*_e+\Delta x)\Delta x$. $\qed$

We remark that our notion of a ``no monopoly" graph is weaker than what has been considered in some other papers~\cite{correa2008pricing, papadimitriou2010new} and therefore, our result is stronger. We are now in a position to prove our main existence result. Unlike~\cite{chawla2008bertrand}, our results do not depend on how much the buyers value the goods (no assumption on how large $\lambda(0)$ can be). The result is also constructive, we are able to characterize the equilibrium prices and properties of the demand function at equilibrium.

\begin{theorem}\label{thm:existence} For any MPE demand function $\lambda$, there exists a Nash equilibrium.
\end{theorem}
\emph{(Proof Sketch)} We define a simple pricing rule for a min-cost flow $\vec{x}$ of any given magnitude $x$ and an instance with $M$ monopolies. The rule is fair from the monopolies' perspective: first every edge is priced at its marginal and then the remaining slack is divided among the monopolies. The rest of the proof lies in showing that $\exists$ $\tilde{x} > 0$ that obeys the conditions of Corollary~\ref{corr_eqconditions} and that no seller can increase or decrease his price to improve profits when the following prices are used.

\[
 p_e(\vec{x}) =
  \begin{cases}
      \hfill \frac{\lx - r(x)}{M} + c_e(x) \hfill & \text{if e is a Monopoly} \\
      \hfill c_e(x_e) \hfill & \text{otherwise} \\
  \end{cases} \qed
\]

\begin{corollary}
\label{corr_eqconditions}
For any inverse demand function $\lambda$ belonging to the class MPE, we are guaranteed the existence of a Nash Equilibrium with a min-cost flow $(\tilde{x}_e)$ of size $\tilde{x}\leq x^*$ such that,
\begin{enumerate}
\item The prices obey the pricing rule above.
\item Either $\frac{\lambda(\tilde{x})-r(\tilde{x})}{M} = \tilde{x}|\lambda'(\tilde{x})|$ or $\tilde{x} = x^*$, the optimum solution.
\end{enumerate}
\end{corollary}

\begin{figure}
\centering
\subfigure[Serial Network]{\label{figure:serial} \includegraphics[scale=1.1]{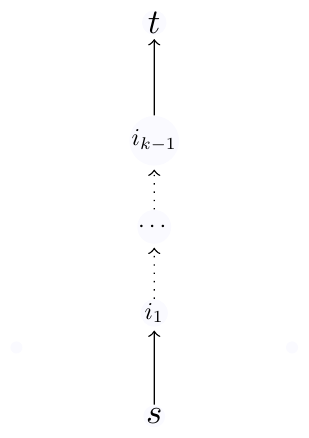}}
\subfigure[Parallel Network]{\label{figure:parallel} \includegraphics[width=0.3\linewidth]{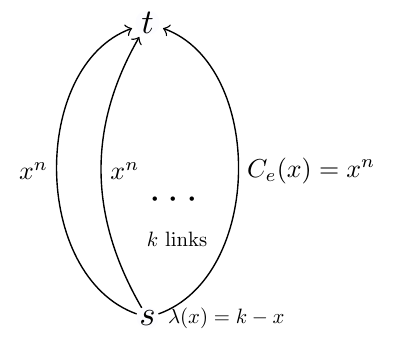}}
\caption{~}
\end{figure}

Before analyzing the efficiency of equilibrium, it is important to address the question of what equilibria are likely to be formed. In markets such as ours, it does not make sense to provide a blanket bound on all stable solutions since some of these are highly unrealistic. In contrast, the solution obeying Corollary~\ref{corr_eqconditions} is the unique solution that satisfies useful desiderata, motivating us to study its efficiency in the following sections. We begin by illustrating some of the ``unreasonable" equilibria that exist even in the simplest of markets.

\begin{itemize}

\item \textbf{Trivial Equilibrium}~\cite{chawla2009price}: In any instance where all paths have length at least $2$, we can easily show the existence of a Nash Equilibrium with zero flow. Suppose that in the example of Figure~\ref{figure:serial}, every seller sets a really high price (say larger than $\lambda(0)$). No buyer can afford the bundle, moreover, no seller can unilaterally lower his price and induce any flow.

\item \textbf{Non-Competitive Equilibrium:} Consider the parallel network shown in Figure~\ref{figure:parallel}. There exists an equilibrium  that has non-zero flow but closely mirrors trivial equilibria. Assume for concreteness that $k=100$ and $n = 102$. Consider the solution where each edge is priced at (say) $p = 98$ and receives $\frac{2}{100}$ units of flow. No edge can lower its price because then it would receive the entire 2 units of flow, which is too expensive to transit. This solution is Pareto-dominated by the optimum equilibrium in Claim~\ref{clm_monopolypos1} since every seller's profit is larger in the latter. As $k \to \infty$, this violates the well-established idea that perfect competition causes efficiency~\cite{harrington1989re, chowdhury2004coalition}.
\end{itemize}
Motivated by these, we now define some natural properties that one desires from market equilibrium.  \textbf{All of these are obeyed by the equilibrium from Corollary~\ref{corr_eqconditions}.} We formally prove this fact in the Appendix.

\begin{enumerate}

\item \textbf{(Non-Trivial Pricing)}: Every edge that does not admit flow must be priced at $0$, or more generally $c_e(0)$.

\item \textbf{(Recovery of Production Costs)}: Given an equilibrium $(\vec{p}, \vec{x})$, every item's price is at least $c_e(x_e)$.

This property induces a basic fairness criteria for the seller: the seller's price for any item is at least the cost of producing it.

\item \textbf{(Pareto-Optimality)}: A Pareto-optimal solution over the space of equilibria is an equilibrium solution such that for any other equilibrium, at least one agent prefers the former solution to the latter. Pareto-Optimality is often an important criterion in games with multiple equilibria; research suggests that in Bertrand Markets, Pareto optimal equilibria are the solutions that arise in practice~\cite{cabon2010end}.

\item \textbf{(Local Dominance)}: Given an equilibrium $(\vec{p}, \vec{x})$, no seller's profit should strictly increase when a small fraction of buyers shift their flow from one path to another. The essence of this property is that the solution is resilient against small buyer perturbations. By definition of an equilibrium, no seller can change his price. However, at the same price, a seller may be able to attract a small fraction of buyers towards his item. If the resulting solution is strictly preferred by the seller, then it indicates that the original equilibrium is not robust. Local dominance provides a strong guarantee: that an equilibrium dominates all neighboring solution along every single component.
\end{enumerate}

We now show that for any given instance with strictly monotone demand and non-zero production costs, either all equilibria are optimal or the equilibrium in Corollary~\ref{corr_eqconditions} is the unique non-trivial equilibrium with that satisfies local dominance. From the perspective of efficiency, this means that we only have to bound the welfare of a single equilibrium solution. In cases where this is not true, all equilibrium solutions are optimal. In addition to uniqueness, we also show that the above solution can be computed efficiently by means of a simple binary search followed by a single min-cost flow computation. As always in the case of real-valued settings (e.g., for convex programming, other equilibrium concepts, etc), ``computing" a solution means getting within arbitrary precision of the desired solution; the exact solution may not be computable efficiently or even be irrational.

\begin{theorem}
\label{thm:uniqueness}
For any given instance with strictly monotone MPE demand and non-zero production costs, we are guaranteed that {\em one of the following} is always true,
\begin{enumerate}
\item There is a unique equilibrium that obeys Local Dominance and Non-Trivial Pricing, {\em or}
\item All equilibria that satisfy Local Dominance and Non-Trivial Pricing maximize welfare.
\end{enumerate}
\end{theorem}

As is the case with our uniqueness result, in most systems, a certain degree of strong monotonicity is essential to show uniqueness. We now focus on understanding the quality of these equilibrium solutions, in comparison with the social optimum.

\section{Effects of Demand Curves and Monopolies on Efficiency of Nash Equilibrium}
\label{sec:efficiency}
In this paper, we are interested in settings where approximately efficient outcomes are reached despite the presence of self-interested sellers with monopolizing power over the market. While for general functions $\lx$ belonging to the class MPE, the efficiency can be arbitrarily bad, we show that for many natural classes of functions it is relatively small, even in the presence of a limited number of monopolies. The proofs can be found in Appendix \ref{app:pos}.

We begin with the special case of buyers with uniform demand, i.e., every buyer in the continuum values the $s$-$t$ paths at the same value $\lambda_0$. In this case, the inverse demand function becomes $\lambda(x) = \lambda_0$ for $x \leq T$ and $0$ otherwise. We show that for buyers with uniform demand, $\exists$ an equilibrium maximizing social welfare. It is however, interesting to note that this solution may not necessarily coincide with the Walrasian Equilibrium, an example of the same is provided in the Appendix.

\begin{theorem}
\label{theorem_linearutilitymain}
Every instance with uniform demand buyers admits an efficient Nash Equilibrium.
\end{theorem}
\emph{(Proof Sketch)} The essence of the proof lies in computing the socially optimal flow, pricing every edge at the marginal and dividing the slack evenly on the monopolies. This way the total price is $\lambda_0$ and no edge can increase its price and as long as $x = T$, reducing the price serves no purpose. When $x < T$, from Proposition~\ref{prop_optflow}, there is no slack so every edge is priced at its marginal cost.$\qed$

\subsection*{Main Result: MHR Demand and Concave Demand}

We now show our main result that the efficiency drop is a factor $1+M$ when the inverse demand function has a monotone hazard rate. As mentioned previously, these correspond exactly to the case when the inverse demand function is log-concave, i.e., $\log(\lambda(x))$ is concave. Log-concavity is a very natural assumption on the inverse demand function and it is not surprising to see that such functions  have received considerable attention in Economics literature~\cite{amir1996cournot, baldenius2000comparative,lyoo2006efficient}. Finally, note that $\lambda(x)$ being concave is a special case of log-concavity. For this class of functions, we obtain an improved bound of $1+\frac{M}{2}$ on the efficiency. The extremely popular linear inverse demand function $\lambda(x) = 1-x$ falls under this class.

\begin{theorem}\label{thm:PoSbody}
The social welfare of Nash equilibrium from Section \ref{sec:existenceBody} is always within a factor of:
\begin{itemize}
\item $1+\frac{M}{2}$ of the optimum for concave $\lambda$.
\item $1+M$ of the optimum for Log-concave (i.e., MHR) $\lambda$.
\end{itemize}
Both these bounds are tight, i.e., there exist instances where the optimum solution has a welfare that is exactly $1+M$ times the Nash Equilibrium for MHR demand and $1+\frac{M}{2}$ times that of the equilibrium for concave demand.
\end{theorem}
\emph{(Proof Sketch for the $1+M$ bound on MHR functions)} \\
The proof relies crucially on our characterization of equilibria. Showing that $\exists$ an equilibrium satisfying several nice properties and Corollary~\ref{corr_eqconditions} was the hard part. Armed with these results, the rest of the efficiency proof becomes extremely intuitive. For MHR functions, the proof hinges on an interesting claim (that may be of independent interest) linking the welfare loss at equilibrium to the profit made by all sellers. Specifically, our key claim is that `the loss in welfare is at most a factor $M$ times the total profit in the market at equilibrium'. In addition, it is also not hard to see that in any market, the profit cannot exceed the total social welfare of a solution.

Why is this property useful? Using profit as an intermediary, we can now compare the welfare lost at equilibrium to the welfare retained. This implies that the welfare loss cannot be too high because that would mean that the profit and hence the welfare retained is also high. But then, the sum of welfare lost $+$ retained is the optimum welfare and is bounded. Therefore, once we bound the welfare lost at equilibrium, we can immediately bound the overall efficiency. Our key claim is,

$$\text{Lost Welfare} = \int_{\tilde{x}}^{x^*}\lx dx - [R(x^*) - R(\tilde{x})]  \leq M (p\tilde{x} - R(\tilde{x})),$$
where $p$ is the payment made by every buyer, $\tilde{x}$ is the amount of buyers in the equilibrium solution and $x^*$, in the optimum. The integral in the LHS can be rewritten as $\int_{\tilde{x}}^{x^*}(\lx - r(x)) dx$. Now, we apply some fundamental properties of MHR functions ($\lambda$) and show in the appendix that for all $x \geq \tilde{x}$, the following is true,
$$\frac{\lx - r(x)}{|\ldx|} \leq \frac{\lambda(\tilde{x}) - r(\tilde{x})}{|\lambda'(\tilde{x})|} = M\tilde{x}.$$
The final equality comes from our equilibrium characterization in Corollary~\ref{corr_eqconditions}. Therefore, we have,
\begin{align*}
\int_{\tilde{x}}^{x^*}(\lx - r(x)) dx & \leq  M\tilde{x} \int_{\tilde{x}}^{x^*}|\ldx|dx & \\
& \leq  M\tilde{x}(\lambda(\tilde{x}) - \lambda(x^*)) & \text{($\lx$ is non-increasing and $\tilde{x} \leq x^*$)}\\
& \leq  M(\lambda(\tilde{x})\tilde{x} - R(\tilde{x})) & \text{$(\lambda(x^*)\tilde{x} \geq r(x^*)\tilde{x} \geq r(\tilde{x})\tilde{x} \geq R(\tilde{x}))$}
\end{align*}
The total payment $p$ on any path must exactly equal $\lambda(\tilde{x})$ (See Equilibrium Def.). $\blacksquare$

We reiterate here that both the results in the theorem make no assumption on the graph structure or cost functions other than the ones mentioned in Section~\ref{sec:prelim}. It is reasonable to assume that even in multi-item markets where consumers may desire large bundles, the number of purely monopolizing goods is not too large: in such cases the equilibrium quality is high. 

\subsection*{Filling the gaps: Other classes of Demand Functions}

The above classes of inverse demand functions encapsulate many different types of functions, but there is a large gap between an efficiency factor of 1 and that of $1+\frac{M}{2}$, for example. We attempt to interpolate between our bounds above by considering two very specific (but important and commonly studied) classes of functions defined below, and obtain the following results (see Appendix \ref{appsection:specificpos} for proofs).

\begin{theorem}
\label{thm_spec2}
\begin{enumerate}
\item Let $F_p$ denote functions of the form $\lx=\lambda_0(1-x^{\alpha})$ for any $\lambda_0\geq 0$ and $\alpha \geq 1$. For $\lambda\in F_p$, the efficiency is at most $(1+M\alpha)^{\frac{1}{\alpha}}$. When $\alpha \geq M$, this quantity is approximately $1+\frac{\log(M\alpha)}{\alpha}$.

\item Let $F_{ced}$ denote functions of the form  $\lx = \lambda_0(1-x)^{\alpha}$ for any $\lambda_0\geq 0$ and $\alpha \geq 1$. For $\lambda \in F_{ced}$, the efficiency is at most $1+M\frac{\alpha}{\alpha+1}$ for $\alpha \geq 1$.

\end{enumerate}
\end{theorem}

\subsubsection*{Discussion}
\begin{enumerate}

\item Polynomially decreasing inverse demand functions such as the ones in $F_p$ are extremely common in papers that assume concave demand, e.g, in \cite{bulow1983note}; the special case when $\alpha=1$ is popularly referred to as ``linear inverse demand functions" in the literature~\cite{abolhassani2014network}. It is interesting to note that the efficiency drops only logarithmically with $M$ for these functions as opposed to linearly in the general case.

\item {\em CED} stands for {\em constant elastic demand} (see for Appendix~\ref{app:inverse_demand}), where the elasticity of demand is $\frac{1}{\alpha}-1$ and denotes how much the demand changes with the changes in price. The main application of our result is showing that when the demand is relatively inelastic, the efficiency can still be much better than $1+M$. It is well-known that several essential commodities have such a relatively inelastic demand: only a large increase in price can lead to a lot of consumers dropping out. For instance, the elasticity of the market for electricity is believed to be somewhere around $-0.3$~\cite{filippini1999swiss}; our results would guarantee a Nash equilibrium within a factor of $\approx 1+0.59M$ of optimum for such functions.
\end{enumerate}
 By varying $\alpha$ in all the functions above, we can interpolate between our results in Theorem \ref{thm:PoSbody}, and quantify how the quality of equilibrium changes as the functions become less concave. For the sake of completeness, we also consider functions that do not have a monotone hazard rate but still fall under the class MPE. Perhaps most important among these seem to be the class of MPE functions where the quantity $x|\ldx|$ is non-decreasing. This class of functions was considered in~\cite{chawla2009price}, where it was shown that the efficiency loss can be as bad as $e^M$. We generalize their results to markets with cost functions, and further are able to slightly improve upon the bound in that paper.

\begin{claim}
\label{thm_mpepos}
If $\lx$ has $x|\ldx|$ non-decreasing, then the loss in efficiency for any instance is at most $\frac{M}{M-1}e^M$.
\end{claim}
Finally, we are also able to characterize the efficiency spectrum between linear and exponential. Our next result shows that for an interesting class of functions with logarithmic inverse demand, the efficiency is $e^{\frac{M}{\alpha}}$. Such demand functions are reasonably popular (see~\cite{bulow1983note, tsitsiklis2012efficiency}) and are considered to represented actual buyer demand in the transportation industry~\cite{evans1987theoretical}.

\begin{claim}
\label{clm_mpespecem}
Let $F_{exp}$ denote the demand and cost functions which can be represented as $\lx - r(x) = |\ln(\frac{x}{a})|^{\frac{1}{\alpha}}$ for $x \leq a$, and $\alpha\geq 1$. Then the efficiency is at most $e^{\frac{M}{\alpha}}$ for $\alpha \geq 1$.
\end{claim}

\section{Looking beyond Graphical Markets: More General Models}
\label{sec:generalizations}
All the results above hold for the somewhat limited case in which the market structure is that of a graph. However, what if all the buyers still have an identical set $B$ of valid bundles (with each buyer possibly having different valuations), but this set $B$ did not correspond to the set of $s$-$t$ paths in a graph? This case can become a lot more intractable, since there can be sellers $e$ which are not true monopolies (they do not belong to all bundles in $B$), but still hold much more power than other sellers. Similar to~\cite{chawla2008bertrand}, this can be formalized through the notion of {\em virtual monopolies}: a seller $e$ is a virtual monopoly for an allocation $x$ if $e$ is part of every bundle which is allocated a non-zero amount in $x$.

Formally, consider a market with a set $E$ of goods. Every buyer $i$ desires one bundle from a subset $B \subseteq 2^E$ of valid and is indifferent as to which bundle she gets. In addition, we also drop the requirement that $c_e(0) =0$ for every edge. In a sense, this assumption helped maintain competition. Consider a market with two parallel links  $e$ and $e'$ but $c_e(0)$ is much larger than $c_{e'}(0)$. It is clear that if the demand is not too high, then $e'$ has the power of a monopoly here as the entry cost for $e$ is just too high. Now, we can generalize our results as follows:

\begin{theorem}
For buyers with identical sets of valid bundles $B_i$, all the results from Section~\ref{sec:existenceBody} and~\ref{sec:efficiency} hold, but with $M$ being the number of {\em virtual monopolies} of the equilibrium solution $\tilde{x}$, instead of the number of monopolies.
\end{theorem}
The proof is non-trivial and does not follow from our previous techniques in Section~\ref{sec:existenceBody}. In fact we consider a new ascending price algorithm that generalizes the pricing rule mentioned in the proof of Theorem~\ref{thm:existence}. The algorithm for setting prices at equilibrium may be of independent interest.

\textbf{Discussion.} While in the worst case, the number of virtual monopolies can be as large as the number of sellers, for markets with reasonable (but not perfect) competition it is likely to be a lot less.

\subsection*{Uniform Demand Buyers with General Combinatorial Valuations}
So far the buyer valuation functions have all had a certain structure: the buyer has a set of desired bundles and cares only about these bundles. What if instead all buyers had a single monotone valuation function $v(S)$ that captures how much they value the set of items $S \subseteq E$. We generalize both our Theorem~\ref{theorem_linearutilitymain} and a result in~\cite{babaioff2014price} and show that for uniform demand buyers with arbitrary combinatorial valuations, there exists an efficient Nash Equilibrium.

Formally, consider a monotone valuation function $v(S)$ defined for every $S \subseteq E$. Every buyer in the continuous population of $T$ receives a value of $v(S)$ from set $S$. Then we show the following,
\begin{theorem}
\label{thm_gen_uniformdemand}
In any setting where uniform buyers have arbitrary combinatorial valuations, there exists an efficient Nash Equilibrium.
\end{theorem}

\section{Multiple-Source Networks and Efficient Equilibrium}
\label{sec:pos1}

\subsection{Single-Source Networks: When does competition lead to complete efficiency?}

We conclude our findings for single-source single-sink networks by showing that when the demand is `somewhat elastic' and sellers are limited in supply, there exists an efficient Nash Equilibrium. This is an interesting special case for capacitated networks not considered in~\cite{chawla2008bertrand} and~\cite{chawla2009price}. Surprisingly, for the exact same demand but with production costs, we show in the Appendix (Claim~\ref{clm_unboundedpos}) that the equilibrium can be arbitrarily inefficient. Therefore, this result reiterates our claim that the behavior in networks with general production costs can be quite different from the special case with capacities.

Formally, consider a networked market where every edge $e$ has a capacity of $c_e$. This is a special convex cost function where $C_e(x) = 0$ if $x \leq c_e$ and $C_e(x) = \infty$ otherwise. The `somewhat elastic' demand functions that we consider have the form $\lx = a x^{-\frac{1}{r}}$ where $a \geq 0$ and $r \geq 1$. The exact value of $r$ depends on $M$.
\begin{claim}
\label{thm_specposcap}
For any single-source single-sink network with edge capacities, $M$ monopolies, and demand function $\lx = ax^{1/r}$, there exists an efficient Nash Equilibrium as long as $r > M$.
\end{claim}

\subsection{Multiple-Source Networks}
We now move on to more general networks where different buyers have different $s$-$t$ paths that they wish to connect and the demand function can be different for different sources. Unfortunately, our intuition from the previous sections does not carry over. As we show in Appendix \ref{app:multsourcesink}, Nash equilibrium may not exist even for instances with two sources and a single sink. Perhaps more surprisingly, we give relatively simple examples in which perfect efficiency is no longer achieved, even with complete absence of monopolies! Nevertheless, we prove that for some interesting special cases, fully efficient Nash equilibrium still exists even when buyers desire different sets of bundles. In particular, we believe that our result on series-parallel networks is an important starting point for
truly understanding multiple-source networks.

We begin by showing that even in a simple network with two sources, a common sink and no monopolies, there is no equilibrium maximizing social welfare.
\begin{claim}
\label{ex:nomonopolyeq_main}
There exists an instance with two sources, one sink, and no monopoly edges for either source, where all equilibria are inefficient.
\end{claim}
\emph{(Proof Sketch)}\\
\begin{wrapfigure}{l}{4cm}
\centering
\includegraphics[scale=0.8]{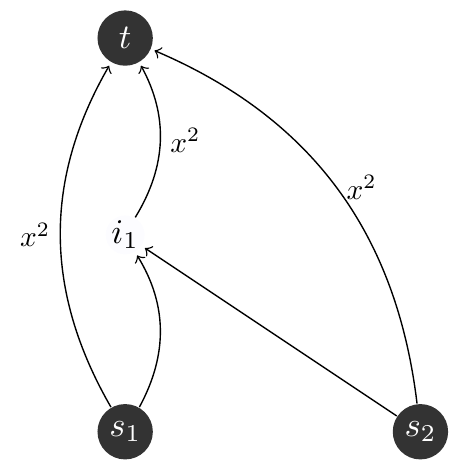}
\caption{Instance with two sources $s_1$, $s_2$ and one sink $t$.}
\label{fig_monopolyposgeq1_main}
\end{wrapfigure}

Consider the instance shown in the figure, where the cost function on each edge is given as the edge weight. Suppose that $s_1$ has demand $\lx = 1-x$ and $s_2$ has $\lx = 4-x$. At the unique optimum point $s_1$ sends a flow of $1/3$ on its direct link (edge $e_1$) to the sink and $s_2$ sends $2$ units of flow, one on each of its paths.

Assume by contradiction that $\exists$ prices stabilizing the optimum flow. For $s_1$, $\lambda(1/3) = 2/3$ and so this must be the price on $e_1$. On the other hand, $e_2 = (i_1, t)$ is carrying $1$ unit of flow at a cost of $1$ and therefore, for non-zero profit, its price must be at least one. But this means that $e_1$ is virtually a monopoly for $s_1$. Therefore, if $e_1$ increases its price to $p'_1=\frac{3}{4}$, its profit will also strictly improve. This contradicts the assumption that the original solution is stable. $\qed$

Despite this discouraging result, we show that there are important special cases where equilibrium allocations do maximize social welfare. We first consider the case when all the buyers desiring to connect a given source-sink pair have uniform demand, i.e., we consider a multiple-source, multiple-sink network where every $(s_t,t_i)$ pair has a buyer demand $\lambda_i(x) = \lambda^0_i$ for $x \leq T_i$ and zero otherwise. For this, we show conditions on both the demand and network topology that lead to efficient equilibria.

\begin{claim}
\label{clm_multsourceundemand}
There exists efficient equilibrium in multiple-source multiple-sink networks with uniform demand buyers at each source if one of the following is true
\begin{enumerate}

\item Buyers have a large demand and production costs are strictly convex.

\item Every source node is a leaf in the network.
\end{enumerate} \end{claim}
The second case commonly arises in several real-world networks including Internet, postal, and transportation networks~\cite{schintler2005complex}. The commonality between these physical networks is that there exists a well-connected, easily-accessible central network that links major hubs, but the last-mile between the hub and the final user is often controlled by a local monopoly and thus the source is a leaf. It is not feasible for firms to compete at the last-mile due to the heavy infrastructure costs and minimal returns.

\vskip 3pt \textbf{Multiple-Source Single-Sink Series-Parallel Networks}
In some sense, Claim~\ref{clm_monopolypos1} embodies the very essence of the Bertrand Paradox, the fact that competition drives prices to their marginal cost. So it is perhaps surprising that this does not hold in general networks as illustrated by Claim~\ref{ex:nomonopolyeq_main}. However, we now show that for a large class of markets which have the series-parallel structure, the absence of monopolies still gives us efficient equilibria. Series-Parallel networks have been commonly used~\cite{correa2008pricing,kuleshov2012efficiency} to model the substitute and complementary relationship that exists between various products in arbitrary combinatorial markets. We conjecture that this result is tight in the sense that no larger class of graphs without monopolies have this property.

We define a multiple-source single-sink graph to be a series-parallel graph if the super graph of the given network obtained by adding a super-source and connecting it to all the sources has the series-parallel structure. The notion of ``no-monopolies" for a complex network has the same idea as a single-source network: there is no edge in the graph such that its removal would disconnect any source from the sink. We are now in a position to show our main result extending the Bertrand Paradox to arbitrary series-parallel graphs.

\begin{theorem}
\label{thm_seriesparallelpos}
A multiple-source single-sink series-parallel network with no monopolies admits a Nash Equilibrium that has the same social welfare as the optimum solution for any given instance where all edges have $c_e(0)=0$.
\end{theorem}
%\emph{(Proof Sketch)} Once again, we take the optimum flow and price every edge at its marginal
%cost. We show that S-P graphs have an interesting property: for every edge $e$ and any source which has a path through $e$, the source must also have a path to $t$ not containing $e$. As seen in Figure~\ref{fig_monopolyposgeq1_main}, this is not generally true for non S-P graphs. Therefore, if $e$ increases its price, then it stands to lose all of its flow. $\qed$

\section{Conclusions and Future Work}
In this work, we considered large, decentralized markets with production costs, where Walrasian Equilibrium may not be an appropriate solution concept. Our main focus is on single-source single-sink network markets where our results provide significant understanding of how monopolies affect the efficiency of the market equilibrium. In particular, our main contribution is that as long as the inverse demand obeys a natural condition (monotone hazard rate), the efficiency is at most $1+M$. This result relies crucially on the characterization of the equilibrium prices and allocation that we provide. Our results also hint that studying the inverse demand as opposed to direct demand provides a better understanding of welfare properties. In particular Figure~\ref{fig:results} indicates that efficiency seems to depend on the hazard rate of the inverse demand. A deeper understanding of welfare as a function of a single parameter that captures most demand functions is an obvious future step.

We also generalize our model in two directions: non-networked markets where buyers desire arbitrary bundles and multiple-source markets where different buyers desire different bundles. For the latter case, our results provide a modest but important first step in understanding the factors controlling efficiency in more complex networks. The study of general multiple-source networks with some structure (trees, series-parallel networks) that builds on existing work appears to be the most promising future direction.

%\appendixhead{ANSHELEVICH}

% Acknowledgments
% Bibliography
\bibliographystyle{plain}
\bibliography{bibliographyec}
                             % Sample .bib file with references that match those in
                             % the 'Specifications Document (V1.5)' as well containing
                             % 'legacy' bibs and bibs with 'alternate codings'.
                             % Gerry Murray - March 2012

% History dates

% Electronic Appendix
\appendix

\section{The Inverse Demand Function}
\label{app:inverse_demand}
Our main aim in this paper is to understand the efficiency of equilibrium allocations for different types of demand.
In order to characterize the structure of demand in the market, we used an inverse demand function $\lx$. Intuitively, the inverse demand function gives us a value $\lx=v$ such that exactly $x$ of the non-atomic buyers hold a value of $v$ or more for the desired bundles. It is not hard to see that the total value derived by the buyers when $x$ buyers buy the desired bundles is the integral of $\lx$ from $0$ to $x$. Alternatively, we could define the inverse demand function for a market with one atomic buyer with a large demand.

{\em Atomic buyers:} While we defined the model above with non-atomic buyers who each desire an infinitesimal amount of any bundle in $B_i$, we could equivalently have defined the market as consisting of a single buyer with a large demand, so that this buyer's utility is $\Lambda(x)$ for obtaining $x$ amount of any combination of bundles in $B_i$. Note that this is very different from simply assuming a single buyer who wants one unit of a good at a fixed value; properties of the utility function $\Lambda(x)$ can greatly change the properties of the system.

\textbf{Direct Demand} The inverse demand function captures the price that the market can pay as a function of the total demand. Such functions are very commonly used in Economics while studying the efficiency of markets with divisible goods as they provide a simple way to relate social welfare and prices. On the other hand, in markets with several atomic buyers, it is also natural to model the distribution of demand using a `direct demand function', which is essentially the inverse of the inverse demand function. Mathematically, we can represent the direct demand by a non-increasing function $D(p)$, such that exactly $D(p)$ buyers hold a value of $p$ or more for the bundles.

Let us also define a quantity $d(p) = -\frac{d}{dp}D(p)$, which in some sense tells us the amount of buyers who who hold a value of exactly $p$ for the bundles.

We now define some commonly used quantities in terms of the direct and inverse demand functions.

\begin{enumerate}
\item The value $\frac{dp}{dx}$ is the change in the value or `lowest-valued' buyers when you add or remove some buyers. It is not hard to see that this is equal to $\lambda'(x)$ at some value of $x$. Similarly, the quantity $\frac{dx}{dp} = -d(p)$.

\item \textbf{Elasticity of Demand} Informally, this is the relative `responsiveness' of the market over the sensitivity to price. Mathematically $\epsilon = \frac{dx/x}{dp/p}$. The elasticity can be rewritten as a function of demand ($x$) as $\epsilon (x) = \frac{\lambda(x)}{x\lambda'(x)}$ because $p$ as a function of $x$ is exactly equal to $\lambda(x)$. The elasticity of demand is negative for the vast majority of goods.

\item \textbf{Constant Elasticity of Demand (CED)} This represents the class of (inverse demand functions) such that the elasticity of demand is constant, i.e., $\epsilon(x) = \epsilon \leq 0$ for all $x$. It known that these functions can be represented in the form of $d(p) = Ap^{\epsilon}$ where $\epsilon < 0$ is referred to as the elasticity of demand. When $-1 < \epsilon < 0$, the demand is said to highly inelastic. As $\epsilon \to -\infty$, the demand becomes more and more elastic. Using some simple algebra, the same inverse demand functions for the CED case are
\begin{align*}
\lx  = (a-x)^{\frac{1}{1-\epsilon}} & \; -1 < \epsilon < 0 & \text{Relatively inelastic demand}\\
\lx  = e^{-x} & \; \epsilon = -1 & \text{Monotone Hazard Rate}\\
\lx  = x^{\frac{1}{1-\epsilon}} & \; -\infty < \epsilon < -1 & \text{Elastic Demand}\\
\end{align*}

Usually the demand for essential commodities (food, electricity) are considered to be inelastic whereas the demand for non-essentials are somewhat elastic depending on the case.
\end{enumerate}

We defined several classes of demand functions in the paper and show bounds on the efficiency for many of these classes. We now discuss each of these demand functions, their interpretations and give examples. Although for ease of presentation, we assume that $\lx$ is continuously differentiable, this assumption is not necessary for any of our results, as we argue in Section \ref{sec:generalizations}.

\begin{description}
\setlength{\itemsep}{-1pt}
\item[Uniform buyers] $\lx=\lambda_0 > 0$ for $x \leq T$. In other words, a population of $T$ buyers all have the same value $\lambda_0$ for the bundles. Alternatively, one or more atomic buyers exist in the market and the total amount demanded by all atomic buyers is $T$. And each atomic buyer receives a value of $\lambda_0$ for every unit of the good(s).

From a traditional Economics point of view, the quantity $\frac{dx}{dp} = 0$ for $p < \lambda_0$.

\item[Concave Demand] $\ldx$ is a non-increasing function of $x$. Informally, $\lx$ is decreasing faster and faster as $x$ increases. This means that if it so happens that the buyers are in agreement over the value of the bundles, then more buyers are clustered around a higher price than a lower price.

This includes the popular linear inverse demand case ($\lx = a-x$)~\cite{abolhassani2014network, tsitsiklis2012efficiency} where the demand drops linearly as price increases. In this case $\frac{dx}{dp}$ is a constant and so changing the price always leads to the same (addition) or removal of buyers. That is buyers at every price are equally senstive to a change in price. Another example function that we examine later is concave and polynomially decreasing, i.e., $\lx = a-x^\alpha$ for $\alpha \geq 1$.

The classic way of representing such functions would be to consider $|\frac{dx}{dp}|$ is non-decreasing with $p$.

\item[Monotone Hazard Rate (MHR) Demand] This is a strict generalization of Concave Demand. Mathematically, it is the class of functions where $\frac{\ldx}{\lx}$ is non-increasing or $h(x)=\frac{|\ldx|}{\lx}$ is non-decreasing in $x$. This is equivalent to the class of {\em log-concave} functions~\cite{amir1996cournot, baldenius2000comparative} and essentially captures inverse demand functions without a heavy tail.

While uniformly decreasing demand (linear inverse demand, $\lx = a-x$) has been assumed more commonly due to its tractable nature, it is more likely that the elasticity of demand is not constant across different prices. Indeed different segments of the market may react differently to a change in price. MHR functions capture a very interesting class of such functions where the responsiveness of a market relative to the value of the buyers is non-decreasing. More concretely if at some price $2p$, an increase of $dp$ in the price leads to a reduction in $dx$ number of buyers. Then at a price $p$, in order to make the same number of consumers ($dx$) drop out, the increase in price has to be at least $\frac{1}{2}dp$. In simple terms, the market cannot be `overly sensitive' at smaller prices compared to its sensitivity at a larger price. Example function: $\lx = e^{-x}$.

Such functions denote the case where $\frac{dx}{dp/p}$ is non-decreasing with $p$.

%decreasing or $h(x)= \frac{|\ldx|}{\lx}$ is an increasing function of $x$. This contains inverse demand functions without a heavy tail, and can be considered a measure of convexity \cite{.}. When $\lx$ is small the function cannot decrease too slowly. Example function: $\lx = e^{-x}$.

\item[Monotone Price Elasticity (MPE)] $xh(x) = \frac{x |\ldx|}{\lx}$ is a non-decreasing function of $x$ which tends to zero as $x \to 0$. This is equivalent to functions where the price elasticity of demand is non-decreasing as the price increases. Notice that $\epsilon(x) = \frac{1}{-xh(x)}$. So since $xh(x)$ is non-decreasing, this means that $\epsilon(x)$ is also non-decreasing with $x$. If we look at the absolute value of $\epsilon(x)$, it is decreasing in $x$, which means that the market cannot be too elastic at a lower price as compared to a higher price. Or buyers who value the item more are more sensitive.

In addition to all papers that capture MHR, Concave and Uniform demand, several other papers have also looked at these kind of functions~\cite{abolhassani2014network, chawla2009price, kuleshov2012efficiency, bimpikis2014cournot, caplin1991aggregation, anderson2003efficiency}.
\end{description}

 While uniformly decreasing demand (linear inverse demand, $\lx = a-x$) has been assumed more commonly due to its tractable nature, it is more likely that the elasticity of demand is not constant across different prices. Indeed different segments of the market may react differently to a change in price. MHR functions capture a very interesting class of such functions where the responsiveness of a market relative to the value of the buyers is non-decreasing. More concretely if at some price $2p$, an increase of $dp$ in the price leads to a reduction in $dx$ number of buyers. Then at a price $p$, in order to make the same number of consumers ($dx$) drop out, the increase in price has to be at least $\frac{1}{2}dp$. In simple terms, the market cannot be `overly sensitive' at smaller prices compared to its sensitivity at a larger price.

Each class of demand functions contains all the previously defined classes, i.e., Uniform $\subseteq$ Concave $\subseteq$ MHR $\subseteq$ MPE. We also extend these definitions to functions which may not be differentiable everywhere. However, it is well known that if the $\lx$ is monotone, then its derivative exists `almost everywhere'. In this case, we define the MPE class to be obey monotonicity where the derivative is defined and at any $x_0$ where it is not defined, we require
$$\lim_{x \to x^-_0}x\frac{|\ldx|}{\lx} \leq \lim_{x \to x^+_0}x\frac{|\ldx|}{\lx}.$$

We also consider specific demand functions that have more specific forms. $\lx = a-x^{\alpha}$ for $\alpha \geq 1$ which are concave but essentially more `concave' than simple linear inverse demand. $\lx = (a-x)^{\alpha}$, which belong to the class MHR but are relatively more inelastic. Notice that linear inverse demand functions form some sort of a border between these two classes.

Finally, we look at functions of the form $\lx - r(x) = |\ln(x/a)|^{1\alpha}$ where $\alpha \geq 1$. Ignoring the structure of the cost function, this essentially means that the direct demand is exponentially decreasing.

\section{Min-Cost Flows}
\label{app:mincostflows}

In this section we prove Proposition \ref{prop:Rstuff}, establishing the properties of the min-cost flow value $R(x)$.

\noindent\textbf{Flow Allocations:} We now discuss flow allocations, which we use only in our proofs. Given a flow $\vec{x}$ whose total magnitude is $x$, we represent the flow by an allocation vector $\alpha$, such that the flow on edge $e$ is $\alpha_e x=x_e$ for $\alpha_e \leq 1$. Then, we can define the total cost as a function of just $x$, i.e., $C^{\alpha}(x) = \sum_e C_e(\alpha_e x)$ for a fixed allocation rule $\alpha$. We also define a total differential cost function, the marginal cost of sending additional flow according to a given allocation.
$$c^{\alpha}(x)=\frac{d}{dx}C^{\alpha}(x) = \sum_e \frac{d}{dx}C_e(\alpha_ex) = \sum_e \alpha_e c_e(\alpha_e x).$$

We are now in a position to prove some properties of the min-cost function $R(x)$. Intuitively, it also seems like $R(x)$ must be continuous since increasing the flow by a small amount should not lead to a large increase in the cost. We begin by showing this formally.
\begin{proposition}
$R(x)$ is continuous $\forall x$.
\end{proposition}
\begin{proof}
%In order to show that the function is continuous, we need to show that for any $p$, $\lim_{x \to p^-}R(x) = \lim_{x \to p^+}R(x)= R(p)$.
Since $R$ is a non-decreasing function, we know the following inequality must hold: $\lim_{\epsilon \to 0}R(x-\epsilon) \leq R(x) \leq \lim_{\epsilon \to 0}R(x+\epsilon)$. Let $R(x) = C^{\alpha}(x)$, i.e., $\alpha$ is the optimal allocation for the min-cost flow of value $x$.  Then for any $\epsilon > 0$, $R(x+\epsilon) \leq C^{\alpha}(x+\epsilon)$, simply because $R(x+\epsilon)$ is the cost of the best allocation, which includes the allocation $\alpha$. Taking the limit on both sides as $\epsilon$ tends to zero, we get that $\lim_{\epsilon \to 0}R(x+\epsilon)$ cannot be any larger than $C^{\alpha}(x) = R(x)$. Thus we get that $\lim_{\epsilon \to 0}R(x+\epsilon) = R(x)$ for all $x$.

Moving on to the other limit, suppose that $\lim_{\epsilon \to 0}R(x-\epsilon) = r_1 < R(x) = \lim_{\epsilon \to 0}R(x+\epsilon)$. For some sufficiently small $\epsilon$, let $\vec{x_1}$ be the min-cost flow corresponding to $R(x-\epsilon)$ and $\vec{x_2}$ be the flow corresponding to $R(x+\epsilon)$. Consider the flow $\vec{x}=\frac{1}{2}\vec{x_1}  + \frac{1}{2}\vec{x_2}$. Clearly, this is a feasible flow of magnitude $x$. Moreover, since $C$ is convex, we know that $C(\vec{x}) \leq \frac{1}{2}C(\vec{x_1})  + \frac{1}{2}C(\vec{x_2}) < R(x)$. The last inequality holds because for some sufficiently small $\epsilon$, the convex combination $\frac{1}{2}C(\vec{x_1})  + \frac{1}{2}C(\vec{x_2})$ lies in between $r_1$ and $R(x)$. This is a contradiction since $R(x)$ is the minimum cost of a flow of magnitude $x$. Thus $R(x)$ is continuous at $x$.
\end{proof}

Recall the derivative of the min-cost flow $r(x) = \frac{d}{dx}R(x)$. Although $R(x)$ is continuous, it is not clear that it is differentiable, so we define the appropriate left and right hand derivatives $r^-(x)$, $r^+(x)$ according to the first principles. We now show that these two are the same. Then, we show that the derivative $r(x)$ is continuous and non-decreasing in $x$.

\begin{proposition}\label{prop:Rdiff}
$R(x)$ is differentiable for all $x$. %, i.e., $r^+(x) = r^-(x)$.
\end{proposition}
\begin{proof}
Recall the differential cost function for a fixed allocation $c^{\alpha}(x)=\sum_e \alpha_e c_e(\alpha_e x)$. Since $C_e$ is differentiable, then $c^\alpha$ is continuous. More precisely,
%Fix a value of $x=x^*$ and let the corresponding optimal allocation vector be $\alpha$. Then, we define
$c^{\alpha}(x) = \displaystyle \lim_{\epsilon \to 0} \frac{C^{\alpha}(x) - C^{\alpha}(x-\epsilon)}{\epsilon} = \lim_{\epsilon \to 0} \frac{C^{\alpha}(x+\epsilon) - C^{\alpha}(x)}{\epsilon}$.

Let $\alpha(x)$ denote the optimal flow allocation for a flow of size $x$ (i.e., $R(x)=C^{\alpha(x)}(x)$). Then, we know that
$$r^-(x) = \lim_{\epsilon \to 0}\frac{R(x) - R(x-\epsilon)}{\epsilon} = \lim_{\epsilon \to 0} \frac{C^{\alpha(x)}(x) - C^{\alpha(x - \epsilon)}(x - \epsilon)}{\epsilon},$$

$$r^+(x) = \lim_{\epsilon \to 0}\frac{R(x+\epsilon) - R(x)}{\epsilon} = \lim_{\epsilon \to 0} \frac{C^{\alpha(x + \epsilon)}(x + \epsilon)  - C^{\alpha(x)}(x)}{\epsilon}.$$

From the above definitions, we see that $r^-(x)$ is always greater than $c^{\alpha(x)}(x)$. This is true because $R(x-\epsilon) \leq C^{\alpha(x)}(x-\epsilon)$ since the former is the min-cost flow of that magnitude, and the latter may be using a suboptimal allocation. Similarly, $R(x+\epsilon) \leq C^{\alpha(x)}(x+\epsilon)$, so %$\displaystyle \lim_{\epsilon \to 0}\frac{R(x^*+\epsilon) - R(x^*)}{\epsilon} \leq \lim_{\epsilon \to 0} \frac{C^{\alpha}(x^*+\epsilon) - C^{\alpha}(x^*)}{\epsilon}$ and
we have $r^+(x) \leq c^{\alpha(x)}(x)$. This gives us our first bound,

$$r^+(x) \leq c^{\alpha(x)}(x) \leq r^-(x).$$

Now, suppose that $r^+(x) < r^-(x)$. This would imply that
$$\lim_{\epsilon \to 0}\frac{R(x) - R(x-\epsilon)}{\epsilon} > \lim_{\epsilon \to 0}\frac{R(x+\epsilon) - R(x)}{\epsilon}.$$ Thus there exists a sufficiently small $\epsilon_0$ such that
%Or equivalently, $\exists \epsilon_0$, such that $\forall \epsilon \leq \epsilon_0$,
$R(x) - R(x-\epsilon_0) > R(x+\epsilon_0) - R(x)$ which implies $R(x) > \frac{1}{2}(R(x+\epsilon_0) + R(x - \epsilon_0))$.

Let the flows corresponding to $R(x+\epsilon_0)$ and $R(x-\epsilon_0)$ be $\vec{x_1}$ and $\vec{x_2}$ respectively. Once again, take the average flow $\vec{x}=\frac{1}{2}(\vec{x_1} + \vec{x_2})$. This flow has a magnitude of $x$ and its cost $C(\vec{x}) \leq \frac{1}{2}(C(\vec{x_1}) + C(\vec{x_2})) = \frac{1}{2}(R(x+\epsilon_0) + R(x - \epsilon_0)) < R(x)$. However, this is a contradiction since no feasible flow of magnitude $x$ can have a cost less than $R(x)$. So we conclude that $r^+(x) \geq r^-(x)$. So finally, we have

$$r^+(x) = c^{\alpha(x)}(x) = r^-(x).$$

Moreover, since all $C_e$ are differentiable, then $c^{\alpha(x)}(x)$ is always finite, thus proving the desired claim that $r(x)$ always exists and is continuous.
\end{proof}

\begin{proposition}
\label{rcontinuous}
$R(x)$ is convex.
\end{proposition}
\begin{proof}
Consider arbitrary flow values $x_1$ and $x_2$, and let $x=\frac{1}{2}x_1 + \frac{1}{2}x_2$. Suppose to the contrary that $R(x)>\frac{1}{2}R(x_1)+\frac{1}{2}R(x_2)$. Now consider the flow vector $\vec{x} = \frac{1}{2}\vec{x_1} + \frac{1}{2}\vec{x_2}$ where $\vec{x_1}, \vec{x_2}$ are the flow vectors corresponding to the minimum cost flows at $x_1$, $x_2$ respectively. Clearly $\vec{x}$ has magnitude $x$. Moreover since $C$ is convex, $C(\vec{x}) \leq \frac{1}{2}(C(\vec{x_1}) + C(\vec{x_2})) =\frac{1}{2}(R(x_1)+R(x_2)) < R(x).$ This is a contradiction as no flow of magnitude $x$ can have a cost less than $R(x)$, and so $R(x)$ must be convex.
\end{proof}

Now that we have a good understanding of the min-cost function $R(x)$, we show that if we take the min-cost flow at any value $x$, the cost of sending an additional infinitesimal flow $dx$ equals the marginal cost of any one path with flow on it. First we show some simple lemmas that yield some insight on the allocation on the min-cost flow.

\begin{lemma}
\label{lemma_nonzeroflowcost}
For any given $x$, let $\vec{x}$ be the minimum cost flow vector and $P_i, P_j$ be any two paths with non-zero flow in $\vec{x}$. Then, $\sum_{e \in P_i}c_e(x_e) = \sum_{e \in P_j}c_e(x_e)$. If $P_i$ is a path with no flow and $P_j$ has non-zero flow in the minimum cost flow, then
$\sum_{e \in P_i}c_e(x_e) \geq \sum_{e \in P_j}c_e(x_e)$.
\end{lemma}
\begin{proof}
The lemma can be formally proved using the Karush-Kuhn-Tucker conditions for the min-cost flow optimization program. However, observe any cost minimizing flow must also be a local optimum (for a convex program, the local and global optima coincide). If the above lemma were not true, and for two such paths suppose that $\sum_{e \in P_i}c_e(x_e) > \sum_{e \in P_j}c_e(x_e)$, without loss of generality. Then, consider a new flow with the flow on $P_i$ reduced by some sufficiently small $\epsilon$ and the flow on $P_j$ increased by the same amount. If the cost of the old flow is $z$, then the cost of the new flow is
\begin{align*}
& z + \sum_{e \in P_j}[C_e(x_e + \epsilon) - C_e(x_e)] + \sum_{e \in P_i}[C_e(x_e) - C_e(x_e - \epsilon)]\\
= & z + \epsilon(\sum_{e \in P_j}c_e(x_e) - \sum_{e \in P_i}c_e(x_e))\\
< z,
\end{align*}
which is a contradiction because $z$ is minimum cost among all feasible flows supporting a flow magnitude of $x$. The same argument works for switching a small amount of flow to a path with zero flow.
\end{proof}

\begin{lemma}
\label{lemma_cequivalence}
For a minimum cost flow vector $\vec{x}$ of magnitude $x$, and for any path $P_i$ with $x_{P_i} > 0$, we have $r(x) = \sum_{e \in P_i}c_e(x_e)$.
\end{lemma}
\begin{proof}
From the proof of $R$ being differentiable, we know that if $\alpha$ is the min-cost allocation for a flow of magnitude $x$, then $c^{\alpha}(x) = r(x)$. So we just have to prove that $c^{\alpha}(x) = \sum_{e \in P_i}c_e(x_e)$ for any path with non-zero flow.

For every path $P_i$ with non-zero flow, we know by Lemma~\ref{lemma_nonzeroflowcost} that $\sum_{e \in P_i}c_e(x_e)$ is the same. Let the value of this quantity by $y$. We also know that for any given edge $e$, $\alpha_e = \sum_{P_i \ni e}\alpha_{P_i}$ by definition. So if $P$ is the set of paths with non-zero flow, we have
\begin{align*}
c^{\alpha}(x) & = & \sum_e \alpha_e c_e(x_e) \\
 & = & \sum_{P_i \in P} \sum_{e \in P_i} \alpha_{P_i}c_e(x_e) \\
 & = & \sum_{P_i \in P} \alpha_{P_i} \sum_{e \in P_i} c_e(x_e) \\
 & = & y \sum_{P_i \in P} \alpha_{P_i}= y \\
\end{align*}
\end{proof}

\section{Proofs from Section 3: Existence and Computation of Equilibrium Prices}
\label{app:existence}

\subsection{Our Pricing rule:}
We redefine our pricing rule which we will use for both the existence and uniqueness proofs. Let $E$ be the set of edges and $c_e(x_e)$ be the differential cost function on each edge. Throughout this section $\vec{x}$ will refer to a valid flow and $x$ to the magnitude of this flow. Consider any instance with $M$ monopolies. For any given minimum cost flow $\vec{x}$ of magnitude $x$, we will use the following pricing rule:
\begin{equation}
\label{appeqn_ourpricingrule}
 p_e(\vec{x}) =
  \begin{cases}
      \hfill \frac{\lx - r(x)}{M} + c_e(x) \hfill & \text{if e is a Monopoly} \\
      \hfill c_e(x_e) \hfill & \text{otherwise} \\
  \end{cases}
\end{equation}

It is easy to verify (from the properties of minimum cost flows) that every flow carrying path has a price of exactly $\lambda(x)$. Indeed, we already know that the total marginal cost of every flow carrying path is $r(x)$. According to the pricing rule, edges are first priced at their marginal cost. In addition, for any flow $\lx-r(x)$ is the total available slack or surplus, which is divided equally among all the monopolies. For example, if all edges are monopolies having the same cost function, then $r(x) = Mc_e(x)$ and $p_e(x) = \frac{\lx}{M}$. Our first simple lemma shows that any edge priced at its marginal cost has no incentive to lower its price.

\begin{lemma}
\label{cl_margcost}
At any flow $\vec{x}$, an edge priced at $p_e = c_e(x_e)$ can never increase its profit by lowering its price.
\end{lemma}

\begin{proof}
Suppose an edge decreases its price from $p_e=c_e(x_e)$ to $p'_e$ and its flow increases from $x_e$ to $x'_e$, then we need to show that $p_ex_e - C_e(x_e) \geq p'_ex'_e - C_e(x'_e)$. Consider the function $p_ex - C_e(x)$. For a fixed $p_e = c_e(x_e)$, its derivative is negative for $x > x_e$. Therefore, for any $x'_e > x_e$, we have $p_ex_e - C_e(x_e) \geq p_ex'_e - C_e(x'_e)$. But since $p'_e \leq p_e$, we know $p_ex_e - C_e(x_e) \geq p_ex'_e - C_e(x'_e) \geq p'_ex'_e - C_e(x'_e)$. This completes the proof.
\end{proof}

\begin{clm_app}{clm_monopolypos1}
In any network with no monopolies (i.e, you cannot disconnect $s$, $t$ by removing any one edge), there exists a Nash Equilibrium maximizing social welfare.
\end{clm_app}
\begin{proof}
The solution is quite straight-forward. We compute the optimum solution $\vec{x^*}$ and price edges according to our pricing rule above, which translates to each edge being priced at $p_e = c_e(x_e)$ since the instance has no monopolies. By Lemma~\ref{lemma_nonzeroflowcost} and Proposition~\ref{prop_optflow}, we know that for every flow carrying path $P_i$, $\lambda(x^*) \geq r(x^*) = \sum_{e \in P_i}c_e(x_e) = \sum_{e \in P_i}p_e$. So all flow carrying paths have the same total price, call it $p^*$. We claim that sending a flow of magnitude $x^*$ is a best-response by the buyers to this price. First consider the case when $\lambda(x^*) > r(x^*)$. According to Proposition~\ref{prop_optflow}, $x^*=T$ is the total flow available in the market. Thus all the buyers have a value of at least $\lambda(x^*)$ for the paths and since the price is not larger than this, all of them purchase one of the $s-t$ paths. When $\lambda(x^*) = r(x^*) = p^*$, it means that exactly $x^*$ of the buyers value the paths at $p^*$ or more, so these many buyers send the flow. In both cases, buyer behavior is indeed a best-response to the prices.

Now we show that sellers cannot change their price unilaterally and increase their profit. We claim that for every edge $e$ with flow, if the edge is removed from the graph, then there exists at least one $s-t$ path with a total price of $p^*$. Then, no edge would wish to unilaterally increase its price as the flow would switch to the alternative path. Since the graph has no monopolies, there exists at least one $s-t$ path not containing $e$ for every edge $e$. Let $P_i = v_1v_2\cdots v_kev_{k+1}\cdots v_r$ be a flow carrying path with $e$ where $v_1=s$ and $v_r=t$. If the flow on $e$ is $x_e < x^*$, then there exists at least one flow carrying path without $e$ and the price of this path is $p^*$, so we are done. If not, then for some $i\leq k$ and some $j\geq k+1$, there must exist a path $\gamma$ between $v_i$ and $v_j$ that only passes through edges with no flow on them. The price on these edges without flow must be $c_e(0) = 0$. Consider the new path $P' = v_1\cdots v_i \gamma v_j\cdots v_r$. The price of this path is no larger than $p^*$, since the price of edges in $\gamma$ is zero. So the price of this path must be exactly $p^*$. Therefore, we conclude that no single edge can increase its price and still retain some flow.

By Lemma~\ref{cl_margcost}, we already know that no edge can priced at its marginal can decrease its price and make more profit, no matter how much the flow increases by. For the $\lambda(x^*) > r(x^*)$ case, the whole market demand is satisfied so decreasing the price has no additional impact anyway. The edges without flow are already priced at $0$, so they cannot decrease their price. This completes the proof.
\end{proof}
\subsection{Proof of Existence}\begin{thm_app}{thm:existence}
For any MPE demand function $\lambda$, there exists a Nash equilibrium Pricing.
\end{thm_app}
\begin{proof}
Recall that a solution of our game is given by $(\vec{p},\vec{x})$. Clearly, once prices are fixed, buyers always buy best-response bundles so we only consider solutions of this form. Finally for any edge, we define $\pe$ to be the increase in the price of an edge from its marginal cost. That is, given a solution, $p_e = c_e(x_e) + \pe$. We will now show some sufficient conditions that this `increased' price must obey in an equilibrium.

\begin{lemma}
\label{lem_eqnconditionsincr}
Given a solution $(\vec{p}, \vec{\tilde{x}})$ with $\vec{\tilde{x}}$ a best-response flow to prices $\vec{p}$, with $\lambda(\tilde{x}) > 0$ and $\lambda(\tilde{x}) \geq \tilde{x}|\lambda'(\tilde{x})|$, we have that no seller $e$ can increase his price and improve profits as long as either one of the following conditions hold,
\begin{enumerate}
\item The good $e$ is tight, i..e, $\exists$ some $s$-$t$ path that does not contain the edge $e$ and has the same total price as the flow-carrying $s-t$ paths that do contain $e$ \emph{(or)}

\item $\pe \geq \tilde{x}|\lambda'(\tilde{x})|$
\end{enumerate}
\end{lemma}
\begin{proof}
The first part of the lemma is fairly trivial: let $P_1$ be any flow-carrying path that contains $e$ and $P_2$ be a path not containing $e$ but having the same overall price as $P_1$. Since buyers always buy the best-response bundles, if seller $e$ increases his price, the price of $P_1$ (or any other flow carrying path containing $e$) would become strictly larger than that of $P_2$. All the buyers that originally purchased from this seller would now shift to using $P_2$ or any other path not containing $e$, and seller $e$'s profit would become zero which cannot be strictly larger than his original profit.

Suppose an edge $e$ is not tight and a seller who obeys $\pe \geq \tilde{x}|\lambda'(\tilde{x})|$ can increase his price (perhaps to some extent) and still have the cheapest paths pass through this edge. In this case, it is necessary that the entire flow $\tilde{x}$ must be using this edge. To see why, suppose if the flow on this edge $x_e < \tilde{x}$, then there must be atleast one other $s$-$t$ path not containing $e$ with non-zero flow on it. Since all the paths with non-zero flow must have the same price ($\lambda(\tilde{x}))$, this boils down to our first condition of $e$ being tight since there is an alternative path with the same price. Therefore, $x_e = \tilde{x}$ and every flow carrying path contains $e$.

Now suppose that seller $e$ increases his price from $p_e$ to $p'_e$ and the resulting flow on the edge is $x$. Then, it is not hard to see that $x < \tilde{x}$, because the price of every flow carrying path has increased by a non-zero amount ($p'_e - p_e$) and if the flow remained the same, then it would mean that $\lambda(x)$ is not uniquely defined at $\tilde{x})$. We first establish the relation between $p'_e$ and $\lx$, namely that $p'_e - p_e = \lx - \lambda(\tilde{x})$. To see why, notice that the flow of size $x$ is a best-response flow to the prices in our original solution, except with price $p'_e$ on edge $e$. As we argued above, all cheapest paths in our given solution $\vec{\tilde{x}}$ have cost $\lambda(\tilde{x})$; thus all cheapest paths in this new pricing have cost $\lambda(\tilde{x})-p_e+p'_e$. Since this pricing results in a flow of size $x$, it must be that $\lambda(x)=\lambda(\tilde{x})-p_e+p'_e$: the buyers who purchase the cheapest bundles are exactly the ones who value them more than $\lambda(x)$. Thus, we know that $p'_e=p_e + \lx - \lambda(\tilde{x})$.

Now consider seller $e$'s original profit, $p_e\tilde{x} - C_e(\tilde{x})$. After $e$ changes its price to $p'_e$, this profit become $p'_ex - C_e(x)$, which by the above argument equals $[p_e + \lx - \lambda(\tilde{x})]x-C_e(x)$. We want to show that this new profit is at most the old profit as long as Condition (2) is obeyed. Define $\pi(x)=[p_e + \lx - \lambda(\tilde{x})]x-C_e(x)$. We will prove that in the domain $(0,\tilde{x}]$, $\pi(x)$ is maximized when $x=\tilde{x}$, thus implying that no matter what amount the seller increases his price by, at the resulting flow of magnitude $x$, his profit cannot be strictly larger than the original profit.

%Before proving that $\pi$ is maximized at $\tilde{x}$, we establish the following helpful inequality.
We now proceed to prove that $\pi(x)$ is maximized at $\tilde{x}$. Specifically, we look at the derivative of $\pi(x)$ and show that it is non negative for $x \leq \tilde{x}$. Recall that the seller obeys Condition (2) which implies,
$$p_e = \pe + c_e(\tilde{x}) \geq \tilde{x}|\lambda'(x)| + c_e(\tilde{x}).$$

%By our pricing rule, and by our choice of $\tilde{x}$, we know that $$p_e=\frac{\lambda(\tilde{x}) - r(\tilde{x})}{M} + c_e(\tilde{x}) = \tilde{x}|\lambda'(\tilde{x})|+c_e(\tilde{x}),$$
Since
$$\pi(x)=[p_e + \lx - \lambda(\tilde{x})]x-C_e(x).$$ Thus the derivative of $\pi(x)$ is equal to
\begin{align*}
\pi'(x) & = p_e + \lx - \lambda(\tilde{x}) + x\ldx - c_e(x) \\
& = \pe + c_e(\tilde{x}) + \lx - \lambda(\tilde{x}) + x\ldx - c_e(x)\\
& \geq \lx - \lambda(\tilde{x}) + \tilde{x}|\lambda'(\tilde{x})| - x|\ldx| + (c_e(\tilde{x}) - c_e(x))
\end{align*}

Since the last term in the parenthesis, $c_e(\tilde{x}) - c_e(x)$ is non-negative ($c_e$ is non-decreasing), in order to show that $\pi'(x) \geq 0$ for all $x \leq \tilde{x}$, it suffices if we show the first terms are non-negative. The following proposition implies that for MPE functions, this is indeed true.

\begin{proposition}
\label{sublem_exist_smallx}
 For $\lambda\in MPE$ and $M\geq 1$, we have that $\lambda(\tilde{x})-\lambda(x)\leq \tilde{x}|\lambda'(\tilde{x})|-x|\lambda'(x)|$ for $x<\tilde{x}$.
\end{proposition}

\begin{proof}
First suppose that $|\lambda'(\tilde{x})| > 0$. Let $A=-(\lambda(\tilde{x})-\lambda(x))$ and $B=-(\tilde{x}|\lambda'(\tilde{x})|-x|\lambda'(x)|)$; we want to show that $A\geq B$. Suppose to the contrary that $A<B$. $A$ is non-negative since $\lambda$ is non-increasing, and $B$ is nonnegative since $B>A$. We know by the property of MPE functions that

\begin{equation}\label{eq.helper2}
\frac{\lambda(\tilde{x})}{\tilde{x}|\lambda'(\tilde{x})|}\leq \frac{\lx}{x|\ldx|} = \frac{\lambda(\tilde{x})+A}{\tilde{x}|\lambda'(\tilde{x})|+B}.
\end{equation}

Let $C=\lambda(\tilde{x})$ and $D=\tilde{x}|\lambda'(\tilde{x})|$, and consider how $C/D$ compares with $(C+B)/(D+B)$. First, notice that we have at $\tilde{x}$,
$$C = \lambda(\tilde{x}) \geq \lambda(\tilde{x}) - r(\tilde{x}) \geq M\tilde{x}|\lambda'(\tilde{x})| \geq \tilde{x}|\lambda'(\tilde{x})| = D.$$ Therefore, $C \geq D$. Thus we know that $C(D+B)\geq D(C+B)$, and thus $\frac{C}{D}\geq \frac{C+B}{D+B}$. By our assumption that $A<B$, this implies that $\frac{C}{D}> \frac{C+A}{D+B}$, i.e.,

$$\frac{\lambda(\tilde{x})}{\tilde{x}|\lambda'(\tilde{x})|} >  \frac{\lambda(\tilde{x})+A}{\tilde{x}|\lambda'(\tilde{x})|+B}.$$
This contradicts Inequality (\ref{eq.helper2}) above, thus proving that $A\geq B$.

What if $\lambda'(\tilde{x}) = 0$? Since MPE functions must have non-decreasing $\frac{x|\lambda'(x)|}{\lx}$, we can only conclude that for all $0 < x \leq \tilde{x}$, $\lambda'(x) = 0$. In this case, it is not hard to see that the proposition trivially holds since $\lambda(\tilde{x}) = \lambda(x)$ and the RHS is zero as well.
\end{proof}

We have therefore shown that for all $0 < x < \tilde{x}$, $\pi'(x)$ is non-negative. This means that $\pi(x)$ is non-decreasing in this region and therefore maximized at $x=\tilde{x}$ in the domain $[0,\tilde{x}]$. Therefore, no `monopoly edge' can benefit by changing its price, as desired.
%A small remark on how much the monopoly can actually increase its price by is in order here. It may be possible (as we show later for the case with Virtual monopolies in Section~\ref{appsec:generalizations}) that a monopoly may only be able to increase his price up to some amount and so may not be able to obtain every possible flow size $x < \tilde{x}$. That said, the above proof is still valid for this case because we have shown that in every possible domain $[x,\tilde{x}]$, $\pi(x)$ is non-decreasing. So our statement is much stronger.
 Also notice that if $\tilde{x} = T$, then $\lambda'(T)$ has to be finite if the prices are finite; the above argument works for this case as well. $\qed$
\end{proof}

We now show sufficient conditions on the other half of seller behavior, namely give conditions on when a seller cannot decrease his price.

\begin{lemma}
\label{lem_eqnconditionsdecr}
Given a solution $(\vec{p}, \vec{\tilde{x}})$ with $\vec{\tilde{x}}$ a best-response flow to prices $\vec{p}$, satisfying $\lambda(\tilde{x}) \geq \tilde{x}|\lambda'(\tilde{x})|$, no seller $e$ can decrease his price and improve profits as long as any  one of the following conditions hold,
\begin{enumerate}
\item $\tilde{x_e} = \tilde{x} =T$ \emph{or}

\item $\pe = 0$ \emph{or}

\item $\pe \leq \tilde{x}|\lambda'(\tilde{x})|$ and $\tilde{x_e} = \tilde{x}$
\end{enumerate}
\end{lemma}
\begin{proof}
Once again the first two conditions are easy to prove. If $\tilde{x_e} = T$, it just means that there are no more buyers left in the market. So for any such seller, a decrease in price is not going to lead to any additional flow. In the second case, $\pe = 0$ simply means that a seller is priced at its marginal price and so by Lemma~\ref{cl_margcost} no such seller would wish to decrease his price. For the final condition, suppose that $\tilde{x} < T$ and $\pe > 0$.

Notice that as with the previous proof, any decrease in price from $p_e$ to $p'_e$ will result in a flow of magnitude $x > \tilde{x}$. As we argued previously, since buyers will only indulge in best-response behavior, it is necessary that $p'_e = p_e + \lambda(\tilde{x}) - \lambda(x)$.

We want to show that the new profit a seller makes is at most the old profit as long as Condition (3) is obeyed. Consider seller $e$'s profit, $\pi(x)=[p_e + \lx - \lambda(\tilde{x})]x-C_e(x)$. We will prove that in the domain $[\tilde{x},T]$, $\pi(x)$ is maximized when $x=\tilde{x}$, thus implying that no matter what amount the seller decreases his price by, at the resulting flow of magnitude $x$, his profit cannot be strictly larger than the original profit.

%Before proving that $\pi$ is maximized at $\tilde{x}$, we establish the following helpful inequality.
We now proceed to prove that $\pi(x)$ is maximized at $\tilde{x}$. Specifically, we look at the derivative of $\pi(x)$ and show that it is not positive for $x \geq \tilde{x}$. Recall that the seller obeys Condition (3) which implies,
$$p_e = \pe + c_e(\tilde{x}) \leq \tilde{x}|\lambda'(x)| + c_e(\tilde{x}).$$

%By our pricing rule, and by our choice of $\tilde{x}$, we know that $$p_e=\frac{\lambda(\tilde{x}) - r(\tilde{x})}{M} + c_e(\tilde{x}) = \tilde{x}|\lambda'(\tilde{x})|+c_e(\tilde{x}),$$
Since
$$\pi(x)=[p_e + \lx - \lambda(\tilde{x})]x-C_e(x).$$ Thus the derivative of $\pi(x)$ is equal to
\begin{align*}
\pi'(x) & = p_e + \lx - \lambda(\tilde{x}) + x\ldx - c_e(x) \\
& = \pe + c_e(\tilde{x}) + \lx - \lambda(\tilde{x}) + x\ldx - c_e(x)\\
& \leq \lx - \lambda(\tilde{x}) + \tilde{x}|\lambda'(\tilde{x})| - x|\ldx| + (c_e(\tilde{x}) - c_e(x))
\end{align*}

Since the last term in the parenthesis, $c_e(\tilde{x}) - c_e(x)$ is not positive ($c_e$ is non-decreasing), in order to show that $\pi'(x) \leq 0$ for all $x \geq \tilde{x}$, it is sufficient to show the first terms are not positive either. The following proposition implies that for MPE functions, this is indeed true.

\begin{lemma}
\label{sublem_exist_largex}
For $\lambda\in MPE$ and $M\geq 1$, we have that $\lambda(\tilde{x})-\lambda(x)\geq \tilde{x}|\lambda'(\tilde{x})|-x|\lambda'(x)|$ for $x>\tilde{x}$.
\end{lemma}
\begin{proof}
First notice that if $\lambda'(\tilde{x}) = 0$, the inequality is trivially true. $\lambda'(\tilde{x})$ is also bonded because the only case where it is not is when $\tilde{x} = T$ and for that case we already know that no seller will decrease the price. Now, let $A=\lambda(\tilde{x})-\lambda(x)$ and $B=\tilde{x}|\lambda'(\tilde{x})|-x|\lambda'(x)|$; we want to show that $A\geq B$. Suppose to the contrary that $A<B$. We know by the property of MPE functions that

\begin{equation}\label{eq.helper1}
\frac{\lambda(\tilde{x})}{\tilde{x}|\lambda'(\tilde{x})|}\geq \frac{\lx}{x|\ldx|} = \frac{\lambda(\tilde{x})-A}{\tilde{x}|\lambda'(\tilde{x})|-B}.
\end{equation}

Let $C=\lambda(\tilde{x})$ and $D=\tilde{x}|\lambda'(\tilde{x})|$, and consider how $C/D$ compares with $(C-A)/(D-A)$. First, notice that $C\geq D$:
this is because by our choice of $\tilde{x}$, we have that $\lambda(\tilde{x})-r(\tilde{x})=M\tilde{x}|\lambda'(\tilde{x})|,$ and we can assume that $M\geq 1$ since $e$ is a monopoly edge. Also, $D\geq B>A$, so we know that $C(D-A)\leq D(C-A)$, and thus $\frac{C}{D}\leq \frac{C-A}{D-A}$. By our assumption that $A<B$, this implies that $\frac{C}{D} \leq \frac{C-A}{D-A} < \frac{C-A}{D-B}$, i.e.,

$$\frac{\lambda(\tilde{x})}{\tilde{x}|\lambda'(\tilde{x})|} <  \frac{\lambda(\tilde{x})-A}{\tilde{x}|\lambda'(\tilde{x})|-B}.$$
This contradicts Inequality (\ref{eq.helper1}) above, thus proving that $A\geq B$.
\end{proof}

Therefore, we have shown that for MPE functions, as long as condition (3) is obeyed, $\pi'(x)$ is positive for all $\tilde{x} \leq x \leq T$. So it means that $\pi(x)$ is non increasing in this domain and therefore it is maximized at $x=\tilde{x}$. Also notice that $\pi(x)$ is a continuous function of $x$. So even if $\lambda'(T)$ is not defined, we know that for all $x < T$, $\pi(x) \leq \pi(\tilde{x})$. This means that $\pi(T) = \lim_{x \to T}\pi(x) \leq \pi(\tilde{x})$ as well.$\qed$
\end{proof}

\begin{corollary}
\label{corr_eqconditionssuff}
Any given solution $(\vec{p}, \vec{x})$ with $\lambda(\tilde{x}) > 0$  and $\vec{x}$ a best-response flow to prices $\vec{p}$ is a Nash Equilibrium if the following conditions are met
\begin{enumerate}
\item $\vec{x}$ is a minimum cost flow of magnitude $x$.

\item All non monopoly edges are priced at their marginal cost, i.e, $\pe= 0$.

\item For all monopoly edges, $\pe \geq \tilde{x}|\lambda'(\tilde{x})|$ and one of the following is true,
\begin{enumerate}
\item $\pe$ is strictly equal to $\tilde{x}|\lambda'(\tilde{x})|$ \emph{or}
\item $\pe = 0$ \emph{or}
\item $\tilde{x} = T$.
\end{enumerate}

\end{enumerate}
\end{corollary}
\begin{proof}
The cost of every flow-carrying path $P$ is exactly $\sum_{e \in P}c_e(\tilde{x}_e) + \sum_{e\in \mathbb{M}}\pe$ where $\mathbb{M}$ is the set of monopolies. Since $\tilde{x}$ is a best-response flow, it must be that $\lambda(\tilde{x})$ equals the cost of these path, so for every monopoly $e$, $\lambda(\tilde{x}) \geq \pe \geq \tilde{x}|\lambda'(\tilde{x})|$. So the requirement for our previous lemmas that $\lambda(\tilde{x}) \geq \tilde{x}|\lambda'(\tilde{x})|$ is satisfied. No monopoly edge will want to change its price due to Lemmas \ref{lem_eqnconditionsincr} and \ref{lem_eqnconditionsdecr}.

Now coming to the non-monopoly edges, since $\pe = 0$, no non-monopoly edge would ever wish to decrease its price. Moreover, since $\vec{x}$ is a minimum-cost allocation, we claim that for every non-monopoly edge $e$, $\exists$ a $s$-$t$ path not containing $e$ with the same price. The proof is very similar to that of Claim~\ref{clm_monopolypos1} and involves showing that for every edge $e$, there exists a \emph{shortcut} or \emph{ear} not containing this edge and having the same price. So we conclude that no edge can increase or decrease its price. Therefore this is a Nash Equilibrium.
\end{proof}

\begin{figure}
\begin{center}
\includegraphics[height=2in]{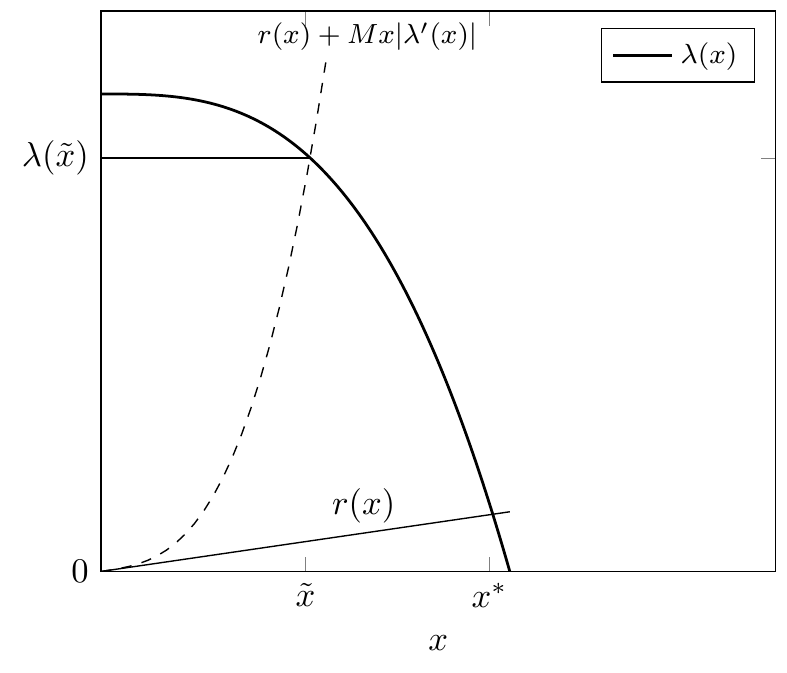}
\end{center}
\caption{The magnitude of the optimal flow $x^*$ and the equilibrium flow $\tilde{x}$.}
\label{fig:existence}
\end{figure}

Recall that at any given minimum cost allocation vector $\vec{\tilde{x}}$, our pricing rule prices all the non-monopoly edges at $\pe = 0$ or equivalently the total price $p_e = c_e(\tilde{x}_e)$. The monopoly edges are priced at $\pe = \frac{\lambda(\tilde{x}) - r(\tilde{x})}{M}$. We have already shown that the flow $\vec{\tilde{x}}$ is a best-response flow to these prices. To show that the prices are stable, paraphrasing Corollary~\ref{corr_eqconditionssuff}, all we need to show is that for any given instance with a MPE demand, we need to show that some $\tilde{x} > 0$ with $\lambda(\tilde{x})>0$ satisfies the following conditions
\begin{itemize}
\item $\lambda(\tilde{x}) - r(\tilde{x}) \geq M\tilde{x}|\lambda'(\tilde{x})|$ (no price increase) and any one of
\begin{enumerate}
\item $\lambda(\tilde{x}) - r(\tilde{x}) = M\tilde{x}|\lambda'(\tilde{x})|$ \emph{or}
\item $\lambda(\tilde{x}) - r(\tilde{x}) = 0$ \emph{or}
\item $\tilde{x} = T$.
\end{enumerate}
\end{itemize}

Let $\vec{x^*}$ be any optimal flow of magnitude $x^*$. Assume for now that $\lambda(x^*) > 0$. If this point satisfies $\lambda(x^*) - r(x^*) \geq Mx^*|\lambda'(x^*)|$, then we claim to have found a point that satisfies all our requirements. Clearly the `no price increase' condition is met by definition. Moreover, either $x^* = T$ in which case Condition (3) is met or $x^* < T$, which implies that $\lambda(x^*) = r(x^*)$ by Proposition~\ref{prop_optflow} and Condition (2) is met. In this case, we set $\tilde{x} = x^*$ this is a Nash Equilibrium. Observe that if the optimum point satisfies this condition, then $|\lambda'(x^*)|$ is bounded because $\lambda(x^*) - r(x^*)$ is bounded. So in the event that $x^* = T$, $\lambda'(T)$ exists and is bounded.

Now suppose that the optimal point does not meet $\lambda(x^*) - r(x^*) \geq Mx^*|\lambda'(x^*)$ (see Figure \ref{fig:existence} for an illustration). Then we claim that $\exists$ some $\tilde{x}\in (0,x^*]$, which satisfies Condition (1) above, and therefore the no price increase condition as well. Since the optimum point does not satisfy the ``No Price Increase" condition, we have $\lambda(x^*) - r(x^*) < \lim_{x \to x^*}Mx|\ldx|$ (the limit is necessary for the case that $x^*=T$ and this $\lambda'(x^*)$ is not defined). Then, we claim that $\exists$ some $\tilde{x}$ obeying the following condition.

\begin{equation}\label{eq:NEcondition} \lambda(\tilde{x})-r(\tilde{x})=M\tilde{x}|\lambda'(\tilde{x})|
\end{equation}

For MPE functions as $x \to 0$, the ratio $\frac{Mx|\ldx|}{\lx}$ also tends to zero. This means that $\exists$ some $\epsilon_0 > 0$ such that for all $0 \leq x \leq \epsilon_0$, $\lx > Mx|\ldx|$ since $M$ is finite. Moreover $r(x)$ goes to zero as $x$ tends to zero. This means that $\lim_{x \to 0}(\lx - r(x) - Mx|\lambda'(x)|) = \lim_{x \to 0}\left(\lx - Mx|\ldx|\right) > 0$ since the term $\lx - Mx|\ldx|$ is strictly positive for sufficiently small $x$ and $r(x) \to 0$. But since the ``no price increase" condition is not true at $x^*$, $\lambda(x^*)-r(x^*) - M x^*|\lambda'(x^*)| < 0$. As $\lx$, $r(x)$ and $\lambda'(x)$ are all continuous for MPE functions (see Appendix~\ref{app:existence} for the case when $\lambda'(x)$ may not be continuous), this means there is an intermediate point where $\lx - r(x) - Mx|\lambda'(x)| =0$. So there exists at least one point $\tilde{x} \in (0,x^*]$, satisfying equation~\ref{eq:NEcondition}. Since $\lambda$ is non-increasing, this means that $\lambda(\tilde{x})\geq \lambda(x^*)>0$.
By the above arguments a min-cost flow of this magnitude with our pricing rule forms a Nash equilibrium, as desired.\\

\textbf{What if $\lambda(x^*)=0$}\\

We needed the assumption that $\lambda(x^*) > 0$, since our Lemmas~\ref{lem_eqnconditionsincr} and \ref{lem_eqnconditionsdecr} required this condition. Suppose that $\lambda(x^*) = 0$ and therefore $r(x) = 0$ everywhere (Proposition \ref{prop_optflow}). If $\lambda(x^*) < \lim_{x \to x^*}Mx|\lambda'(x)|$, then there still exists some $\tilde{x} > 0$ with $\lambda(\tilde{x})>0$ satisfying Equation~\ref{eq:NEcondition} by the same argument as above.

Suppose this is not true and that $\lambda(x^*) = \lim_{x \to x^*}Mx|\lambda'(x)| = 0$. Then, we show (see below) that there still exists some $x < x^*$ where $\lambda(x) - r(x) = \lambda(x) < Mx|\ldx|$. By the same argument, we can consider the domain $(0,x)$ and show that there exists some point satisfying Equation~\ref{eq:NEcondition}.

Assume by contradiction that for all $x < x^*$, $Mx|\ldx| \leq \lx$. By definition of the demand function, we know for all $x < x^*$, $\lambda(x) > 0$. Take some $x_1$ sufficiently close to $x^*$ and consider the following expression,
\begin{align*}
\lambda(x^*) & = \lambda(x_1) - \int_{x=x_1}^{x^*} |\ldx|dx\\
& \geq \lambda(x_1) - \int_{x=x_1}^{x^*} \frac{\lx}{Mx} dx\\
& \geq \lambda(x_1)(1-\frac{1}{M}\log(\frac{x^*}{x_1}))
\end{align*}

But this tells us that $\lambda(x^*) > 0$, a contradiction. So what this tells us is that if $\lambda(x^*) = 0$, then if we take some reasonably close point where $\lambda(x_1) > 0$, then in the domain $[x_1,x^*]$ at least some $x$ must satisfy $Mx|\ldx| > \lx$. $\qed$
\end{proof}

For the purpose of upper-bounding efficiency, it is enough to consider the welfare of non-optimal equilibria.  Therefore, in the rest of the paper, we will focus only on the case where the equilibrium must satisfy $\pe = \tilde{x}|\lambda'(\tilde{x})|$.

\begin{cor_app}{corr_eqconditions}
For any inverse demand function $\lambda$ belonging to the class MPE, we are guaranteed the existence of a Nash Equilibrium with a min-cost flow $(\tilde{x}_e)$ of size $\tilde{x}\leq x^*$ such that,
\begin{enumerate}
\item The prices obey the pricing rule above.
\item Either $\frac{\lambda(\tilde{x})-r(\tilde{x})}{M} = \tilde{x}|\lambda'(\tilde{x})|$ or $\tilde{x} = x^*$, the optimum solution.
\end{enumerate}
\end{cor_app}

We make a small remark here. In case, there are multiple $\tilde{x} > 0$ satisfying the conditions of Corollary~\ref{corr_eqconditions}, we take our solution to be the largest $\tilde{x} > 0$ belonging to this set. As we will show below, such a $\tilde{x}$ must exist (the set is closed) and moreover, this is a very rare case that arises only when the cost functions are fully zero.

\begin{lemma}
\label{applem_multipletilde}
Suppose $\exists$ multiple $x > 0$ satisfying $\frac{\lambda(x)-r(x)}{M} = x|\lambda'(x)|$. Then, it is the case that production costs must be zero at all these points. Moreover, the set of points satisfying the equation must be continuous and there must be some maximal $\tilde{x}$ belonging to this set, i.e., the set of points satisfying the equation must be closed from above.
\end{lemma}
\begin{proof}
As per our proof of existence, if the optimal solution satisfies $\lambda(x^*)-r(x^*) \geq Mx^*|\lambda'(x^*)|$, then it is an equilibrium. Therefore according to the Either-Or statement of Corollary~\ref{corr_eqconditions}, we are concerned with the $\lambda(x)-r(x) = Mx|\lambda'(x)|$ case only when $\lambda(x^*)-r(x^*) < Mx^*|\lambda'(x^*)|$. So, for the case we are interested in all $x$ satisfying the equation are smaller than $x^*$. We begin by showing that if $x$ is not unique, then the production costs must be zero at all such $x$.

\begin{lemma}
\label{lem_subuniqueness3}
Let $\tilde{x} < x^*$ be some point satisfying $\lambda(\tilde{x}) - r(\tilde{x}) = M\tilde{x}|\lambda'(\tilde{x})|$. Then, all $x < \tilde{x}$ must satisfy $\lambda(\tilde{x}) - r(\tilde{x}) > M\tilde{x}|\lambda'(\tilde{x})|$ and all $x > \tilde{x}$ must satisfy $\lambda(\tilde{x}) - r(\tilde{x}) < M\tilde{x}|\lambda'(\tilde{x})|$ as long as the demand is MPE and the production costs are non-zero.
\end{lemma}
\begin{proof}
First assume that at $x=\tilde{x}$, $\lambda'(\tilde{x}) < 0$, i.e., $\lx$ is strictly decreasing at this point. Consider $x > \tilde{x}$, we show that $\lx - r(x) < Mx|\ldx|$. Since $\lx$ is strictly decreasing at $\tilde{x}$, $\forall x > \tilde{x}, \; \lx < \lambda(\tilde{x})$. Now since $\lx \in \text{MPE}$,
\begin{equation}
\label{eqn_MPE}
\frac{\lx}{x|\ldx|} \leq \frac{\lambda(\tilde{x})}{\tilde{x}|\lambda'(\tilde{x})|}.
\end{equation}
We claim that since $r(x) \geq r(\tilde{x}) > 0$,
$$\frac{\lx - r(x)}{x|\ldx|} \leq \frac{\lx - r(\tilde{x})}{x|\ldx|} < \frac{\lambda(\tilde{x}) - r(\tilde{x})}{\tilde{x}|\lambda'(\tilde{x})|} = M$$

To see why the last inequality is strict in the above expression, first suppose that $x|\ldx| \leq \tilde{x}|\lambda'(\tilde{x})|$. In this case, the last inequality trivially holds due to equation~\ref{eqn_MPE} and the fact that $\lx < \lambda(\tilde{x})$. If this is not the case and that $x|\ldx| > \tilde{x}|\lambda'(\tilde{x})|$. Assume by contradiction that the last equality fails to hold strictly. This must mean that
$$r(\tilde{x})(\frac{1}{\tilde{x}|\lambda'(\tilde{x})|} - \frac{1}{x|\ldx|}) \geq \frac{\lambda(\tilde{x})}{\tilde{x}|\lambda'(\tilde{x})|} - \frac{\lx}{x|\ldx|} > \lambda(\tilde{x})((\frac{1}{\tilde{x}|\lambda'(\tilde{x})|} - \frac{1}{x|\ldx|}),$$
which in turn implies that $r(\tilde{x}) > \lambda(\tilde{x})$, a contradiction. This completes the proof that  $\lx - r(x) < Mx|\ldx|$ holds for all $x > \tilde{x}$.

Similarly suppose that $x < \tilde{x}$, we can show that $\lx - r(x) > Mx|\ldx|$ in in the same fashion. Now, what if $\lambda'(\tilde{x}) = 0$? The equilibrium satisfies the condition $\lambda(\tilde{x}) - r(\tilde{x}) = M\tilde{x}|\lambda'(\tilde{x})|$ which in turn equals zero. This means that $\lambda(\tilde{x}) = r(\tilde{x})$ and so $\tilde{x}$ has to be one of the solutions maximizing social welfare. However, we have assumed $\tilde{x} < x^*$ for any optimal solution and so this case is not possible.
\end{proof}

Since the production costs are zero, this means that for all such $x$, 
\begin{equation}
\label{eqn_noncosteq}
\frac{\lambda(x)}{x|\ldx|}=M.
\end{equation}

 Now it is easy to deduce (since $\lx$ is a MPE function) that the set of all such $x$ has to be continuous. The only thing left to prove is that this set is closed from above, i.e., $\exists$ a maximal $\tilde{x} > 0$ Equation~\ref{eqn_noncosteq}.

Assume by contradiction that this is not true and that for all $x \in [x_1, \tilde{x})$, Equation~\ref{eqn_noncosteq} is satisfied, for some given $x_1$. Then for every $x$ in the limit $x \to \tilde{x}^-$, $\lx = Mx|\ldx|$. We know that both $\lx$ and $\ldx$ are continuous and finite. This means that the equation must hold in the limit as well. Therefore, at $\tilde{x}$, $\lambda(\tilde{x}) = M\tilde{x}|\lambda'(\tilde{x})$, which completes the proof. $\qed$
\end{proof}

\subsection{Proof of Properties Satisfied and Uniqueness}

We have now established that for MPE demand there exists at least one $\tilde{x} > 0$ which is an equilibrium. We now show that the equilibrium from Corollary~\ref{corr_eqconditions} also satisfies several nice properties.

\begin{claim}
\label{appclm_eqproperties}
The equilibrium that we characterize in Corollary~\ref{corr_eqconditions} satisfies the following properties: Non-Trivial Pricing, Recovery of Production Costs, Pareto-Optimality and Local Dominance.
\end{claim}
\begin{proof}
The first two properties are rather trivial and follow from the pricing rule in Equation~\ref{appeqn_ourpricingrule}. Therefore, we focus on showing Pareto-Optimality and Local Dominance. As usual, let $\tilde{x}$ represent the flow at our equilibrium.

\textbf{Pareto-Optimality:} Recall that a given equilibrium $(\vec{p}, \vec{x})$ is said to be Pareto-optimal \emph{over the space of equilibria} if for any other equilibrium, the utility of at least one strategic agent (buyer or seller) strictly goes down or the utility of all the agents remains the same. We show this in cases.

\emph{Case I:} $\lambda(x) > \lambda(\tilde{x})$.

This means that $x < \tilde{x}$. Such an equilibrium cannot dominate our solution because at least one buyer whose utility is positive in our solution receives no allocation in this solution.

\emph{Case II:} $\lambda(x) = \lambda(\tilde{x})$ and $x \leq \tilde{x}$.

In this case, any buyer who loses out on an allocation had zero utility to begin with. So all buyers are indifferent. We therefore focus on the sellers. Consider the total profit made by all the sellers in this new solution, this is equal to the total payment made by the buyers minus the cost incurred and can be written as
$$\pi(\vec{p}, \vec{x}) = \lambda(x) x - C(\vec{x}) \leq \lambda(\tilde{x}) x - R(x).$$

However, we know that in our solution $\lambda(\tilde{x}) \geq r(\tilde{x})$. Since $x \leq \tilde{x}$, we have $ \lambda(\tilde{x})(\tilde{x} - x) \geq R(\tilde{x}) - R(x)$. So, it is necessary that $\pi(\vec{p}, \vec{x}) \leq \tilde{\pi}$, where the latter quantity is the profit from our solution. Therefore, the only possible cases are: i) Every seller's profit in the new solution is smaller than or equal to his profit in our solution, ii) If there is a seller whose profit strictly increases, then the profit must strictly decrease for at least one seller. This is exactly the definition of Pareto-optimality.\\

\emph{Case III:} $x > \tilde{x}$ and  $\lambda(x) = \lambda(\tilde{x})$.

In this case, $\lambda'(\tilde{x})=0$, which means that at our equilibrium $\lambda(\tilde{x}) - r(\tilde{x}) \geq M\tilde{x}|\lambda'(\tilde{x})|$. If this inequality holds strictly, we know that $\tilde{x} = x^* = T$ and so $x > T$ does not make sense here. If it holds at equality, we know that $\lambda(\tilde{x}) - r(\tilde{x}) = 0$. This means  $r(x) \geq r(\tilde{x}) = \lambda(\tilde{x})$. Considering the sum total of profit made by all sellers in the two solutions, we have

$$\lambda(\tilde{x})\tilde{x} - R(\tilde{x}) = \lambda(\tilde{x}) x - R(\tilde{x}) - r(\tilde{x})(x-\tilde{x}) \geq \lambda(x)x - R(x).$$

By the same argument as in Case II, the profit at the new solution cannot dominate our solution.

\emph{Case IV:} $x > \tilde{x}$ and  $\lambda(x) < \lambda(\tilde{x})$.

This case is slightly tricky so we will first outline the steps involved in the proof before formalizing it. In our solution, let $\tilde{P}(N_M)$ be the total price of the non-monopolies along any flow carrying path. We know this quantity is the same for all such paths and equal to the sum of the marginal costs of the non-monopolies along that path. Using the properties of equilibria, we will show that in the new solution, the price of the non-monopolies along any flow carrying path has to be smaller than $\tilde{P}(N_M)$. Moreover, since $x > \tilde{x}$, there must be at least one flow carrying path where every edge has a flow that is strictly greater than its original flow. But since the total price along this path is not larger than $\tilde{P}(N_M)$, this means that at least one non-monopoly must have the same (or smaller) price as before but a larger flow. As per lemma~\ref{cl_margcost} the profit of this monopoly must strictly decrease. This completes the sketch. We proceed via a series of small claims.

First, at our equilibrium $\lambda(\tilde{x}) - r(\tilde{x}) = M\tilde{x}|\lambda'(\tilde{x})|$. Indeed, if the LHS is strictly larger than the RHS, we know by Corollary~\ref{corr_eqconditions} that $\tilde{x} = x^*$ and moreover, $x^* = T$. So $x > x^*$ does not make sense here. %Secondly, $\lambda(x) < \lambda(\tilde{x})$. If this were not true and we had equality instead, then $\lambda'(x)$ is zero in the interval $[\tilde{x}, x^*]$ and $\lambda(\tilde{x}) - r(\tilde{x}) \geq M\tilde{x}|\lambda'(\tilde{x})|$, which suggests that $\tilde{x}$ is an optimum equilibrium, a contradiction.

\begin{claim}
For every monopoly $e$, its price in the new solution is at least $c_e(x) + x|\lambda'(x)|$.
\end{claim}
\begin{proof}
Consider the profit made by this monopoly when its price is $p$ and the price of every other monopoly $e'$ is still $p_{e'}$. Assuming that the corresponding flow is given by $x_2$, this profit is $\pi(p) = p_e x_2 - C_e(x_2)$. The monopoly should not be able to increase its price and improve the profit. First suppose that $\lambda(\tilde{x}) > \lx$. This implies that $\frac{d\pi}{dp} \leq 0$. Differentiating and substituting $p=p_e$ and $x_2 = x$ gives us, $x - \frac{p_e - c_e(x)}{|\lambda'(x)|} \leq 0$. \end{proof}

Define $rm(x)$ to be the total marginal cost incurred by all $M$ monopolies when there is a flow of $x$ through them. So, the total price of all monopolies is at least $rm(x) + Mx|\lambda'(x)|$. Therefore, the total price of the non-monopolies along any flow carrying $s$-$t$ path is no larger than $\lambda(x) - rm(x) - Mx|\lambda'(x)|$. We now show that this quantity is also smaller than $\tilde{P}(N_M)$, the total price of the non-monopolies in our solution.

\begin{claim}
For any flow carrying $s$-$t$ path in the new solution, the total price of non-monopolies on that path is strictly smaller than $\tilde{P}(N_M) = \lambda(\tilde{x}) - rm(\tilde{x}) - M\tilde{x}|\lambda'(\tilde{x})|$.
\end{claim}
\begin{proof}
We only have to show that,
$$\lambda(x) - rm(x) - Mx|\lambda'(x)| < \tilde{P}(N_M) = \lambda(\tilde{x}) - rm(\tilde{x}) - M\tilde{x}|\lambda'(\tilde{x})|.$$

Since $\lambda(x)$ is a MPE function, we know that $\frac{\lambda(\tilde{x})}{M\tilde{x}|\lambda'(\tilde{x})|} \geq \frac{\lx}{Mx|\lambda'(x)|}$. The left hand side of this inequality is at least one since $\lambda(\tilde{x}) = r(\tilde{x}) + Mx|\lambda'(\tilde{x})|$. By the exact same reasoning as in Lemma~\ref{sublem_exist_largex}, we can obtain that

$$\lambda(\tilde{x}) - M\tilde{x}|\lambda'(\tilde{x})| \geq \lambda(x) - Mx|\lambda'(x)|.$$

In fact, we claim that the above inequality is strict. This is not hard to see: since $\lambda(\tilde{x}) > \lambda(x)$, equality will arise only when $\frac{\lambda(\tilde{x})}{M\tilde{x}|\lambda'(\tilde{x})|} = \frac{\lx}{Mx|\lambda'(x)|} = 1$. This means that the production cost functions are zero at both these points. However, as we remarked at the end of Corollary~\ref{corr_eqconditions} (also see Lemma~\ref{applem_multipletilde}), we have already taken $\tilde{x}$ to be the largest flow satisfying $\frac{\lambda(\tilde{x})}{M\tilde{x}|\lambda'(\tilde{x})|} = 1$, which contradicts $x > \tilde{x}$. Therefore, it is true that
$$\frac{\lambda(\tilde{x})}{M\tilde{x}|\lambda'(\tilde{x})|} > \frac{\lx}{Mx|\lambda'(x)|},$$

and $\lambda(\tilde{x}) - M\tilde{x}|\lambda'(\tilde{x})| > \lambda(x) - Mx|\lambda'(x)|$. Since $rm(\tilde{x}) \leq rm(x)$, we get the desired result. \end{proof}

Since $x > \tilde{x}$, it is not hard to show by some simple properties that there exists at least one $s$-$t$ path where every edge has a flow strictly greater than its flow in $(\tilde{x})_e$. Let this path be $P_1$ and let the sum of prices of non-monopolies along this path be $P(N_M)$. We know that
$$P(N_M) < \tilde{P}(N_M) = \sum_{e \in P_1 \setminus \mathcal{M}} c_e(\tilde{x}_e).$$
Then, there is an edge with price strictly smaller than $c_e(\tilde{x}_e)$ and flow strictly larger than $\tilde{x}_e$. By similar reasoning as in Lemma~\ref{cl_margcost}, the profit of this edge has to be strictly smaller than in our solution.

\textbf{Local Dominance:} It is rather easy to show that our solution is locally dominant. Suppose that a small population ($\epsilon$) of buyers shift their flow from one path to another. Then, for any monopoly seller, the total flow does not change so such a seller is indifferent to the perturbation. Consider any non-monopoly seller $e$ who is priced at $\tilde{p}_e = c_e(\tilde{x}_e)$. The total flow on this edge either increases or decreases by $\epsilon$. Suppose it increases, then we know from Lemma~\ref{cl_margcost} that the profit on the edge cannot increase because the cost increases at a larger rate. Suppose the total flow decreases, then the change in profit is
$$(C_e(\tilde{x}_e) - C_e(\tilde{x}_e - \epsilon)) - \tilde{p}_e \epsilon \leq c_e(\tilde{x}_e)\epsilon - \tilde{p}_e \epsilon$$

The above quantity is not positive since $c_e(\tilde{x}_e) = \tilde{p}_e$. $\qed$
\end{proof}

We first prove uniqueness and then show that the equilibrium that we are interested can be computed efficiently.

\begin{thm_app}{thm:uniqueness}
For any given instance with strictly monotone MPE demand and non-zero production costs, we are guaranteed that one of the following is always true,
\begin{enumerate}
\item There is a unique equilibrium that obeys Local Dominance and Non-Trivial Pricing.
\item All equilibria that satisfy Local Dominance and Non-Trivial Pricing maximize welfare.
\end{enumerate}
\end{thm_app}
\begin{proof}
Let $(\vec{p}, \vec{x})$ be any equilibrium that obeys both non-trivial pricing and local dominance. We show via a series of simple lemmas that this equilibrium is unique and moreover, has to satisfy the conditions of Corollary~\ref{corr_eqconditions}.

\begin{lemma}
\label{sublem_u1}
Any non-monopoly edge $e$ with $p_e > 0$ must have $x_e < x$.
\end{lemma}
This just means that no non-monopoly edge can monopolize the flow in an equilibrium.
\begin{proof}
The proof is similar to that of Claim~\ref{clm_monopolypos1} so we only sketch the details. Assume by contradiction that such an edge is carrying the full flow. We know that the price on this edge at equilibrium is strictly non-zero.  Now, we also know that every edge without flow is priced at $c_{e'}(0) = 0$. As with Claim~\ref{clm_monopolypos1}, we can show that for some predecessor $u$ and successor $v$ of $e$, there must exist a $u$-$v$ path not containing $e$ with zero flow. However, this means that the path is priced at $0$ and that the buyers are not taking the cheapest $u$-$v$ path, a contradiction of the equilibrium.
\end{proof}

\begin{lemma}
\label{sublem_u2}
For every non-monopoly edge $e$, its price $p_e$ is exactly $c_e(x_e)$.
\end{lemma}
\begin{proof}
The proof follows from Local Dominance. Suppose that a small number $(\epsilon)$ of buyers shift their flow from one path to another. Consider any non-monopoly edge $e$ with $x_e < x$: if its price is strictly larger than $c_e(x_e)$, then for some sufficiently small $\epsilon$, its profit would strictly increase at the new flow since $p_e \epsilon > C_e(x_e + \epsilon) - C_e(x_e)$. Similarly, if its price is strictly smaller than $c_e(x_e)$, then it can lose some buyers and improve its profit as per Claim~\ref{cl_margcost}.

Now, look at any non-monopoly edge which has $x_e=x$. By the proof of Lemma~\ref{sublem_u1}, it is clear that this edge must be priced at $0$. We also claim that $c_e(x) = 0$. Suppose that this is not true, then the edge is facing non-zero production costs as well giving it a negative profit. Such a solution cannot be an equilibrium because then $e$ would rather increase its price and lose the flow.\end{proof}

\begin{lemma}
\label{sublem_u3}
$\vec{x}$ has to be a minimum cost flow of magnitude $x$.
\end{lemma}
\begin{proof}
Since this is an equilibrium, every flow carrying path $P$ must have the exact same price, $\lx$. Moreover, for any other path $P'$, its price cannot be strictly smaller than $\lx$. Recall that the monopolies ($\mathcal{M}$) are common to every $s$-$t$ path and every non-monopoly edge $e$ has a price $p_e = c_e(x_e)$. Let $p_M$ be the sum of prices of the monopoly edges. Assume by contradiction that this is not a min-cost flow. Then, by the KKT conditions, there must exist a flow carrying path $P$ and some other path $P'$, such that

$$\sum_{e \in P}c_e(x_e) > \sum_{e \in P'}c_e(x_e).$$
Since the monopolies are common to both the paths, we can remove the terms corresponding to $\mathcal{M}$ from both sides. But recall that for the non-monopoly edges, the price is its marginal cost. This gives us $\lx - p_M > p(P') - p_M$ or $\lx > p(P')$, where $p(P')$ is the total price on the path $P'$. But this is a contradiction since $\lx$ should be smaller or equal to the price on any other path. This completes the proof.
\end{proof}

\begin{proposition}
\label{prop_eqlessthanopt}
Let $x^*$ denote the optimal flow. Then, $x \leq x^*$.
\end{proposition}
\begin{proof}
Without loss of generality, assume that the optimal solution is unique. Suppose that by contradiction that this proposition is not true. Then, clearly, $\lambda(x) < r(x)$. However, at an equilibrium, every monopoly has to have a price that is at least its marginal. Indeed, since we assumed that the demand is strictly monotone, it means that a monopoly priced below its marginal can always increase its price and improve profits as the flow decreases. We already know that non-monopolies cannot be priced below the marginal. This means that the sum of prices along any path is at least the total marginal cost $r(x)$, which is greater than $\lx$, a contradiction.
\end{proof}

\begin{lemma}
\label{sublem_u4}
At any sub-optimal equilibrium ($x < x^*$), every monopoly's price $p_e = c_e(x) + x|\ldx|$.
\end{lemma}
\begin{proof}
Define a profit function for some monopoly $e$, which represents the profit made by that monopoly when its price is $p$ and every other edge is priced according to $\vec{p}$. Mathematically, $\pi(p) = p\bar{x} - C_e(\bar{x})$, where $\bar{x}$ is the resulting flow when the monopoly changes its price to $p$ from $p_e$. Since this is an equilibrium and it is sub-optimal, $\frac{d\pi}{dp} = 0$ at $p=p_e$, i.e, no monopoly can either increase or decrease its price from the equilibrium point. Differentiating gives us the desired result.
\end{proof}

We proceed to show the first half of the theorem, that if the equilibrium guaranteed by  Corollary~\ref{corr_eqconditions} is sub-optimal, then it is unique. We know from the above three lemmas (\ref{sublem_u2}, \ref{sublem_u3}, \ref{sublem_u4}) that the price of every flow carrying path is exactly $r(x) + Mx|\lambda'(x)| = \lx$. This means that the equilibrium has to satisfy the conditions of Corollary~\ref{corr_eqconditions}. Since production costs are non-zero, we know from Lemma~\ref{lem_subuniqueness3} that there is a unique point satisfying $\lambda(\tilde{x}) - r(\tilde{x}) = M\tilde{x}|\lambda'(\tilde{x})|$.

This completes the proof of the first half of the theorem. Now, suppose that Corollary~\ref{corr_eqconditions} gives us an optimal equilibrium. We need to show that no sub-optimal equilibrium that satisfies non-trivial pricing and local dominance exists for this instance. First, since there is an optimal equilibrium, we know that $\frac{\lambda(x^*) - r(x^*)}{x^*|\lambda'(x^*)|} \geq M.$ By the MPE property, we also know that for all $x < x^*$,

$$\frac{\lx - r(x)}{x|\ldx|} \geq \frac{\lambda(x^*) - r(x^*)}{x^*|\lambda'(x^*)|} \geq M.$$ Now assume by contradiction that at the equilibrium $x < x^*$. Since all non-monopolies are priced at $c_e(x_e)$, at least one monopoly edge should be priced at $c_e(x) + \frac{\lx - r(x)}{M}$ or more. We claim that this edge can lower its price and increase profits.

More concretely, suppose that the monopoly edge whose price we lower bounded is $e$ and define $\pi(p)$ to be the profit of that monopoly when its price is $p$ and other edges don't change their price. Look at the derivative of the profit function at $p=p_e$,
$$\left(\frac{d\pi}{dp}\right)_{p=p_e} = x - \frac{p_e - c_e(x)}{|\lambda'(x)|} \leq x - \frac{\lx - r(x)}{M|\ldx|}.$$

The above term is not clearly not positive since $\frac{\lx - r(x)}{M|\ldx|} \geq x^* > x$, contradicting the equilibrium.$\qed$\end{proof}

\subsubsection*{Computing the equilibrium}
We now show that an equilibrium satisfying the conditions of Corollary~\ref{corr_eqconditions} can be computed efficiently up to any required precision. 

\begin{theorem} A Nash equilibrium from Theorem \ref{thm:existence} can be computed by a simple binary search procedure for any $\lambda\in\text{MPE}$.
\end{theorem}
\begin{proof}
The first step in our procedure involves computing the optimal flow using a convex program and checking if it satisfies the condition of Lemma~\ref{lem_eqnconditionsincr}, i.e., $\lambda(x^*) - r(x^*) \geq Mx^*|\lambda'(x^*)|$. If that is the case, then by Corollary~\ref{corr_eqconditionssuff}, we are done. The following algorithm therefore, is for the case when $\lambda(x^*) - r(x^*) < Mx^*|\lambda'(x^*)|$.

First, suppose that production costs are non-zero. According to Lemma~\ref{lem_subuniqueness3}, since the equilibrium is not optimal, it is a unique point $\tilde{x} < x^*$ satisfying $\lambda(\tilde{x}) - r(\tilde{x}) = M\tilde{x}|\lambda'(\tilde{x})|$. Furthermore, by Lemma~\ref{lem_subuniqueness3}, it is true that for all $x > \tilde{x}$, $\lambda(\tilde{x}) - r(\tilde{x}) < M\tilde{x}|\lambda'(\tilde{x})|$ and $\forall x < \tilde{x}$, $\lambda(\tilde{x}) - r(\tilde{x}) > M\tilde{x}|\lambda'(\tilde{x})|$. This motivates a binary search approach.

The algorithm is trivial. We maintain a window $[x_1, x_2]$ at every iteration. Initially set $x_1=0$ and $x_2=x^*$. Let $x=\frac{1}{2}(x_1+x_2)$, and compute $\alpha(x)=\lx-r(x) - Mx|\ldx|$. If $\alpha(x) > 0$, update the window to $[x,x_2]$ and if $\alpha(x) < 0$, update the window to $[x_1, x]$. If $\alpha=0$, then we are done. Repeat this until the difference between the two extreme points falls within the desired precision and return (say) the mid-point. Since the size of the window is halved at every point, the algorithm terminates in time proportional to $O(\log(T))$, where $T$ is the total population of buyers in the market.

What if production costs are zero at equilibrium? There may be a set of points satisfying $\lambda(\tilde{x}) - r(\tilde{x}) = \lambda(\tilde{x}) = M\tilde{x}|\lambda'(\tilde{x})|$. But we know from Lemma~\ref{applem_multipletilde} that this set of points is continuous and that there is a maximal $\tilde{x}$ belonging to this set (the set is closed from above). So for every $x < \tilde{x}$,  $\lambda(\tilde{x}) - r(\tilde{x}) = \lambda(\tilde{x}) \leq M\tilde{x}|\lambda'(\tilde{x})|$ and the inverse of this is true for $x > \tilde{x}$. So, we can use a similar binary search procedure once again.
\end{proof}

\section{Proofs from Section 4: Effects of Demand Curves and Monopolies on Efficiency}
\label{app:pos}

\subsection*{Example with one monopoly where Walrasian Equilibirum is not (Nash) stable}

\begin{example}{}
\label{ex:1}
Consider a single good controlled by a single seller who can produce $x$ amount of this good at a cost of $x^2$. Let there be a large number of infinitesimal buyers in the market who in total desire one unit of this good, such that a fraction $(1-p)$ of the buyers value the good at price $p$ or more. The unique Walrasian price is $p=2/3$ where exactly one-third of the buyers purchase the good and the seller produces the same amount.
The seller's profit is then $p*1/3 - (1/3)^2 = 1/9$. However, suppose that the seller increases her price to $p'=3/4$: the demand drops to one-fourth but the seller's profit is now $p'*1/4 - (1/4)^2 = 1/8$, which is strictly larger than the original profit. So, if the seller can increase her price and anticipate the resulting demand, then she stands to benefit by breaking the Walrasian Equilibrium.
\end{example}

\subsection{Uniform Buyer Demand}
We begin with the simplest buyer demand function $\lx = \lambda_0$ for $x\leq T$. The function is defined for $x=0$ to $T$ and continuity and differentiability hold in one direction at $x=T$, so it trivially belongs to the class MPE and existence is guaranteed. We are able to show that Nash Equilibrium is efficient for this case. It is important to mention here that the Nash Equilibrium that we construct may {\em not} necessarily be a Walrasian Equilibrium although the allocations are the same. Consider the simplest example with just one good, $\lambda_0 = 1$, $T=1$ and $C_e(x)=\frac{x^2}{4}$. The Walrasian equilibrium has a price of $p_W = 0.5$ on the edge and all buyers purchase the bundle. On the other hand, any (efficient) Nash Equilibrium must have a price of $p_N = 1$ on the edge and the total flow is still $\tilde{x} = 1$.\\

\begin{thm_app}{theorem_linearutilitymain}
Every instance with uniform demand buyers admits an efficient Nash Equilibrium. Moreover, this holds even when $\exists$ edges with $c_e(0) > 0$.
\end{thm_app}
Our proof relies on the fact that if the instance contains monopolies, then all the bundles with non-zero flow on them must have a total price equaling $\lambda_0$ and the bundles with zero flow on them must have a price not less than $\lambda_0$. 

\begin{proof}
If there exist several optimum solutions, then denote by $\vec{x^*}$ the optimum solution with the maximum possible flow. In other words, any solution with a total flow of size $x' > x^*$ is non-optimal. We show that there exist a set of prices such that this optimal flow forms a Nash Equilibrium. We begin with the case where $c_e(0) = 0$ for all edges. For this case, we claim that our pricing rule always returns a Nash Equilibrium.

First recall that our pricing rule ensures all flow carrying paths are priced at exactly $\lambda(x^*) = \lambda_0$, which for uniform demand is the value for the paths held by all buyers. This means that if any edge increases its price, all paths containing that edge would now have a price strictly larger than $\lambda_0$. No buyer would be able to afford such a path and therefore, the flow on the edge would drop to zero and its profit cannot increase. So we only need to show that sellers cannot decrease their price and improve profits.

Also recall from Lemma~\ref{cl_margcost} that edges priced at their marginal cost can never decrease their price and increase profits whatever be the resulting flow. This means that we should only worry about edges which have a price strictly larger than their marginal cost. The only edges for which this is possible are the monopolies, which are priced at $p_e = c_e(x^*) + \frac{\lambda(x^*)-r(x^*)}{M}$. For $p_e > c_e(x^*)$, we can conclude that $\lambda(x^*) > r(x^*)$ if any edge is to be priced above its marginal cost. However, by Proposition~\ref{prop_optflow}, this implies that $x^* = T$ and therefore at the optimum flow all the buyers (population of $T$) have non-zero allocation. In such a case, decreasing prices would have no effect since there are no more unallocated buyers left in the market. We therefore conclude that no edge can increase or decrease its price and thus for the $c_e(0) = 0$ case, our pricing rule returns a Nash Equilibrium with flow $x^*$.

The proof for the $c_e(0)>0$ case is slightly different. There will now be some additional edges monopolizing the flow at the optimum; we call them Virtual Monopolies (VM) (see also Section~\ref{sec:generalizations}). Formally, given a solution $\vec{x^*}$, VMs are the edges which contain the total flow $x^*$, but may not be contained in all $s-t$ paths. The set of VMs includes the pure monopolies but also edges monopolizing flow at $x^*$.

For this case, we define a slightly modified process to obtain prices on all the edges and then show that these prices are stable for the optimal flow. Begin by pricing each edge at its marginal price, $c_e(x_e)$. We now show how to increase the price on only the virtual monopolies such that at each (set of) prices, the optimal flow remains a best-response. Let $P$ be the set of $s-t$ paths with positive flow and $P'$ without flow. We know by the properties of the optimal flow (Lemma \ref{lemma_nonzeroflowcost} and \ref{lemma_cequivalence}) that the price on any $P_i \in P$ is initially $p_0=r(x^*)$ and no path in $P'$ can have a larger price.

The algorithm proceeds as follows: at each stage pick a VM (say $e$). Increase its price until either the total price of all $s-t$ paths with flow is now $\lambda_0$ or until we reach a point where increasing its price anymore would mean that the paths in $P$ are no longer the minimum price paths, i.e., if all paths in $P$ have a price of $p_t$, then $\exists$ a path in $P'$ not containing $e$ and with a total price of $p_t$. Note that since $e$ belongs to all $P_i \in P$, increasing the price on $e$ leads to an equal increase on the prices of the paths in $P$. And at every stage, the paths in $P$ still remain the min-price paths so buyer behavior is a best-response.

Now, at the end of this process, we have reached a stage where either all paths in $P$ have a price $\lambda_0$ or no virtual monopoly can increase its price without losing all its flow to some path in $P'$. If the price on all paths is $\lambda_0$, then in that case no edge can increase its price since this is the maximum price that buyers will pay. So no edge can increase its price, virtual monopoly or otherwise.

Finally, as with our previous case the only edges which can decrease their price and `potentially' make a larger profit are those priced at $p_e > c_e(x^*_e)$. The existence of such edges would imply that $\lambda(x^*) > r(x^*)$, since if $\lambda_0 = \lambda(x^*) = r(x^*)$, then our process would terminate trivially at the first iteration since we cannot increase the price on any edge without losing the flow. For this case, no VM will decrease its price because there are no more unallocated buyers left, and so it would not gain any new flow by decreasing its price. This completes the proof. \end{proof}

\subsection{Monotone Hazard Rate}
Our main result in this section is showing a tight bound on the efficiency of markets where the (inverse) demand has a monotone hazard rate. We now formally define the class $MHR$ and show that it is equivalent to log-concave functions.

\begin{definition}{\emph{Monotone Hazard Rate.}}
An aggregate (decreasing) function $F(x)$ is said to follow a monotone hazard rate if the hazard rate $h(x)=\frac{-\frac{d}{dx}(F(x))}{F(x)}$ is a non-decreasing function of $x$.
\end{definition}

\begin{definition}{\emph{Class MHR.}}
An inverse demand function is said to belong to the class $MHR$ if it has a monotone hazard rate, i.e., $\displaystyle \frac{-\ldx}{\lx} = \frac{|\ldx|}{\lx}$ is a non-decreasing function of $x$.
\end{definition}

Although the aggregate function $F(x)$ is usually normalized so that its maximum value is $1$, the same definition can be extended to a scaled version of $F(x)$ since the scaling factor appears both in the function and its derivative. Thereby, the definition carries over to our inverse demand functions.

%\begin{proposition}
%The class $MHR$ of inverse demand functions is equivalent to the set of log-concave inverse demand functions.
%\end{proposition}
%\begin{proof}
%A function $\lx$ is said to be log-concave if $H(x) = \log(\lx)$ is concave. And $H(x)$ is concave if the absolute value of its derivative is non-decreasing. In this case, $|h(x)| = -\frac{d}{dx}H(x) = \frac{-\ldx}{\lx}$. This is true because $\lx$ is non-increasing so its derivative is never positive and thereby its absolute value is $-\ldx$. $h(x)$ corresponds to the hazard rate of the function. So we conclude that all log-concave functions belong to MHR. The other direction can be shown similarly.
%\end{proof}

\textbf{Interpretation.} Commonly used inverse demand (other than concave) functions belonging to this class include exponential functions similar to $\lx = e^{-x}$~\cite{lyoo2006efficient} and relatively inelastic functions of the form $\lx = (a-x)^{\alpha}$ for $\alpha > 1$ (refer Section~\ref{app:inverse_demand} for more details). While uniformly decreasing demand (linear inverse demand, $\lx = a-x$) has been assumed more commonly due to its tractable nature, it is more likely that the elasticity of demand is not constant across different prices. Indeed different segments of the market may react differently to a change in price. MHR functions capture a very interesting class of such functions where the responsiveness of a market relative to the value of the buyers is non-decreasing. More concretely if at some price $2p$, an increase of $dp$ in the price leads to a reduction in $dx$ number of buyers. Then at a price $p$, in order to make the same number of consumers ($dx$) drop out, the increase in price has to be at least $\frac{1}{2}dp$. In simple terms, the market cannot be `overly sensitive' at smaller prices compared to its sensitivity at a larger price.

\begin{proposition}
Any inverse demand function which is concave ($\ldx$ is non-increasing) belongs to the class $MHR$.
\end{proposition}
For a concave function, $|\ldx|$ is non-decreasing and $\lx$ is non-increasing, the ratio $\frac{|\ldx|}{\lx}$ has to be non-decreasing as well.

\begin{lemma}
\label{lem_mhrcost}
Let $\lambda(x)$ be any inverse demand function satisfying Monotone Hazard Rate (i.e. $\frac{|\lambda'(x)|}{\lambda(x)}$ is non-decreasing). Given an instance specifying a graph $G$ and cost functions $C_e(x)$, then the function $\frac{|\lambda'(x)|}{\lambda(x) - r(x)}$ is also non-decreasing $\forall x \leq x^*$, where $x^*$ is the size of the optimum flow for that instance.
\end{lemma}
\begin{proof} Recall from Proposition~\ref{rcontinuous} that for a given graph and convex cost functions, $r(x)$ is non-decreasing in $x$.
Also recall that as long as $x \leq x^*$, $\lx \geq r(x)$, since $\lambda(x^*)=r(x^*)$ and $\lambda$ is non-increasing.  Consider $x_1 \leq x_2$. Since $\lambda(x)$ is MHR, we know $\frac{|\lambda'(x_1)|}{\lambda(x_1)} \leq \frac{|\lambda'(x_2)|}{\lambda(x_2)}$, which implies $|\lambda'(x_1)|\lambda(x_2) \leq |\lambda'(x_2)|\lambda(x_1)$. Let us consider two cases\\
\textbf{Case I:} $|\lambda'(x_1)| > |\lambda'(x_2)|$\\
We need to show
\begin{align}
& \frac{|\lambda'(x_1)|}{\lambda(x_1) - r(x_1)} &\leq & \frac{|\lambda'(x_2)|}{\lambda(x_2) - r(x_2)}\\
\Longleftrightarrow & |\lambda'(x_1)|\lambda(x_2) & \leq & |\lambda'(x_2)|\lambda(x_1) + |\lambda'(x_1)|r(x_2) - |\lambda'(x_2)|r(x_1)
\end{align}
But we already know that $|\lambda'(x_1)|\lambda(x_2)  \leq  |\lambda'(x_2)|\lambda(x_1) + |\lambda'(x_1)|(r(x_2) - r(x_1))$, where the term on the LHS is less than the first term on the RHS due to the MHR assumption and the second term on the RHS is always positive as $r(x)$ is non-decreasing. Since $|\lambda'(x_1)| > |\lambda'(x_2)|$, we know $|\lambda'(x_1)|(r(x_2) - r(x_1)) \leq |\lambda'(x_1)|r(x_2) - |\lambda'(x_2)|r(x_1)$, which gives us the required result.\\
\textbf{Case II:} $|\lambda'(x_1)| \leq |\lambda'(x_2)|$
In this case, we have
$$\frac{|\lambda'(x_1)|}{\lambda(x_1) - r(x_1)} \leq  \frac{|\lambda'(x_2)|}{\lambda(x_1) - r(x_1)} \leq \frac{|\lambda'(x_2)|}{\lambda(x_2) - r(x_2)}.$$
The last step comes from the fact that $\lambda(x_2)-r(x_2) \leq \lambda(x_1)-r(x_1)$ since $\lambda(x)$ is a non-increasing function and $r(x)$ is non-decreasing.
\end{proof}

%\begin{corollary}
%\label{corr_lconcave}
%If $\lx$ is a concave decreasing function, then $\lx - c^{\alpha}(x)$ is also concave decreasing.
%\end{corollary}
%Proof is essentially Case II of the lemma's proof.

We can now show our main result. We prove the bound on MHR functions here and show the bound for concave inverse demand in the following sub-section.\\

\begin{thm_app}{thm:PoSbody}
The social welfare of Nash equilibrium from Section \ref{sec:existenceBody} is always within a factor of:
\begin{itemize}
\item $1+M$ of the optimum for Log-concave (i.e., MHR) $\lambda$ and this factor is tight.
\end{itemize}
\end{thm_app}

\begin{proof}
Let $\vec{x^*}$ be the optimum solution with a total flow of magnitude $x^*$. We use $\vec{\tilde{x}}$ to denote the Nash Equilibrium with flow of size $\tilde{x}$ as guaranteed by our existence result and which obeys the conditions of Corollary~\ref{corr_eqconditions}. Recall from our algorithm that among all feasible flows of magnitude $\tilde{x}$, $\vec{\tilde{x}}$ is the flow with minimum cost, and that $\tilde{x}\leq x^*$. Tthe social welfare of the optimum flow is $\Lambda(x^*)-R(x^*)$, and the social welfare of our equilibrium solution is $\Lambda(\tilde{x})-R(\tilde{x})$. Thus the efficiency ($\eta$) is given by,
\begin{align}
\frac{\int_{0}^{x^*}[\lx - r(x)]dx}{\int_{0}^{\tilde{x}}[\lx - r(x)]dx} & = &\frac{\int_{0}^{\tilde{x}}[\lx - r(x)]dx}{\int_{0}^{\tilde{x}}[\lx - r(x)]dx} &+& \frac{\int_{\tilde{x}}^{x^*}[\lx - r(x)]dx}{\int_{0}^{\tilde{x}}[\lx - r(x)]dx}\\
\label{eqn_final_1}&= & 1 &+& \frac{\int_{\tilde{x}}^{x^*}[\lx - r(x)]dx}{\int_{0}^{\tilde{x}}[\lx - r(x)]dx}
\end{align}

%In the above steps, we just split up the integral in the numerator from $0$ to $x^*$ into two parts, from $0$ to $\tilde{x}$ and $\tilde{x}$ to $x^*$. Now,
We know from Lemma~\ref{lem_mhrcost} that $\frac{|\lambda'(x)|}{\lambda(x) - r(x)}$ is non-decreasing in $x$. So $\forall x \geq \tilde{x}$, we have $\frac{\lx -r(x)}{|\ldx|} \leq \frac{\lambda(\tilde{x}) - r(\tilde{x})}{|\lambda'(\tilde{x})|}$. Equivalently, this implies that $\lx-r(x) \leq |\ldx| \frac{\lambda(\tilde{x}) - r(\tilde{x})}{|\lambda'(\tilde{x})|}$.

The term following $|\ldx|$ is a constant with respect to $x$. So substituting this in the integral in the numerator, we have,
$$\eta \leq 1 + \frac{\lambda(\tilde{x}) - r(\tilde{x})}{|\lambda'(\tilde{x})|}\frac{\int_{\tilde{x}}^{x^*} |\ldx|}{\int_{0}^{\tilde{x}}[\lx - r(x)]dx}.$$

From Corollary~\ref{corr_eqconditions}, we know that $\lambda(\tilde{x}) - r(\tilde{x}) = M\tilde{x} |\lambda'(\tilde{x})|$ at the equilibrium point $\tilde{x}$. Substituting this in the above upper bound for $\eta$, we can further simplify it as
$$\eta \leq 1 + M\tilde{x}\frac{\int_{\tilde{x}}^{x^*} |\ldx|dx}{\int_{0}^{\tilde{x}}[\lx - r(x)]dx}.$$

Now, look at the denominator in the above expression. $\lx-r(x)$ is a non-increasing function of $x$. So, we can lower bound $\int_{0}^{\tilde{x}}[\lx - r(x)]dx$ by $\tilde{x}(\lambda(\tilde{x}) - r(\tilde{x}))$. 

Finally, we can bound the efficiency as follows,
\begin{align*}
\eta & \leq & 1 + M\tilde{x}\frac{\int_{\tilde{x}}^{x^*}|\ldx|dx}{\int_{0}^{\tilde{x}}[\lx - r(x)]dx} & \\
& \leq & 1 + M\tilde{x}\frac{\int_{\tilde{x}}^{x^*}|\ldx| dx}{\tilde{x} (\lambda(\tilde{x}) - r(\tilde{x}))} & \\
& = & 1 + M\frac{\int_{\tilde{x}}^{x^*}-\ldx dx}{\lambda(\tilde{x}) - r(\tilde{x})} & \; \text{ (Since $\ldx$ is negative)}\\
& = & 1 + M\frac{\lambda(\tilde{x}) - \lambda(x^*)}{\lambda(\tilde{x}) - r(\tilde{x})}\\
& \leq & 1+M &
\end{align*}
The last step comes from the fact that $r(\tilde{x}) \leq r(x^*) \leq \lambda(x^*)$ since at the optimum, $\lambda(x^*) \geq r(x^*)$ and we know from Proposition~\ref{prop_mincostfunction} that $r$ is non-decreasing. We now show that this bound is tight.

\textbf{Tight Example for MHR functions:}\\
As is standard with such examples, we consider a function at the boundary of MHR, $\lx = e^{-x}$ for $x\geq \frac{1}{M}$ and $\lx =e^{-1/M}$ for $0 \leq x \leq 1/M$. Note that this function is not differentiable at $1/M$, although it still obeys our more general definition of MHR functions for piecewise-differentiable functions (the hazard rate is 1 for $x>1/M$ and 0 for $x<1/M$, and thus non-decreasing). Alternatively, we can consider a ``smoothed" version of this function which behaves the same way outside of the neighborhood of $1/M$ but is continuously differentiable in that neighborhood as well; the result still holds for such a function.

There are $M$ sellers forming a path of length $M$ between the source and the sink. The production costs for all sellers are identical and are very close to 0 for $x$ smaller than $x^*$, for some large $x^*>1/M$, at which point the production costs increase very rapidly, so that $r(x^*)=\lambda(x^*)$, and thus $x^*$ is the size of the flow at optimum. The social welfare is $\int_0 ^{x^*} \lambda(x)dx$ since the production cost is essentially zero; this welfare converges to $\frac{M+1}{M}e^{-1/M}$ as $x^*$ is large.

%Thus, as $x^*$ becomes larger, the welfare of the optimal solution approaches $\Lambd
%We assume that each seller has zero production cost. First, it is not hard to see that this demand function belongs to the class $MHR$ even though its derivative is not defined at $x=\frac{1}{M}$. First for $x < \frac{1}{M}$, the hazard rate $h(x) = 0$ since the function is a constant. For $x > \frac{1}{M}$, the hazard rate $h(x) = \frac{1}{e^{-x}}|\frac{d}{dx}(e^{-x})| = 1$. The hazard rate is non-decreasing and $\lim_{x \to \frac{1}{M}^-}h(x) \leq \lim_{x \to \frac{1}{M}^+}h(x)$, so the function is in MHR.
%At the optimum solution, the graph admits the entire flow of $x^*$ since there is no production cost.

Now consider any non-trivial equilibrium (with non-zero flow $\tilde{x})$ with prices $\vec{p}$. First we claim that $\tilde{x} \geq \frac{1}{M}$. Indeed if $\tilde{x} < \frac{1}{M}$, then the total price of the path equals $\lambda(\tilde{x})=e^{-1/M}$. Any edge can decrease its price infinitesimally and receive a flow of $\frac{1}{M}$ and increase its profit (since there are $\frac{1}{M}$ buyers who value the path at least $e^{-1/M}$). 

Next consider the possibility of equilibrium at some point where $\lambda'(x)$ is defined. By Lemma~\ref{lemma_eqnconditionsgen}, when $\lambda'_+$ and $\lambda'_-$ are equal, the equilibrium condition reduces to every monopoly edge satisfying $p_e + \tilde{x}\lambda'(\tilde{x}) =0$. Further simplification tells us that if there exists such an equilibrium, it must satisfy $\lambda(\tilde{x}) = M\tilde{x}|\lambda'(\tilde{x})|$. But there is no such point satisfying this when $x > \frac{1}{M}$.

Now consider $\tilde{x} = \frac{1}{M}$. Can this point be an equilibrium? Once again by the conditions of Lemma~\ref{lemma_eqnconditionsgen}, every edge must satisfy $p_e \leq \tilde{x}|\lambda'_+(\tilde{x})| = \frac{e^{-M}}{M}$. So the sum of edge prices is no larger than $e^{-M}$. But in order to induce a flow of $\tilde{x} = \frac{1}{M}$, the total price of path must be exactly $e^{-M}$. So we conclude that at the only possible equilibrium point, all edges have an equal price of $\frac{e^{-M}}{M}$. Clearly, sellers cannot increase their prices at this point. This point also satisfies the conditions of Theorem~\ref{cl_discolambdaex} so it is not hard to show that every buyer's profit is maximized at this price and thus we have an equilibrium.

The social welfare at $\tilde{x}$ is $\frac{e^{-1/M}}{M}$, so the efficiency tends to $1+M$. For this family of instances, this is the unique non-trivial equilibrium.
\end{proof}

\subsubsection{Concave Demand Functions}
We now show that the efficiency for $\lx$ being a concave function is $1+\frac{M}{2}$. Recall that an inverse demand function $\lx$ is said to be concave if its derivative $\ldx$ is decreasing or equivalently $|\ldx|$ is increasing. Although concave functions belong to the family $MHR$, they have a strictly increasing hazard rate, which is why we obtain an improved bound on the efficiency. The efficiency when the market demand is concave has received some attention before~\cite{anderson2003efficiency,johari2005efficiency}, albeit for slightly different models. 

\begin{thm_app}{thm:PoSbody}
\textbf{(Part 2)} The social welfare of Nash equilibrium from Section \ref{sec:existenceBody} is always within a factor of:
\begin{itemize}
\item $1+\frac{M}{2}$ of the optimum for concave $\lambda$ and this bound is tight.
\end{itemize}
\end{thm_app}
\begin{proof}
The proof uses the same notation as the proof MHR functions and we build upon some of the arguments made there. We know that the following is an upper bound on the efficiency,
$$1 + \frac{\int_{\tilde{x}}^{x^*}[\lx - r(x)]dx}{\int_0^{\tilde{x}}[\lx - r(x)]dx}.$$

\begin{figure}[htbp]
\centering
\includegraphics[scale=1.3]{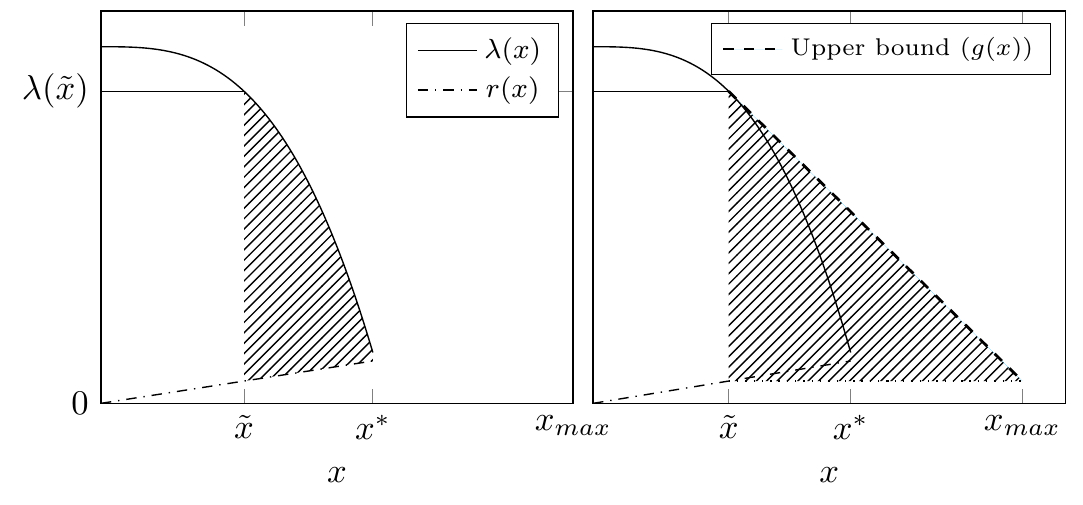}
\caption{Plot of an arbitrary concave inverse demand function}
\label{fig_concavepos}
\end{figure}

In terms of the Figure~\ref{fig_concavepos}, the welfare term in numerator is exactly the shaded area in the left figure, i.e., the area between $\lx$ and $r(x)$ in the desired region. Now since $\lx$ is concave, we know in the region $\tilde{x}$ to $x^*$, the function decreases at a rate of $|\lambda'(\tilde{x})|$ or faster. Therefore we define an auxiliary function $g(x)$ (dashed lines in the right figure), such that $g(\tilde{x}) = \lambda(\tilde{x})$ and $g(x)$ decreases linearly at a rate of $|\lambda'(\tilde{x})|$. Since $\lx$ is concave, we conclude that for all $x \in [\tilde{x},x^*]$.

Now, we can say that the shaded area in the left figure which we need to evaluate is upper bounded by the shaded area in the right figure. That is, the shaded triangle represents the area between $g(x)$ and $r(\tilde{x})$ from $\tilde{x}$ to some point $x_{max} \geq x^*$, where $g(x) = r(\tilde{x})$. The area under this triangle is given by,
$$\Delta = \frac{1}{2}(x_{max}-\tilde{x})(\lambda(\tilde{x}) - r(\tilde{x})) = \frac{1}{2}\frac{\lambda(\tilde{x}) - r(\tilde{x})}{|\lambda'(\tilde{x})|}(\lambda(\tilde{x}) - r(\tilde{x})).$$

Going back to our original expression for efficiency ($\eta$), we have
\begin{align*}
\eta & \leq & 1 + \frac{\int_{\tilde{x}}^{x^*}[\lx - r(x)]dx}{\int_0^{\tilde{x}}[\lx - r(x)]dx}\\
& \leq & \frac{\Delta}{\tilde{x}(\lambda(\tilde{x}) - r(\tilde{x}))} \\
& = & 1+ \frac{1}{2}\frac{\lambda(\tilde{x}) - r(\tilde{x})}{|\lambda'(\tilde{x})|}\frac{(\lambda(\tilde{x}) - r(\tilde{x}))}{\tilde{x}(\lambda(\tilde{x}) - r(\tilde{x}))}\\
& = & 1+ \frac{1}{2}M
\end{align*}
The last step follows from Corollary~\ref{corr_eqconditions}. This completes the proof of the efficiency bound. Looking at the above figure, the bound appears to be a weak one since there is a large gap between the area of the triangle and that of the curve. However, in the worst case $\lx -r(x)$ is a linear decreasing function and the two areas coincide. As the following example illustrates, there exists instances where no equilibrium can achieve a social welfare better than a factor $1+\frac{M}{2}$ over the optimum. %

\textbf{Tight Example.} Consider a simple path of $M$ links, having the following inverse demand function
\begin{align*}
\lx = &  2M+1 & 0 \leq x \leq 1\\
\lx =  & (2M+1)-2(x-1) & 1 \leq x
\end{align*}
This function is only piecewise concave, but a ``smoothed" version of this function where the neighborhood of $x=1$ is continuously differentiable still gives the desired bound (although our upper bound also holds for piecewise concave functions; see Section~\ref{sec:generalizations}). Every edge has a cost function given by $C_e(x) = \frac{1}{M}x$. The following figure shows both $\lx$ and $r(x)$ for a sample instance with $M=4$. 

\begin{figure}[h]
\centering
\includegraphics[scale=0.7]{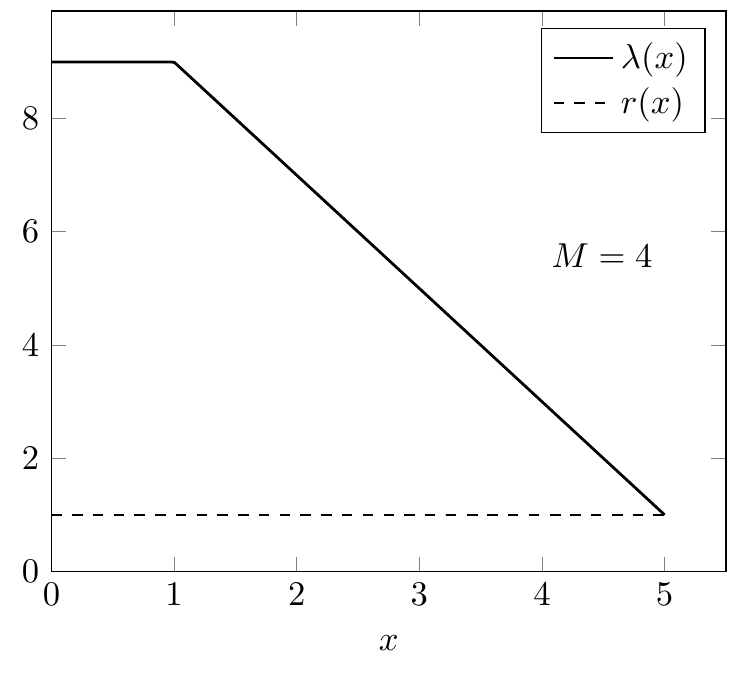}
\caption{Tight example for concave inverse demand functions.}
\label{fig_concavetightex}
\end{figure}

%First we claim that the function is concave. This is not hard to observe since the rate of decrease is $0$ between $x=0$ and $x=1$ and the rate of decrease (slope of the curve) is $2$ between $1$ and $M+1$, so the rate at which the function is decreasing is ``non-decreasing".

We remark here that the cost functions do not satisfy $c_e(0)=0$, this can be viewed as the limiting case for a function with $c_e(0)=0$, but which increases rapidly for $x>0$ until it becomes a constant function $c_e(x)=1/M$.

From the figure, at the optimum solution $\lx$ and $r(x)$ meet at $x^*=M+1$. The social welfare of the solution is the total area under the curve minus the area under the dashed line, which equals $2M + M^2$. We claim that the point $\tilde{x} = 1$ is a NE when the edges are priced according our pricing rule. To see why this is true, note that every edge is priced at $p_e = 2+\frac{1}{M}$ so the total price on the path is $2M+1$ so a flow of $\tilde{x}=1$ is a valid best-response. Each edge makes a profit of $\pi_e = (2+\frac{1}{M})*1 - \frac{1}{M} = 2$. Note that $2M+1$ is the maximum value held by the buyers for this bundle, so no edge has any incentive to increase its price.

We now consider the profit of an edge which decreases its price. Suppose any one seller reduces his price by $\Delta$ and the flow increases by $\epsilon$, then $\frac{\Delta}{\epsilon} = 2$, which is the slope of the linear part of the curve.

The new profit of the edge is given by
\begin{align*}
\pi'_e & = (2+\frac{1}{M} - \Delta)(1+\epsilon) - \frac{1}{M}(1+\epsilon)\\
& = (2+\frac{1}{M} - 2\epsilon)(1+\epsilon) - \frac{1}{M}(1+\epsilon)\\
& = (2 - 2\epsilon^2) \leq 2.
\end{align*}
So the new profit is never larger than the old profit. So clearly this pricing gives an equilibrium. The social welfare at this equilibrium is just $(2M+1)-1 = 2M$.

We now claim that for every equilibrium of this instance, the total flow is $1$, i.e., $\exists$ no equilibrium where the flow $\tilde{x} \neq 1$. Indeed, it is not too hard to show that if there is an equilibrium with $\tilde{x} < 1$, the edge with the largest price can lower its price infinitesimally and increase flow up to one. Suppose there is an equilibrium with flow $x$ strictly larger than 1. The total price on the path must be strictly less than $2M+1$ and so at least one edge has a price $p_e < \frac{2M+1}{M} = 2+\frac{1}{M}$. By some basic algebra, we can verify that if this edge increases its price so that the new flow is $1$, its profit will also strictly increase.

So, we conclude that every equilibrium solution has to have a flow of $\tilde{x}=1$ and has a social welfare of $2M$. The efficiency for this instance is therefore,
$$\eta = \frac{2M+M^2}{2M} = 1 + \frac{M}{2}.$$\end{proof}

\subsection{Proofs: Filling the gaps: Other classes of Demand Functions}
\label{appsection:specificpos}
In this section, we consider $\lx$ parameterized by a single parameter $\alpha$ in order to the spectrum between $\text{Efficiency}=1$ and $\eta=e^{M}$. More specifically, we consider two classes of demand functions which we term $F_p$, $F_{ced}$; we will consider a third specific class $F_{exp}$ in the following sub-section.

\begin{definition}
Let $F_p$ denote concave inverse demand functions that decrease polynomially with degree $\alpha$, i.e., $\lx=\lambda_0(1-x^{\alpha})$ for any $\lambda_0\geq 0$ and $\alpha \geq 1$.
\end{definition}

For demand functions in $F_p$, we will show that the efficiency ($\eta$) is $(1+M\alpha)^{1/\alpha}$. When $\alpha$ is comparable to (or larger than) $M$, this quantity becomes $\approx \frac{\log(M \alpha)}{\alpha}$ and lies between $1$ and $1+M/2$. In other words, for a fixed large value of $\alpha$, $\eta$ increases almost logarithmically with $M$. Note that for the ease of exposition, we have defined the function so that it has a root at $x=1$. These functions can be easily scaled to have a root at any other positive $x$, (e.g, $\lx = \lambda_0(1-(\frac{x}{a})^{\alpha})$) without changing the results. 

\begin{definition}
Let $F_{ced}$ denote the class of MHR functions with inverse demand function $\lx = \lambda_0(1-x)^{\alpha}$ for any $\lambda_0\geq 0$ and $\alpha \geq 1$. %Equivalently, these functions can be represented as constant elastic demand functions in the demand space, i.e., $d(p) = c_0p^{\frac{1}{\alpha}}$.
\end{definition}

{\em CED} stands for {\em constant elastic demand}: it is not hard to verify that these functions have a constant elasticity of demand (refer Appendix for the relation between inverse demand and direct demand). We show that the efficiency for these functions is $1+M\frac{\alpha}{\alpha+1}$ and therefore spans the region between $1+\frac{M}{2}$ and $1+M$.

\begin{thm_app}{thm_spec2}
\begin{enumerate}
\item Let $F_p$ denote functions of the form $\lx=\lambda_0(1-x^{\alpha})$ for any $\lambda_0\geq 0$ and $\alpha \geq 1$. For $\lambda\in F_p$, the efficiency $\eta$ is at most $(1+M\alpha)^{\frac{1}{\alpha}}$. When $\alpha \geq M$, this quantity is approximately $1+\frac{\log(M\alpha)}{\alpha}$.

\item Let $F_{ced}$ denote functions of the form  $\lx = \lambda_0(1-x)^{\alpha}$ for any $\lambda_0\geq 0$ and $\alpha \geq 1$. For $\lambda \in F_{ced}$, $\eta$ is at most $1+M\frac{\alpha}{\alpha+1}$ for $\alpha \geq 1$.

\end{enumerate}
\end{thm_app}

\begin{proof}\textbf{Proof of Statement 1.} Let $\lx = (1-x^{\alpha})$, where $\alpha \geq 1$. The exact same proof holds when there are scaling factors. The function is concave so we know $\exists$ a Nash Equilibrium obeying the conditions of Corollary~\ref{corr_eqconditions}. Let $r(x)$ be the differential min-cost function. At the equilibrium point $\tilde{x}$, we have $\lambda(\tilde{x}) - r(\tilde{x}) - M\tilde{x}|\lambda'(\tilde{x})|=0$. Substituting $\lx = (1-x^{\alpha})$ and $\ldx = -\alpha x^{\alpha-1}$, we get $(1-\tilde{x}^{\alpha} - M\alpha	\tilde{x}^{\alpha}) = r(\tilde{x})$. At the optimum point $x^*$, we know that $\lambda(x^*) \geq r(x^*)$ by Proposition~\ref{prop_optflow}. Since $r(x)$ is non-decreasing, we have, $1-(x^*)^{\alpha} \geq r(x^*) \geq r(\tilde{x}) = 1 - \tilde{x}^{\alpha}(1+M\alpha)$. Therefore, we get $(\frac{x^*}{\tilde{x}})^{\alpha} \leq (1 + M\alpha)$ or equivalently, $\frac{x^*}{\tilde{x}} \leq (1+M\alpha)^{\frac{1}{\alpha}}$.

Recall that social welfare of a min-cost flow is size $x$ is equal to $\int_0^x[\lambda(t)-r(t)]dt$. Now, since $\lambda(x)-r(x)$ is non-increasing, and the social welfare of $\tilde{x}$ equals $\int_0^{\tilde{x}}[\lambda(t)-r(t)]dt$ since it is a min-cost flow, then we know that the efficiency is bounded by $\frac{x^*}{\tilde{x}}$. So for functions in $F_p$, we have the $\eta \leq \frac{x^*}{\tilde{x}} \leq (1+M\alpha)^{\frac{1}{\alpha}}$. For a fixed value of $M$, as $\alpha$ increases the bound tends to $1+\frac{\log(M\alpha)}{\alpha}$. $\blacksquare$

\textbf{Proof of Statement 2.}\\
Once again, we consider $\lx = (1-x)^{\alpha}$ for $\alpha \geq 1$. It is easy to verify that the proof is valid when this function is scaled and/or shifted. The function belongs to MHR, so Corollary~\ref{corr_eqconditions} holds. Moreover, it is convex.

So at equilibrium, we have $(1-\tilde{x})^{\alpha} - r(\tilde{x}) = M\alpha \tilde{x} (1-\tilde{x})^{\alpha-1}$ or $r(\tilde{x}) = (1-\tilde{x})^{\alpha} - M\alpha \tilde{x} (1-\tilde{x})^{\alpha-1}$.	So now the efficiency is bounded by
$$\eta \leq 1 + \frac{\int_{\tilde{x}}^{x^*} \lx - r(x) dx}{\tilde{x}(\lambda(\tilde{x}) - r(\tilde{x}))}.$$

Consider the following function $f(x) = \lx - r(\tilde{x})$. Since $r(x)$ is non-decreasing, we have that $\lx - r(x) \leq \lx - r(\tilde{x})$, $\forall x \geq \tilde{x}$. So we can bound the integral in the numerator above by $\int_{\tilde{x}}^{x^*} f(x)dx$. Furthermore, let $y\geq x^*$ be the point such that $\lambda(y) = (1-y)^{\alpha} = r(\tilde{x})$; this point is larger than $x^*$ since at the optimum solution $\lambda(x^*) \geq r(x^*) \geq r(\tilde{x})$ and $r$ is non-decreasing. For each $x\leq y$, we have that $f(x)\geq 0$, and so the numerator above can be bounded by $\int_{\tilde{x}}^{y} f(x)dx$. Thus our bound on the efficiency becomes,
\begin{align*}
\eta & \leq & 1 + \frac{\int_{\tilde{x}}^{y}\lx dx - r(\tilde{x})(y - \tilde{x})}{\tilde{x}(\lambda(\tilde{x}) - r(\tilde{x}))}\\
& = & 1 + \frac{\frac{1}{\alpha+1}\left((1-\tilde{x})^{\alpha+1} - (1-y)^{\alpha+1}\right) - r(\tilde{x})(y - \tilde{x})}{\tilde{x}(\lambda(\tilde{x}) - r(\tilde{x}))}\\
& = & 1 + \frac{\frac{1}{\alpha+1}\left((1-\tilde{x})^{\alpha+1} - (1-y)^{\alpha+1} - r(\tilde{x})(y - \tilde{x}) - \alpha r(\tilde{x})(y - \tilde{x})\right)}{\tilde{x}(\lambda(\tilde{x}) - r(\tilde{x}))}
\end{align*}
Now, we substitute $r(\tilde{x})(y-\tilde{x}) = r(\tilde{x})((1-\tilde{x}) - (1-y)) = r(\tilde{x})(1-\tilde{x}) - r(\tilde{x})(1-y)$. We know that $r(\tilde{x})(1-\tilde{x}) = (1-\tilde{x})^{\alpha+1} - M \alpha \tilde{x} (1-\tilde{x})^{\alpha}$ from the equilibrium conditions and that $r(\tilde{x})(1-y) = (1-y)^{\alpha+1}$ due to our choice of $y$. So, $r(\tilde{x})(y-\tilde{x}) = (1-\tilde{x})^{\alpha+1} - M \alpha \tilde{x} (1-\tilde{x})^{\alpha} - (1-y)^{\alpha+1}$. Substituting this in the bound for $\eta$, we get
\begin{align*}
\eta & \leq & 1 + \frac{1}{\alpha+1} \frac{M\alpha \tilde{x} (1-\tilde{x})^{\alpha} - \alpha r(\tilde{x})(y - \tilde{x})}{\tilde{x}(\lambda(\tilde{x}) - r(\tilde{x}))}\\
\eta & \leq & 1 + \frac{\alpha}{\alpha+1} \frac{M\tilde{x} (1-\tilde{x})^{\alpha} - r(\tilde{x})(y-\tilde{x})}{\tilde{x}(\lambda(\tilde{x}) - r(\tilde{x}))}
\end{align*}
Now, we use the idea that since $\lx$ is convex, $y - \tilde{x} \geq M\tilde{x}$, which we will prove later. So substituting this in the above inequality, we get
$$\eta \leq 1+\frac{\alpha}{\alpha+1} \frac{M\tilde{x} (1-\tilde{x})^{\alpha} - r(\tilde{x})(M\tilde{x})}{\tilde{x}(\lambda(\tilde{x}) - r(\tilde{x}))} = 1 + \frac{\alpha}{\alpha+1}M,$$
as desire. Now, we prove the required sub-claim.
\begin{lemma}
If $\lx$ is convex and $\lambda(y)=r(\tilde{x})$, then $y - \tilde{x} \geq M \tilde{x}$.
\end{lemma}
\begin{proof}
Let $r(x) = r(\tilde{x})$. We know at the equilibrium point that $|\lambda'(\tilde{x})| = \frac{\lambda(\tilde{x}) - r(\tilde{x})}{M\tilde{x}}$. Since $\lx$ is convex, it means that $|\ldx|$ is non-increasing, which means that for any $x \geq \tilde{x}$, $|\ldx| \leq |\lambda'(\tilde{x})|$. So, we have $\lx \geq \lambda(\tilde{x}) - |\lambda'(\tilde{x})|(x - \tilde{x})$. Consider the following point $x_2 = \tilde{x} + M\tilde{x}$. At $x_2$, we have $\lambda(x_2) - r(\tilde{x}) \geq \lambda(\tilde{x}) - |\lambda'(\tilde{x})|(x_2 - \tilde{x}) - r(\tilde{x}) = \lambda(\tilde{x}) - r(\tilde{x}) - \frac{\lambda(\tilde{x}) - r(\tilde{x})}{M\tilde{x}} \times M\tilde{x} = 0$. Since at $x_2$, $\lambda(x_2) \geq r(\tilde{x})=\lambda(y)$, and $\lambda$ is non-increasing, then it must be that $y\geq x_2$. So, we get $y - \tilde{x} \geq x_2 - \tilde{x} = M\tilde{x}$.
\end{proof}
$\blacksquare$ \end{proof}

\subsection{Monotone Price Elasticity (MPE) with $|x\ldx|$ increasing}
Recall the definition of the MPE functions: $x\frac{|\ldx|}{\lx}$ is non-decreasing in $x$ and further tends to zero as $x$ tends to zero. It is not hard to see that MHR functions also fall under this class: if $h(x) = \frac{|\ldx|}{\lx}$ is increasing then $xh(x)$ is also increasing with $x$. So MPE is a generalization of MHR. We consider another class of MPE functions that do not have a monotone hazard rate here. 

First we state without proof a simple claim that incorporates the cost functions into the hazard rate. The proof of this is identical to the proof of Lemma~\ref{lem_mhrcost}, replacing $|\lambda'(x)|$ with $x|\lambda'(x)|$.

\begin{claim}
\label{claim_lrMPE}
If $\lx$ is a MPE function, i.e., $\frac{x|\ldx|}{\lx}$, is non-decreasing, then $\frac{x|\ldx|}{\lx - r(x)}$ is also non-decreasing, where $r(x)$ is the differential min-cost function.
\end{claim}

We show a bound on the price on efficiency for another large class of MPE functions: ones in which $x|\ldx|$ is non-decreasing. As an example for such functions, consider $\lx=|\log(\frac{x}{a})|$ for $x=0$ to $a$, which is not MHR but is MPE.\\

\begin{clm_app}{thm_mpepos}
If $\lx$ has $x|\ldx|$ non-decreasing, then the loss in efficiency for any instance is at most $\frac{M}{M-1}e^M$.
\end{clm_app}
\cite{chawla2009price} showed a bound of $e^{k}$ on the Price of Anarchy for instances with non-decreasing $x|\ldx|$, when each edge has a fixed capacity. Here $k$ is the length of the longest $s-t$ path. Viewed in this context, there are two things we can infer from the above theorem. First, our bound depends only on the number of monopolies instead of the longest $s-t$ path, which can indeed be arbitrarily large. Moreover, we obtain a bound tighter than just $e^{M}$. Secondly, our result shows that the bound obtained in~\cite{chawla2009price} can be extended to instances where each item has a production cost as well.
\begin{proof}
Let $\vec{\tilde{x}}$ be the equilibrium flow of size $\tilde{x}$ obeying Corollary \ref{corr_eqconditions}, and let $x^*$ be the size of the optimum flow maximizing social welfare. We first show that $\frac{x^*}{\tilde{x}} \leq e^M$. Note that this immediately gives us a bound of $e^M$ on the $\eta$ since the flow at equilibrium always has a higher value than the flow after that. Following this, we tighten the bound to $\frac{M}{M-1}e^{M-1}$.

\begin{claim}
For any given instance with $x|\ldx|$ being non-decreasing, it must be that $\frac{\tilde{x}}{x^*} \geq e^{-M}$.
\end{claim}

\begin{proof}
Following our pricing rule, at equilibrium, we have $\frac{\lambda(\tilde{x}) - r(\tilde{x})}{M} = \tilde{x}|\lambda'(\tilde{x})|$. Since $x|\ldx|$ is non-decreasing, it must be that for all $x \geq \tilde{x}$, $x|\ldx| \geq \tilde{x}|\lambda'(\tilde{x})| = \frac{\lambda(\tilde{x}) - r(\tilde{x})}{M}$.
Equivalently for $x \geq \tilde{x}$,
$$|\ldx| = -\ldx \geq \frac{\lambda(\tilde{x}) - r(\tilde{x})}{Mx}$$

Integrating both sides from $x=\tilde{x}$ to $x^*$, we get
$$\lambda(\tilde{x}) - \lambda(x^*) \geq \frac{\lambda(\tilde{x}) - r(\tilde{x})}{M}\ln(\frac{x^*}{\tilde{x}}).$$
This implies
$$\ln(\frac{x^*}{\tilde{x}}) \leq M \frac{\lambda(\tilde{x}) - \lambda(x^*)}{\lambda(\tilde{x}) - r(\tilde{x})}.$$
Finally, we get $\frac{x^*}{\tilde{x}} \leq \exp{\left(M\frac{\lambda(\tilde{x}) - \lambda(x^*)}{\lambda(\tilde{x}) - r(\tilde{x})}\right)} \leq e^{M}$. The last inequality follows from the fact that at the optimum flow $x^*$, $\lambda(x^*) \geq r(x^*) \geq r(\tilde{x})$.
\end{proof}

Now for the rest of the proof. Define $f(x) = \ln(\lx - r(x))$. Then $f'(x) = \frac{\ldx - r'(x)}{\lx - r(x)} \leq \frac{\ldx}{\lx - r(x)}$ since $r'(x)$ is always non-negative as $r(x)$ is non-decreasing. So, by integrating $f'(x)$ between $\tilde{x}$ and $x>\tilde{x}$, we get

$$f(x) - f(\tilde{x}) = \int_{\tilde{x}}^{x} f'(x)dx \leq \int_{\tilde{x}}^{x} \frac{\ldx}{\lx - r(x)}dx.$$
Now since $\lx$ is a MPE function and $\ldx$ is negative, we have for $x \geq \tilde{x}$, $\displaystyle \frac{\ldx}{\lx - r(x)} \leq \tilde{x}\frac{\lambda'(\tilde{x})}{x(\lambda(\tilde{x}) - r(\tilde{x}))} = \frac{\tilde{x}}{-xM\tilde{x}} = \frac{-1}{Mx}$. This is due to our equilibrium conditions from Corollary \ref{corr_eqconditions}. Putting this in the above integral, we get,
$$\ln(\frac{\lx - r(x)}{\lambda(\tilde{x}) - r(\tilde{x})}) \leq \int_{\tilde{x}}^{x} \frac{\ldx}{\lx - r(x)}dx \leq \int_{\tilde{x}}^{x} \frac{-1}{Mx}dx = \frac{-1}{M}\ln(\frac{x}{\tilde{x}})$$

Since $\log(x)$ is a monotone increasing function, we can directly say,
$$\lx - r(x) \leq \displaystyle(\lambda(\tilde{x}) - r(\tilde{x}))(\frac{\tilde{x}}{x})^{\frac{1}{M}}.$$
From our previous proofs, we know that the efficiency is bounded from above by $\displaystyle 1+\frac{\int_{\tilde{x}}^{x^*}[\lx - r(x)]dx}{\tilde{x}(\lambda(\tilde{x}) - r(\tilde{x}))}$. Substituting our bound for $\lx - r(x)$ from the above inequality, we get
$$\eta \leq  1 +  \frac{\int_{\tilde{x}}^{x^*}(\frac{\tilde{x}}{x})^{\frac{1}{M}}dx}{\tilde{x}} =  1 + \frac{M}{M-1}\frac{(x^*)^{\frac{M-1}{M}} - (\tilde{x})^{\frac{M-1}{M}}}{\tilde{x}^{\frac{M-1}{M}}}.$$
We already know that $\frac{x^*}{\tilde{x}} \leq e^M$ from the claim proven above. So putting this in the above inequality gives us,
$$\eta \leq 1 + \frac{M}{M-1}(e^M)^{\frac{M-1}{M}} - \frac{M}{M-1} \leq \frac{M}{M-1}e^{M-1}.\qed$$ 
\end{proof}

\begin{clm_app}{clm_mpespecem}
Let $F_{exp}$ denote the demand and cost functions which can be represented as $\lx - r(x) = |\ln(\frac{x}{a})|^{\frac{1}{\alpha}}$ for $x \leq a$, and $\alpha\geq 1$. Then the efficiency $\eta$ is at most $e^{\frac{M}{\alpha}}$ for $\alpha \geq 1$.
\end{clm_app}
\begin{proof}
Since $\lx - r(x) = |\ln(\frac{x}{a})|^{\frac{1}{\alpha}}$, at the optimum flow $x^*$, we have $\lambda(x^*) - r(x^*) \geq 0$ so $x^* \leq a$ and $\tilde{x}<a$. Since $r(x)$ is non-decreasing, $r'(x)$ is non-negative. This gives us $\lambda'(x) \geq \lambda'(x) - r'(x) = \frac{-1}{\alpha}|\ln(\frac{x}{a})|^{\frac{1}{\alpha}-1}\frac{1}{x}$. At equilibrium, we know that $\lambda(\tilde{x}) - r(\tilde{x}) = M\tilde{x}|\lambda'(\tilde{x})| \leq M\tilde{x} |\lambda'(x) - r'(x)|$. Thus, $|\ln(\frac{\tilde{x}}{a})|^{\frac{1}{\alpha}} \leq M\tilde{x}\frac{1}{\alpha}|\ln(\frac{\tilde{x}}{a})|^{\frac{1}{\alpha}-1}\frac{1}{\tilde{x}}.$ Canceling out the common terms, we get at equilibrium $|\ln(\frac{\tilde{x}}{a})| \leq \frac{M}{\alpha}$ or $\tilde{x} \geq ae^{-\frac{M}{\alpha}}$. So $\frac{x^*}{\tilde{x}} \leq e^{\frac{M}{\alpha}}$, which gives us the desired bound on the efficiency using the arguments from the proof of Statement 1 above. \end{proof}

\subsection*{Example: Unbounded Efficiency for MPE functions}

\begin{claim}
\label{clm_unboundedpos}
The social welfare can be arbitrarily bad for MPE functions.
\end{claim}
\begin{proof}
Consider a two-link path ($M=2$). The same example can generalized for any value of $M$, however assuming a specific value of $M$ allows us to derive the equilibrium and optimum flows analytically. The demand function has the following structure: $\lx = x^{-1/r}$ for $r > M$. However, this function does not belong to the class MPE since $\lim_{x \to 0}\frac{x|\ldx|}{\lx} \neq 0$. Therefore, we truncate this function at some very small $\epsilon > 0$ such that $\lx=x^{-1/r}$ for $x \geq \epsilon$ and $\lx = (\epsilon)^{-1/r}$ for $x < \epsilon$. A sample demand function of this form is represented in Figure~\ref{fig:posunbounded}.

\begin{figure}
\includegraphics[scale=0.7]{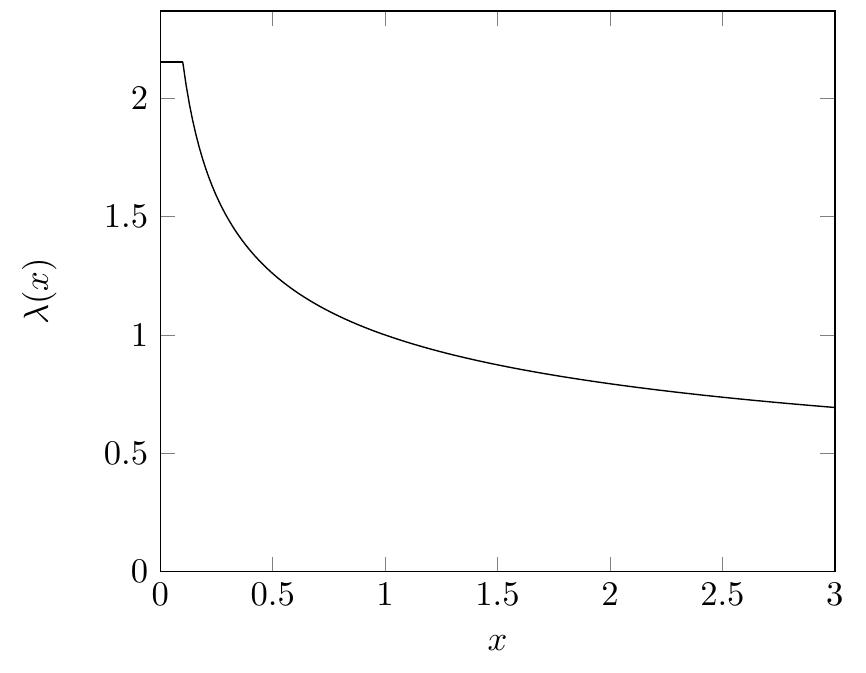}
\caption{Demand function of the form $\lx = x^{-1/r}$ truncated at some small $\epsilon$.}
\label{fig:posunbounded}
\end{figure}

We assume that each edge has a cost function given by $C_e(x) = c_0x$ for some positive $c_0$. The unique optimum flow is the solution with $(x^{*})^{-1/r}=2c_0$. Now, the most general condition that any equilibrium point must obey is given by Lemma~\ref{lemma_eqnconditionsgen}. For points where $\lambda'(x)$ is defined which in our case is everywhere except $x=\epsilon$, the condition reduces to $p_e + \tilde{x}\lambda'(\tilde{x}) - c_e(\tilde{x}) = 0$. Since both edges have the same $c_e$, this means that in any equilibrium at a point other than $\epsilon$, both the edges need to have the exact same price $p_e$.

Since the sum of prices is $\lambda(\tilde{x})$, the price on each edge must be equal to $\frac{\lambda(\tilde{x})}{2}$. First let's see if any $\tilde{x} \geq \epsilon$ satisfies the equilibrium conditions. Substituting the value for $p_e$ and the fact that $x\ldx = \frac{-1}{r}x^{-1/r}$, we get the equilibrium condition as
$$\tilde{x}^{-1/r}(\frac{1}{2} - \frac{1}{r}) = c.$$

It is not hard to see that if we price the edges according to our pricing rule at this point, which once again just reduces to $p_e = \lambda(\tilde{x})/2$, this point satisfies the conditions of Lemmas~\ref{lem_eqnconditionsincr} and~\ref{lem_eqnconditionsdecr} for $x$ in $[\epsilon, x^*]$. We do not care about the region between $[0,\epsilon]$ since $\lx$ is constant in this region so the profit is no larger than $\pi(\epsilon)$.

Now we have a closed form expression for both the optimum point $x^*$ and NE $\tilde{x}$. Using some basic algebra and computing the social welfare of both these solutions, we see that as $r \to M$, the ratio of the social welfares becomes arbitrarily large. At the same time, Theorem~\ref{thm_specposcap} implies that if the edges have capacities instead of costs then for any $r < M$, the efficiency is one.

We can construct similar instances for larger values of $M$ and $r > M$ and show that the unique equilibrium can be arbitrarily worse than the optimum.
\end{proof}

\section{Proofs from Section 5: Generalizations}
\label{appsec:generalizations}

For ease of exposition, we assumed that both $\lambda$ and $C_e$ for all $e$ are continuously differentiable. However, we now reprove our existence and other results for the case when $\lx$ and $C_e$ are continuous but only piecewise differentiable.

\subsection{Dropping Differentiability in Demand}
%Throughout our paper, we have assumed the existence of $\ldx$ and $r(x)$, which relies on the fact that both $\lx$ and $C_e(x)$ are continuously differentiable at every $x$. This was mainly for ease of exposition. In this section, we show that all our results carry over without any change whatsoever when $\lx$ and $C_e(x)$ are continuous but not necessarily differentiable. Notice that we still assume that $\lx$ and $C_e(x)$ are monotone, therefore continuity guarantees the existence of left hand and right hand derivatives at all points.
%
%\subsubsection*{Continuous $\lx$ and atomic-buyers}
We now generalize our model to consider the case when $\lx$ is continuous, monotone non-increasing, but not necessarily differentiable. All our results are valid even for these general demand functions as long as the function is differentiable piecewise. We now formalize this notion.

For almost-everywhere differentiable $\lx$, we define the class MPE to be as follows. Define $\lambda'_-(x) = \lim_{x \to x^-}\lambda'(x)$ and $\lambda'_+(x) = \lim_{x \to x^+}\lambda'(x)$. A function is MPE if these limits exist and for all $x_1< x_2$ we have that $$\frac{x_1|\lambda'_-(x_1)|}{\lambda(x_1)}\leq \frac{x_1 |\lambda'_+(x_1)|}{\lambda(x_1)}\leq \frac{x_2 |\lambda'_-(x_2)|}{\lambda(x_2)}\leq \frac{x_2|\lambda'_+(x_2)|}{\lambda(x_2)},$$ and if $\frac{x_1 |\lambda'_+(x_1)|}{\lambda(x_1)}$ tends to 0 as $x\to 0$.
In particular, the property also implies that for any $x$,
%$$\frac{x|\lambda'_-(x)|}{\lambda(x)} = \lim_{t \to x^-}\frac{t|\lambda'(t)|}{\lambda(t)} \leq \lim_{t \to
%x^+}\frac{t|\lambda'(t)|}{\lambda(t)} = \frac{x|\lambda'_+(x)|}{\lambda(x)}.$$
%
%This implies that for every $x$ where the derivative does not exist, the MPE property forces the function to
we have $|\lambda'_-(x)| \leq |\lambda'_+(x)|$.

\begin{theorem}
\label{cl_discolambdaex}
For continuous but non-differentiable $\lx$ belonging to the class MPE, there always exists a Nash equilibrium with a
non-trivial flow of $\tilde{x} > 0$ obeying our pricing rule. Furthermore this equilibrium obeys all our desired properties: Non-Trivial Pricing, Recovery of Production Costs, Pareto-Optimality in the space of all equilibria and Local Dominance.
\end{theorem}
\begin{proof}

As with our continuously differentiable case, we give a series of lemmas that capture some sufficient conditions for equilibrium. We then show that for MPE functions, there must always exist a point satisfying these conditions and also obeying our pricing rule. Recall that the quantity $\pe$ is defined to be the increased price on an edge from its marginal cost, i.e., the total price $p_e = \pe + c_e(x_e)$. Usually when we refer to $\pe$, we refer to the equilibrium prices.

First we show that Lemma \ref{lem_eqnconditionsincr} still holds if $\lambda$ is not continuously differentiable, with the condition in the lemma being $\tilde{x}|\lambda'_-(\tilde{x})|$ instead of $\tilde{x}|\lambda'(\tilde{x})|$. The proof is very similar to that of Lemma~\ref{lem_eqnconditionsincr} so we only sketch the differences here. Define the profit function $\pi(x)=[p_e + \lx - \lambda(\tilde{x})]x-C_e(x)$. We will prove that in the domain $(0,\tilde{x}]$, $\pi(x)$ is maximized when $x=\tilde{x}$, thus implying that no matter what amount the seller increases his price by, at the resulting flow of magnitude $x$, his profit cannot be strictly larger than the original profit. Once again, consider the derivative of this function. Since $\lambda'(x)$ may not exist at all points, we simply define the derivative in the limit as well, i.e., $\pi'_-(x)$ and $\pi'_+(x)$ at points where the derivative does not exist.

\begin{align*}
\pi'_+(x) & = p_e + \lx - \lambda(\tilde{x}) + x\lambda'_+(x) - c_e(x) \\
& \geq \lx - \lambda(\tilde{x}) + \tilde{x}|\lambda'_-(\tilde{x})| - x|\lambda'_+(x)| + (c_e(\tilde{x}) - c_e(x))
\end{align*}

Since the last term in the parenthesis, $c_e(\tilde{x}) - c_e(x)$ is non-negative ($c_e$ is non-decreasing), in order to show that $\pi'(x) \geq 0$ for all $x \leq \tilde{x}$, it suffices if we show the first terms are non-negative. The following proposition implies that for MPE functions, this is indeed true.

\begin{proposition} For $\lambda\in MPE$ and $M\geq 1$, we have that $\lambda(\tilde{x})-\lambda(x)\leq \tilde{x}|\lambda'_-(\tilde{x})|-x|\lambda'_+(x)|$ for $x<\tilde{x}$.
\end{proposition}
The proof of this proposition is very similar to that of Proposition~\ref{sublem_exist_smallx} and we use the fact that,
\begin{equation*}
\frac{\lambda(\tilde{x})}{\tilde{x}|\lambda'_-(\tilde{x})|}\leq \frac{\lx}{x|\lambda'_+(x)|}.
\end{equation*}
We have argued now that for all $x$, $\pi'_+(x) \geq 0$. Now for any given $x$, $\pi'_-(x) \geq \pi'_+(x)$, since $|\lambda'_+(x)| \geq |\lambda'_-(x)|$, so $\pi'_-(x) \geq \pi'_+(x) \geq 0$ for all $x \leq \tilde{x}$. This completes the proof of Lemma \ref{lem_eqnconditionsincr}.

Now we show that Lemma \ref{lem_eqnconditionsdecr} still holds, with the condition in the lemma being $\tilde{x}|\lambda'_+(\tilde{x})|$ instead of $\tilde{x}|\lambda'(\tilde{x})|$. We want to show that the new profit a seller makes is at most the old profit as long as Condition (3) is obeyed.  Consider seller $e$'s profit, $\pi(x)=[p_e + \lx - \lambda(\tilde{x})]x-C_e(x)$. Similar to the extension of Lemma~\ref{lem_eqnconditionsincr}, we can show that in the domain $[\tilde{x},T]$, $\pi(x)$ is maximized when $x=\tilde{x}$, thus implying that no matter what amount the seller decreases his price by, at the resulting flow of magnitude $x$, his profit cannot be strictly larger than the original profit. The proof uses the following Lemma,

\begin{lemma}
\label{sublem_exist_largex_diff}
For $\lambda\in MPE$ and $M\geq 1$, we have that $\lambda(\tilde{x})-\lambda(x)\geq \pe-x|\lambda'_-(x)|$ for $x>\tilde{x}$.
\end{lemma}
\begin{proof}
The proof is slightly trickier than all the other cases. The following equation however is the main step

\begin{equation*}
\frac{\lambda(\tilde{x})}{\pe} \geq \frac{\lambda(\tilde{x})}{\tilde{x}|\lambda'_+(\tilde{x})|}\geq \frac{\lx}{x|\lambda'_-(x)|}
\end{equation*}
The equation is true because $\pe \leq \tilde{x}|\lambda'_+(\tilde{x})|$. Now since we have defined all the prices to be positive, it is not possible for any edge in a flow carrying path to be priced larger than $\lambda(\tilde{x})$. This implies that $\lambda(\tilde{x}) \geq p_e \geq \pe$. Now our previous analysis for Proposition~\ref{sublem_exist_largex} carries over to this case.
\end{proof}
Similarly, $\pi'_+(x) \leq \pi'_-(x) \leq 0$. This completes the proof that Lemma \ref{lem_eqnconditionsdecr} still holds, and gives us the following corollary:

\begin{corollary}
\label{lemma_eqnconditionsgen}
Any given solution $(\vec{p}, \vec{x})$ with $\vec{x}$ being a best-response allocation to prices $\vec{p}$ and $\lambda(\tilde{x}) > 0$ is a Nash Equilibrium if the following conditions are met
\begin{enumerate}
\item $\vec{x}$ is a minimum cost allocation of magnitude $x$.

\item All non monopoly edges are priced at their marginal cost, i.e, $\pe= 0$.

\item For all monopoly edges, $\pe \geq \tilde{x}|\lambda'_-(\tilde{x})|$ and one of the following is true,
\begin{enumerate}
\item $\pe \leq \tilde{x}|\lambda'_+(\tilde{x})|$ \emph{or}
\item $\pe = 0$ \emph{or}
\item $\tilde{x} = T$.
\end{enumerate}

\end{enumerate}
\end{corollary}
To show that such a solution exists, once again we price edges according to our pricing rule and so for monopolies $\pe = \frac{\lambda(\tilde{x}) - r(\tilde{x})}{M}$ at some $\tilde{x}$ where we have taken the minimum cost flow. We claim that either the optimum point $x^*$ satisfies $\lambda(x^*) - r(x^*) \geq Mx^*|\lambda'_-(x^*)|$ or there exists $\tilde{x} > 0$ that is a Nash Equilibrium. Clearly if the optimum point satisfies the desired condition, we are done from the corollary, just as in the argument in Section \ref{app:existence}. Suppose that is not true, then $\lambda(x^*) - r(x^*) < Mx^*|\lambda'_-(x^*)|$. But since the function belongs to the class MPE, $\lim_{x \to 0}\lambda(x) - r(x) > \lim_{x \to 0}Mx|\lambda'_+(x)| > \lim_{x \to 0}Mx|\lambda'_-(x)|$. Since $\lx - r(x)$ is continuous, the function $\frac{Mx|\lambda'(x)|}{\lambda(x)}$ is monotone and $\frac{Mx|\lambda'_-(x)|}{\lambda(x)} \leq \frac{Mx|\lambda'_+(x)|}{\lambda(x)}$, we claim there must exist some intermediate value of $x$, where $Mx|\lambda'_-(x)| \leq \lx - r(x) \leq Mx|\lambda'_+(x)|$. Why is this true?

First, if $\lambda'(x)$ is defined in some region (say) $[x_1,x_2]$ such that $\lambda(x) - r(x) > Mx_1|\lambda'(x_1)|$, and $\lambda'(x_2) - r(x_2) < Mx_2|\lambda'(x_2)|$, then there must exist some $x \in [x_1, x_2]$ where $\lx - r(x) = Mx|\ldx|$. Indeed, recall that if a function is differentiable in some region $[x_1,x_2]$, the derivative cannot have jump discontinuities in this region. Moreover since this is  MPE function, the derivative is finite everywhere and the limit at any point exists, and so the curve $\lx - r(x) - Mx|\ldx|$ is continuous and equal 0 at some point in $[x_1,x_2]$.

So the only place where the analysis breaks is for points where $\ldx$ is not defined. Since $\lx$ is differentiable almost everywhere, we only have to worry about point discontinuities. But in this case, we will have that $|\lambda'_-(x)| < |\lambda'_+(x)|$, and since this is an MPE function, we will obtain a point where $\lambda(x) - r(x) - Mx|\lambda_+'(x)| < 0 < \lambda(x) - r(x) - Mx|\lambda_-'(x)|$, as desired.

Summing up, we are therefore, guaranteed the existence of an equilibrium that is either optimal or satisfies the following equation.

\begin{equation}
\label{eqn_generalourrule}
\tilde{x}|\lambda'_-(\tilde{x})| \leq \frac{\lambda(\tilde{x}) - r(\tilde{x})}{M} \leq \tilde{x}|\lambda'_+(\tilde{x})|
\end{equation}

It is not hard to see that this equilibrium obeys all our desired properties. In particular, the proof of Claim~\ref{appclm_eqproperties} carries over with just one minor difference: considering $\lambda'_-$ and $\lambda'_+$ accordingly instead of just the differential.
\end{proof}
%
%Once again we can show that, when $\lx$ is strictly decreasing at $x=\tilde{x}$, this $S$ is a single point. The additional case arises when it is strictly decreasing in one direction only, i.e., $|\lambda'_-(\tilde{x})| = 0$. In this case, as long as $\tilde{x}$ is not optimal, for this function to belong to the class MPE, it is necessary that $\lambda'(x) = 0$ for all $x$ in $[0,\tilde{x})$. This means that a monopoly would want to increase its price since $\pe \leq x|\lambda_+(x)| = 0$ is no longer true here. Therefore we conclude that once again as long as the equilibrium is not optimal, it is unique and as with the previous case, we can show that it can be computed using a binary search algorithm.

All of our other theorems on the efficiency follow almost immediately from Equation~\ref{eqn_generalourrule}. For instance, for MHR functions,
$$\int_{\tilde{x}}^{x^*}(\lx - r(x))dx \leq \frac{\lambda(\tilde{x}) - r(\tilde{x})}{\tilde{x}|\lambda'_+(x)|}\int_{\tilde{x}}^{x^*}(|\ldx|)dx \leq M\int_{\tilde{x}}^{x^*}(|\ldx|)dx.$$ Note that the integration has to be performed piecewise. The same idea applies to the proof of Concave and MPE functions.

We make a small remark on uniqueness here. A restricted version of our uniqueness result does carry over to this setting, mainly that as long as there is some strong monotonicity, the equilibrium flow ($\tilde{x}$) is unique. However, since Equation~\ref{eqn_generalourrule} only gives us upper and lower bounds for prices, the same flow may be induced by several different sets of prices. Notice however, that does not effect our efficiency results because efficiency depends only the allocation: prices are intrinsic and do not affect the welfare directly.

\subsubsection*{Piecewise Cost Functions}

We now move on to the case when the production cost for any item $C_e(x)$ is continuous but not necessarily differentiable. It is extremely common in network flow literature and even in markets of continuous goods~\cite{correa2008pricing} to consider such cost functions. Formally, for a given instance, we assume that $C_e(x)$ is convex, continuous but not (necessarily) differentiable for all $e$. However, we assume that it is piecewise differentiable, i.e., if at some $x=x_0$, $\frac{d}{dx}C_e(x_0) = c_e(x_0)$ does not exist, then both $\lim_{x \to x^-_0}c_e(x) = c^-_e(x_0)$ and $\lim_{x \to x^+_0}=c^+_e(x_0)$ exist and are finite. Once again, if we take $R(x)$ to be the min-cost flow function as defined previously, then it is not too hard to see that $R(x)$ is also continuous and convex.  We define $r^-(x)$ and $r^+(x)$ accordingly wherever $R(x)$ is not differentiable as the left hand and right hand limits of the derivative. Since the functions are convex, we can immediately infer that $c_e^-(x) \leq c_e^+(x)$ and $r^-(x) \leq r^+(x)$ always hold. We still assume that $\forall e$, $c_e(0)= 0$, which we relax in the following section.

Before showing that \textbf{all our efficiency results} extend to this general case, it is useful to first highlight the difficulties encountered when facing piecewise cost functions and why our techniques do not directly apply here. First notice that we used marginal cost pricing extensively for non-monopolies, but this is no longer applicable here because the marginal cost may not be unique. Second, all our efficiently results used Corollary~\ref{corr_eqconditions}. Therefore, we need to define a new version of this based on $r^+(\tilde{x})$ and $r^-(\tilde{x})$ and show that it is still an equilibrium. We first focus on this condition since it is easier to resolve. We prove that mildly modifying the condition for the case with non-differentiable $C_e(x)$ directly yields a sufficient condition for equilibrium existence. We do not reprove the lemmas since their proofs are analogous to the original.

\begin{claim}
\label{clmpccost_sellerpricechange}
Given a solution vector pair $(\vec{p},\vec{\tilde{x}})$ for an instance with $M \geq 1$, where $\vec{\tilde{x}}$ is a best-response flow and all flow carrying paths are priced at $\lambda(\tilde{x}) \geq \tilde{x}|\lambda'(\tilde{x})|$ and non-flow paths are priced at or above $\lambda(\tilde{x})$, no seller $e$ can change his price and improve profits as long as
\begin{enumerate}
\item $p_e \geq \tilde{x}|\lambda'(\tilde{x})| + c^-_e(x)$.

\item $p_e \leq \tilde{x}|\lambda'(\tilde{x})| + c^+_e(x)$ \emph{(or)} $\tilde{x}=T$.
\end{enumerate}
\end{claim}

While the previous claim gave us sufficient conditions that the prices should obey so that sellers cannot increase or decrease their price, the following two lemmas give us an idea of what exactly these prices should be.

\begin{lemma}
\label{lem_eqnconditionspccost1}
Given a min-cost flow $\vec{\tilde{x}}$ for an instance with $M \geq 1$ and a quantity $\tilde{c}_e$ for all $e$ such that,
\begin{enumerate}
\item $r^-(\tilde{x}) \leq \lambda(\tilde{x}) - M\tilde{x}|\lambda'(\tilde{x})| \leq r^+(\tilde{x})$.
\item $\forall e$, $c^-_e(\tilde{x}_e) \leq \tilde{c}_e \leq c^+_e(x_e)$.
\item For any flow carrying path $P$, $\sum_{e \in P}\tilde{c}_e = \lambda(\tilde{x}) - M\tilde{x}|\lambda'(\tilde{x})|$.
\end{enumerate}

Then, pricing all non-monopoly edges at $p_e = \tilde{c}_e$ and all monopoly edges at $p_e = \tilde{c}_e + \tilde{x}|\lambda'(x)|$ results in a Nash Equilibrium.
\end{lemma}

\begin{lemma}
\label{lemma_eqnconditionsoptpccost}
Suppose that the optimal flow satisfies $\lambda(x^*) - Mx^*|\lambda'(x^*)| \geq r^-(x^*)$ and we are given $\tilde{c}_e$ for all edges satisfying the following
\begin{enumerate}
\item $ c^-_e(x_e) \leq \tilde{c}_e \leq c^+_e(x_e)$
\item For all flow carrying paths $P$, $\sum_{e \in P}\tilde{c}_e = r^-(x^*) \leq r^+(x^*) \leq \lambda(x^*)$.
\end{enumerate}
Then $x^*=T$. Further, the solution where all non-monopoly edges are priced at $p_e=\tilde{c}_e$ and monopoly edges are priced at $p_e = \tilde{c}_e + \frac{\lambda(x^*) - r^-(x^*)}{M}$ is a Nash  Equilibrium.
\end{lemma}
The proof follows almost directly from Lemma~\ref{clmpccost_sellerpricechange}. Both the lemmas above depend primarily on the $\tilde{c}_e$ components which ensure that all the paths are balanced. We now show how to compute these $\tilde{c}_e$'s for any given instance satisfying all of the requirements above. The proof uses a primal-dual like approach but is combinatorial and constructive. We actually show something stronger, for any given $p^*$ lying between $r^-(x)$ and $r^+(x)$, we can always compute a set of prices of the edges such that all desired paths have a cost of $p^*$. The following lemma provides the main algorithmic insight required to prove the existence of equilibrium and efficiency bounds.

\begin{claim}
\label{clm_balancingprices}
For any given min-cost flow $\vec{x}$ and a price $p^*$ satisfying $r^-(x) \leq p^* \leq r^+(x)$, we can always compute a vector of $\tilde{c}_e$ for each edge $e$ obeying the following requirements
\begin{enumerate}
\item For all $e$, $c^-_e(x_e) \leq \tilde{c}_e \leq c^+_e(x_e)$.
\item For any flow carrying path $P$, $\sum_{e \in P}\tilde{c}_e = p^*$ and for any non flow-carrying path, $P'$, $\sum_{e \in P'}\tilde{c}_e \geq p^*$.
\end{enumerate}
\end{claim}
\begin{proof}
Observe that if $\forall$ edges, $c^-_e(x_e) = c^+_e(x_e)$, the proof of this claim is trivial. We just set $\tilde{c}_e = c^-_e(x_e) = c^+_e(x_e)$ for all edges. It is not hard to see that for this case, $r^+(x) = r^-(x)$. Since what we have is a min-cost flow, by Lemma~\ref{lemma_nonzeroflowcost}, for all flow paths $\sum_{e}c_e(x_e) = \sum_e \tilde{c}_e = r(x) = p^*$. Indeed, when all the cost functions are twice differentiable, the way to ensure equal price on all paths is to set the price on every edge to be its marginal cost. What about the general case, where for each edge the marginal cost of removing and adding flow does not coincide?

We begin by extending Lemma~\ref{lemma_nonzeroflowcost} to piecewise cost functions. If $\mathcal{P}$ is the set of paths with non-zero flow in $\vec{x}$ and $\mathcal{P'}$ is the set of paths with zero flow, then it is not hard to see the following.
$$\max_{P \in \mathcal{P}}(\sum_{e \in P}c_e^-(x_e)) \leq r^-(x) \leq p^* \leq r^+(x) \leq \min_{P \in \mathcal{P} \cup \mathcal{P'}}(\sum_{e \in P}c_e^+(x_e)).$$

In the other words, $r^-(x)$ is at least the marginal cost of removing one unit of flow from any flow carrying path and $r^+(x)$ is at most the marginal cost of adding one unit of flow to any path. So even if we price edges at either of the two marginals, all the paths may not have the same price of $p^*$. In order to balance the price on all the paths, we turn to the concept of node potentials which are usually obtained from the dual to the min-cost flow LP. First, given the marginal cost of removing an infinitesimal unit of flow ($k^-_e$) and the marginal cost of adding an infinitesimal unit of flow ($k^+_e$) for every single edge, we formally define two node potentials as follows.

\begin{definition}
The negative potential $\pi^-_v$ is defined on every node to be the maximum cost `saved' by removing an infinitesimal unit of flow from the source $s$ to $v$.
\end{definition}

\begin{definition}
The positive potential $\pi^+_v$ is defined on every node to be the minimum cost incurred for sending an infinitesimal unit of flow from the source $s$ to $v$.
\end{definition}

Given a min-cost flow $\vec{x}$, these two marginal costs are well defined for every edge as $c^-_e(x_e)$ and $c^+_e(x_e)$ respectively. In this case, the node potentials obey some nice properties which we now show. First, we redefine these node potentials formally in terms of the residual graph obtained for a min-cost flow, although the (re)definition extends to any negative and positive marginal costs as long as the residual graph does not have negative cycles.

\textbf{Alternative definition of node potentials.}
Define $G^R$ to have the same set of nodes as the original graph but for every $(u,v)$ in the original graph, the residual graph has two directed edges $(u,v)$ and $(v,u)$. The weight on the forward edge $(u,v)$ is said to be $k^+_e = c^+_e(x_e)$ and the reverse edge has a negative weight of $-k^-_e = -c^-_e(x_e)$. Now for any given node $v$, it is not hard to verify that there exists at least one $v$-$s$ path with a negative total weight denoting the change in cost by removing an infinitesimal unit of flow along that path. Then, $\pi^-_v$ is the absolute value of the total cost of the minimum cost $v$-$s$ path. Similarly $\pi^+_v$ denotes the cost of sending some flow along the minimum cost $s$-$v$ path, which is positive.

We begin by making the following observations about the node potentials when the marginal costs are derived from a min-cost flow.
\begin{enumerate}
\item For the sink node $t$, $\pi^-_t = r^-(x)$ and $\pi^+_t = r^+(x)$. Indeed $r^-(x)$ and $r^+(x)$ denote the cheapest possible way of removing and sending one infinitesimal unit of flow respectively, which coincide with the definition of node potentials for the sink node $t$.

\item For every node $v$, $\pi^-_v \leq \pi^+_v$. If this is not true, then we would have identified a negative cycle in the residual graph.

\item For every edge $e=(u,v)$ in the original graph, $\pi^-_v \geq \pi^-_u + c^-_e(x_e)$ and $\pi^-_u \geq \pi^-_v - c^+_e(x_e)$.

\item For every edge $e=(u,v)$ in the original graph, $\pi^+_v \leq \pi^+_u + c^+_e(x_e)$ and $\pi^+_u \leq \pi^+_v - c^-_e(x_e)$.
\end{enumerate}

The third observation is true because the path from $v$ to $s$ via $u$ is a candidate for the path which provides the maximum `gain' in cost from $v$ to $s$. So the cost along this path is just $c^-_e(x_e) + \pi^-_u$ which provides a lower bound for the actual maximum savings path from $v$ to $s$. Similarly, returning to the residual graph interpretation, flow can be removed from $u$ to $s$ by sending some additional flow along $(u,v)$ and then removing the sent flow from $v$ to $s$. So, this path acts as a lower bound for $\pi^-_u$. A similar argument can be used to prove Observation (4). Before giving our algorithm to adjust prices along the paths, we first show how these node potentials can be applied to derive edge prices.

\begin{lemma}
\label{sublem_priceusingnodepot}
Given any set of node potentials obeying Observations (2), (3) and (4) above, if we define a quantity $\tilde{k}_e$ on every flow carrying edge $e=(u,v)$ to be $\pi^-_v - \pi^-_u$ and on non flow-carrying edges to be $c^+_e(0)=0$, then the vector of $\tilde{k}_e$'s satisfies the following conditions
\begin{enumerate}
\item For all $e$, $c^-_e(x_e) \leq \tilde{k}_e \leq c^+_e(x_e)$.
\item For any flow-carrying path $P$, $\sum_{e \in P} \tilde{k}_e = \pi^-_t$. For any non-flow carrying path, this quantity cannot be smaller.
\end{enumerate}
\end{lemma}
\begin{proof}
The first part of the lemma follows almost directly from Observation (3). First notice that any flow carrying is a simple path and thus each node must appear only once on this path. So $\sum_{e \in P}\tilde{k}_e = \sum_{e = (u,v) \in P}\pi^-_v - \pi^-_u = \pi^-_t - \pi^-_s = \pi^-_t$. Now what about non-flow carrying simple paths (paths with cycles will clearly have a larger cost)? Let $P'$ be some such path. There must be several contiguous sequences in this path consisting only of non-flow carrying edges. Let $P'' = e_1, e_2, \ldots, e_{|P''|}$ be any such contiguous segment connecting two nodes $v_i$ and $u_{i+1}$. Using the second half of Observation (3) for all nodes in this sequence, we get $\pi^-_{v_{i}} \geq \pi^-_{u_{i+1}} - 0$ since $c^+_{e_i}(x_{e_i}) = 0$ for all these edges as they do not carry flow.

Let us decompose the path $P'$ into flow-carrying and non-flow carrying segments. In particular, let $u_1, v_1, u_2, v_2  , \ldots, u_k, v_k$ denote the starting and ending nodes of every flow carrying segment in that order so that $u_1 = s$ and $v_k=t$. Then $\sum_{e \in P'}\tilde{k}_e = \sum_{i = 1}^k(\pi^-_{v_i} - \pi^-_{u_i})$. Using, our previous observation that $\pi^-_{v_i} \geq \pi^-_{u_{i+1}}$, this becomes,
$$\sum_{e \in P'}\tilde{k}_e = \sum_{i=1}^{k-1}(\pi^-_{v_i} - \pi^-_{u_{i+1}}) + \pi^-_t \geq \pi^-_t.$$
This completes the proof that all non flow-carrying paths are priced no smaller than $\pi^-_t$.
\end{proof}

Notice that for the given instance, $\pi^-_t \leq p^* \leq \pi^+_t$. If $\pi^-_t = p^*$, then by Lemma~\ref{sublem_priceusingnodepot}, we have proved our main claim since setting $\tilde{c}_e = \tilde{k}_e$ for every edge satisfies the requirements of the claim. Suppose this is not the case and we have, $\pi^-_t < p^* \leq \pi^+_t$, then we somehow need to transform the instance into a new instance satisfying all the observations above such that $\pi^-_t = p^*$. We do this via the following high-level procedure.

Consider the same graph with some given marginal costs $k^-_e$ and $k^+_e$ on the edges. At every iteration of the procedure, we choose an edge $e$ and either increase $k^-_e$ or decrease $k^+_e$. However, we ensure that at each step some invariants similar to the observations above are maintained. Changing the marginal  costs on the edges in turn changes the node potentials but we ensure that no negative cycles are created in the residual graph. We terminate the procedure when $\pi^-_t = p^*$.

We now describe the procedure exactly. Consider the same network as before and initialize $\forall e$, $k^-_e = c^-_e(x_e)$ as the negative marginal cost and $k^+_e = c^+_e(x_e)$ as the positive marginal cost. That is if we call this iteration $0$, then $k^-_e(0) = c^-_e(x_e)$ and $k^+_e(0) = c^+_e(x_e)$. Since the node potentials depend only on the negative and positive marginal costs, we can define the node potentials for this graph. At the beginning of some iteration $i+1$, let $(S,\bar{S})$ denote a cut of this graph such that $S$ is the set of nodes $v$ satisfying $\pi^-_v(i) = \pi^+_v(i)$. Clearly at initialization, $s \in S$ and $t \in \bar{S}$. Each iteration of the algorithm is described by the following steps,
\begin{enumerate}
\item Beginning of iteration $i+1$
\item Pick any edge $e=(u,v)$ of the graph going across the cut
\begin{itemize}
\item If $u \in S$ and $v \in \bar{S}$, then increase $k^-_e$ until either $\pi^-_v(i+1) = \pi^+_v(i)$ or $\pi^-_t(i+1) = p^*$. In the latter case, we terminate the algorithm.
\item If $v \in S$ and $u \in \bar{S}$, then decrease $k^+_e$ until either $\pi^-_v(i+1) = \pi^+_v(i)$ or $\pi^-_t(i+1) = p^*$. In the latter case, we terminate the algorithm.
\end{itemize}
\item Recompute the $(S,\bar{S})$-cut as defined above for the new values of $k^-_e$ and $k^+_e$.
\end{enumerate}

We now show some invariants that the algorithm maintains. Notice that the end of any iteration (say $i$) coincides with the beginning of the next iteration $(i+1)$. When we say $\pi^-_v(i)$ or $k^-_e(i)$, we refer to the values at the end of iteration $i$.

\begin{lemma}
\label{sublem_pccostinvariants}
The above algorithm maintains the following invariants at the end of every iteration
\begin{enumerate}
\item $c^-_e(x_e) \leq k^-_e \leq k^+_e \leq c^+_e(x_e)$ for all $e$.
\item For all $v$, $\pi^-_v \leq \pi^+_v$.
\item For every edge $e=(u,v)$, $\pi^-_v \geq \pi^-_u + k^-_e$ and $\pi^-_u \geq \pi^-_v - k^+_e$.
\item For every edge $e=(u,v)$, $\pi^+_v \leq \pi^+_u + k^+_e$ and $\pi^+_u \leq \pi^+_v - k^-_e$.
\end{enumerate}
\end{lemma}
\begin{proof}
We prove this by induction on the steps of the algorithm. Since we initialized $k^-_e(0) = c^-_e(x_e)$ and $k^+_e(0) = c^+_e(x_e)$, the invariants are definitely true at the base step. At each iteration, we pick a single edge going across the cut and either increase $k^-_e$ or decrease $k^+_e$, so we just have to show that increasing or decreasing these as defined in the algorithm will not violate the invariants. Suppose that the invariants above are true until the end of the $i^{th}$ iteration. First consider the case when in the $i+1$ iteration we pick some edge $e=(u,v)$ such that $u \in S$, $v \in \bar{S}$ and so we increase $k^-_e$.

At the beginning of the iteration $k^-_e(i)$ and $k^+_e(i)$ still satisfy the requirements before we increase $k^-_e$. Since we do not alter $k^-$ and $k^+$ for any other edge, we only need to show that the first invariant is not violated for $e$. Recall that by definition, $\pi^-_u(i) = \pi^+_u(i)$ and $\pi^-_v(i) < \pi^+_v(i)$. During this iteration, we increase $k^-_e$ such that the path from $v$ to $s$ along which the maximum cost can be saved must pass through $e$ and moreover, the cost saved cannot be larger than $\pi^+_v(i)$. We begin by making three simple observations that must be obeyed after the increase in $k^-_e$.
\begin{enumerate}

\item \emph{The new value of $k^-_e$, i.e., $k^-_e(i+1)$ cannot be larger than $\pi^+_v(i) - \pi^-_u(i)$.} \\This is not hard to see since when $k^-_e(i+1) = \pi^+_v(i) - \pi^-_u(i)$, we have a flow removing path from $v$ to $s$ that removes flow on edge $e$ providing a gain of $\pi^+_v(i) - \pi^-_u(i)$ and then recursively removes flow from $u$ along the maximum gain path of $u$, which provides a gain of at least $\pi^-_u(i)$. The total gain along this path is therefore $\pi^+_v(i)$ and so the algorithm must halt here.

\item \emph{At the end of iteration $i+1$, the residual graph cannot have any negative cycles and therefore all the maximum gain or minimum cost paths are simple paths.}\\
Assume recursively that the residual graph does not have any negative cycles upto iteration $i$. Now, the only edge whose marginal cost we've changed is the reverse edge $(v,u)$ whose cost went from $-k^-_e(i)$ to $-k^-_e(i+1)$. So any new negative cycle, if at all it exists now must use this edge $(v,u)$. But we claim that the rest of the cycle which is essentially some path from $u$ to $v$ cannot have a cost smaller than $\pi^+_v(i) - \pi^+_u(i)$. This is easy to see because if such a path with smaller cost exists then at the end of iteration $i$, while computing $\pi^+_v$, we could have sent flow from $s$ to $u$ and then along this path from $u$ to $v$, which is smaller than $\pi^+_v(i)$, a contradiction. So the total cost along the cycle after the increase in $k^-_e$ is not smaller than $-k^-_e(i+1) + \pi^+_v(i) - \pi^+_u(i) \geq 0$.

\item \emph{At the end of iteration $i+1$, $\pi^+_v(i+1) = \pi^+_v(i)$, $\pi^+_u(i+1) = \pi^+_u(i) = \pi^-_u(i) = \pi^-_u(i+1)$.}
The first statement that $\pi^+_v(i+1) = \pi^+_v(i)$ is easy to see. $\pi^+_v$ is the minimum cost of sending flow from $s$ to $v$ and such a path cannot send flow to $v$ and then remove flow along $e$. At iteration $i+1$, the maximum gain path for $u$ cannot remove flow along the edge $e$ since it would violate the simple path condition, so $\pi^-_u(i+1) = \pi^-_u(i)$. Suppose that $\pi^+_u(i+1)$ is no longer equal to $\pi^+_u(i)$. Then, it could have only decreased since looking at the residual graph, the only change is that $k^-_e$ became more negative. Then, it is the case that,
$$\pi^+_u(i) > \pi^+_u(i+1) = \pi^+_v(i+1) - k^-_e(i+1) \geq \pi^-_u(i).$$
This is a contradiction since at iteration $i$, $\pi^-_u(i) = \pi^+_u(i)$.
\end{enumerate}

Now we can show that Invariant (1) is true at the end of iteration $i+1$ as well.Using the fact that $\pi^-_u(i) = \pi^+_u(i)$, we can derive a bound for $k^-_e(i+1)$ as,
\begin{align*}
k^-_e(i+1) & \leq \pi^+_v(i) - \pi^-_u(i) & \\
& = \pi^+_v(i) - \pi^+_u(i) & \\
& \leq k^+_e(i) & \text{By Invariant (4), $\pi^+_v(i) \leq \pi^+_u(i) + k^+_e(i)$} \\
& = k^+_e(i+1) &
\end{align*}

So, the first invariant is not violated at the end of the current iteration. We also make a small remark about the node potentials here. Since the maximum gain path for $v$ at the end of the iteration uses edge $e$, we have
\begin{equation}
\label{eq:pccosthelper1}
\pi^-_v(i+1) = k^-_e(i+1) + \pi^-_u(i+1) \leq (\pi^+_v(i) - \pi^-_u(i)) + \pi^-_u(i) = \pi^+_v(i).
\end{equation}

Moving on to invariant (2), Equation~\ref{eq:pccosthelper1} shows that $\pi^-_v(i+1) \leq \pi^+_v(i) = \pi^+_v(i+1)$. Invariant (2) is therefore, not violated for vertex $v$ or $u$ (Observation 3). What about the other nodes? First, we show that the increase in $k^-_e$ does not lead to a change in $\pi^+_w$ for any node $w$. In other words, the cost of sending an infinitesimal unit of flow from $s$-$w$ does not change. Assume by contradiction that for some $w$, the quantity $\pi^+_w(i+1) < \pi^+_w(i)$\footnote{We have already seen from the residual graph interpretation that this cannot increase}. Then, since we changed the marginal cost for only one edge, the minimum cost $s$-$w$  path for sending flow must contain $e=(u,v)$. In particular, since we only increased the negative marginal cost of $e$, the minimum cost $s$-$w$ path must be removing flow along $e$. So, $\pi^+_w(i+1) = \pi^+_v(i+1) - k^-_e(i+1) + q$, where $q$ represents the cost of the path from $u$ to $w$. Since $\pi^+_v(i+1) = \pi^+_v(i)$ and $k^-_e(i+1) \leq (\pi^+_v(i) - \pi^+_u(i))$, we have
\begin{align*}
\pi^+_w(i) & > \pi^+_w(i+1) = \pi^+_v(i+1) - k^-_e(i+1) + q \\
& \geq \pi^+_v(i) - (\pi^+_v(i) - \pi^+_u(i)) + q \\
& \geq \pi^+_u(i) + q
\end{align*}

So, we have shown that $\pi^+_w(i) > \pi^+_u(i) + q.$ However, this contradicts the fact that we have sent flow on the minimum cost path at time $i$, since we might have as well directly sent flow from $s$ to $u$ on a path costing $\pi^+_u(i)$. This means that $\pi^+_w(i) \leq \pi^+_u(i) + q$, and so the positive potential on any node cannot change. What about the negative potential? Since we are increasing $k^-_e$, it might increase for some nodes. However, we show that Invariant(2) is still obeyed and so for every $w$, $\pi^-_w(i+1) \leq \pi^+_w(i+1) = \pi^+_w(i)$. Suppose the negative potential increases for some node $w$. Then the path from $w$ to $s$ must be removing flow on the edge $e$. So this gives us, $\pi^-_w(i+1) = q_1 + k^-_e(i+1) + q_2 \leq q_1 + (\pi^+_v(i) - \pi^-_u(i)) + \pi^-_u(i+1)$, where $q_1$ is the amount gained by removing flow from $w$ to $v$ and $q_2$ from $u$ to $s$. Now let us go back to time $i$, one possible way of sending flow from $s$ to $v$ is to send flow from $s$ to $w$ and then remove flow from the $w$-$v$ path gaining cost of $q_1$. Mathematically, this gives us $\pi^+_v(i) \leq \pi^+_w(i) - q_1$. Now we can show our desired bound, namely
\begin{align*}
\pi^-_w(i+1) & \leq q_1 + (\pi^+_v(i) - \pi^-_u(i)) + \pi^-_u(i+1)\\
& = (q_1 + \pi^+_v(i))\\
& \leq \pi^+_w(i)\\
& = \pi^+_w(i+1)
\end{align*}
This completes the proof of invariant (2).

Invariants (3) and (4) must hold, because they are a property of node potentials and do not depend on the exact value of $k^-_e$ or $k^+_e$. In particular, given any edge $e'=(w,y)$, the maximum gain from removing an infinitesimal unit of flow from $s$ to $y$ is not smaller than the gains from removing a unit of flow along the $(w,y)$ edge and then removing the flow along the path providing maximum gains for $w$. Mathematically, this translates to $\pi^-_y \geq \pi^-_w + k^-_{e'}$. Similarly given the same edge, flow from $s$ to $w$ can be removed by sending an infinitesimal amount of flow from $w$ to $y$ and then removing flow along the maximum gains path for $y$. Mathematically, this means that $\pi^-_w \geq \pi^-_y - k^+_e$. Invariant (4) can be argued similarly.

\textbf{What if $u \in \bar{S}$ and $v \in S$.}\\ The proof for this case is similar so we only sketch it. By definition, $v$ must satisfy $\pi^-_v(i) = \pi^+_v(i)$. In this case, we decrease $k^+_e$ until $\pi^-_u(i+1) = \pi^+_u(i)$ or $\pi^-_t(i+1)=p^*$. As with our previous proof, it is not hard to see that at iteration $i+1$, $\pi^+_u(i) \geq \pi^-_u(i+1) = \pi^-_v(i) - k^+_e(i+1)$. This means that, $k^+_e(i+1) \geq \pi^-_v(i) -\pi^+_u(i) \geq k^-_e(i+1)$. The last inequality comes from Invariant (4) which states that $\pi^+_u(i) \leq \pi^+_v(i) - k^-_e(i)$. Invariant (1) is therefore not violated.

Once again, $\pi^+_w(i+1) = \pi^+_w(i)$ for every node $w$, since we can show as we did previously that the minimum cost path for sending flow from $s$ to $w$  will not add any flow on the $(u,v)$ edge. Now consider $\pi^-_w(i+1)$. It will change only if the new `maximum savings' path will add flow on $e$. This means that (if $q_1$ is as defined before)
\begin{align*}
\pi^-_w(i+1) & = q_1 - k^+_e(i+1) + \pi^-_v(i) &\\
& \leq q_1 - (\pi^-_v(i) -\pi^+_u(i)) + \pi^-_v(i) &\\
& = q_1 + \pi^+_u(i)& \\
& \leq \pi^+_w(i) = \pi^+_w(i+1) & \text{Since $\pi^+_u(i) \leq \pi^+_w(i) - q_1$ by Invariant (4)}
\end{align*}
\end{proof}

Now we have shown that our algorithm has some nice properties. We just need to show that it terminates and when it does at say some step $K$, $\pi^-_t(K) = p^*$.

\begin{lemma}
\label{sublem_termination_pccost}
As long as $\pi^-_t < p^*$, there exists a cut $(S,\bar{S})$ as defined above with $t \in \bar{S}$. Moreover, at every iteration of the algorithm, for some edge $e$, there is either a non-zero increase in $k^-_e$ or a non-zero decrease in $k^+_e$.
\end{lemma}
\begin{proof}
The first part is easy to see. Recall from the proof of Lemma~\ref{sublem_pccostinvariants} that $\pi^+_v$ does not change for any node $v$. Therefore $\forall v$, at any iteration $i$, $\pi^+_v(i) = \pi^+_v(0)$. This means for the sink node $\pi^+_t(i) = r^+(x)$ at every iteration. In order for $t$ to be in the $S$ side of the cut at the end of some iteration $i$, it must be that $\pi^-_t(i) = \pi^+_t(i) = r^+(x) \geq p^*$. So the contraposition of this statement is that as long as $\pi^-_t(i) < p^*$, $t \in \bar{S}$. So a non-trivial cut must exist at every iteration of the algorithm.

Suppose at the beginning of some iteration $i+1$, $\pi^-_t(i) < p^*$, and we pick some edge $e=(u,v)$. Is it possible that $k^-_e = k^+_e$ and therefore there is a zero change in the marginal costs in this iteration? Assume by contradiction that this is the case. First suppose that $v \in \bar{S}$ and $u \in S$ so that $\pi^-_u(i) = \pi^+_u(i)$ and $\pi^+_v(i) > \pi^-_v(i)$. We know from Invariants (3) that the following must be true,

$$\pi^-_v(i) \geq \pi^-u(i) + k^-_e(i) = \pi^-u(i) + k^+_e(i)$$

and by invariant (4),

$$\pi^+_v(i) \leq \pi^+_u(i) + k^+_e(i) \leq \pi^+_u(i) + (\pi^-v(i) - \pi^-_u(i)) = \pi^-_v(i).$$

This is clearly a contradiction since $\pi^+_v(i) > \pi^-_v(i)$. Similarly suppose the edge $e=(u,v)$ is directed such that $u \in \bar{S}$. Then, assuming to the contrary that $k^-_e(i) = k^+_e(i)$, gives us a contradiction via the exact same inequalities mentioned above.
\end{proof}

Now, it is easy to why the algorithm must terminate when $\pi^-_t(K) = p^*$. At every iteration, for some edge $e$, the difference $k^+_e - k^-_e$ strictly decreases without violating $k^+_e \geq k^-_e$. And as long as $\pi^-_t < p^*$, we are guaranteed an edge satisfying $k^-_e < k^+_e$. This cannot happen forever and so the algorithm must terminate when $\pi^-_t = p^*$.

Once this set of node potentials are reached, we can just apply Lemma~\ref{sublem_priceusingnodepot} and obtain a $\tilde{c}_e = \tilde{k}_e$ for each edge as defined in the lemma. That is, the quantity $\tilde{c}_e$ on each edge $e = (u,v)$ is simply $\pi^-_v(K) - \pi^-_u(K)$. By Lemma~\ref{sublem_priceusingnodepot}, we are guaranteed that for all flow carrying paths these quantities sum up to $\tilde{c}_e$ and that $c^-_e(x_e) \leq \tilde{c}_e c^+_e(x_e)$ since our node potentials satisfy the required invariants.
\end{proof}

Now, we can immediately obtain our existence and efficiency results.

\begin{claim}
For any instance without monopolies, we can obtain a Nash Equilibrium with the optimal flow.
\end{claim}
\begin{proof}
At the optimum $x^*$, it must be true that either $r^+(x^*) \geq \lambda(x^*) \geq r^-(x^*)$ or $\lambda(x^*) \geq r^+(x^*) \geq r^-(x^*)$ and $x^*=T$. In the first case, we simply set $p^* = \lambda(x^*)$, use Claim~\ref{clm_balancingprices} and set the price on every edge $p_e = \tilde{c}_e$ as per that claim. We already that the buyer behavior is a best-response flow since the price on all flow carrying paths must sum up to $p^* = \lambda(x^*)$. No seller would wish to decrease their price since their current price $p_e \leq c^+_e(x_e)$ and so by Lemma~\ref{cl_margcost}, an increased flow cannot lead an increase in profits. It is also not hard to see that for every edge $e$, there must be a shortcut not containing this edge and having the same cost as the flow containing sub-path. The argument for this is similar to what we showed in Claim~\ref{clm_monopolypos1} since even in this case, $p_e = c^+_e(0) = 0$ for all edges without flow.

What if $\lambda(x^*) > r^+(x^*)$. In this case, we set $p^* = r^+(x^*)$ and use Claim~\ref{clm_balancingprices} and once again price edges to be $p_e = \tilde{c}_e$. 
\end{proof}

We now summarize the properties of the Nash Equilibrium whose existence we just proved. 

\begin{claim}
For instances with MPE demand $\lambda$ and $M\geq 1$, $\exists$ a Nash Equilibrium which is either optimal or satisfies
$$ \lambda(\tilde{x}) - r^+(\tilde{x}) \leq M\tilde{x}|\lambda'(\tilde{x})| \leq \lambda(\tilde{x}) - r^-(\tilde{x}).$$
Moreover, this equilibrium obeys all our desired results: Non-Trivial Pricing, Recovery of Production Costs, Pareto-Optimality over the space of equilibria and Local Dominance.
\end{claim}
\begin{proof}
First suppose that the optimum solution $x^*$ satisfies, $Mx^*|\lambda'(x^*)| \leq \lambda(x^*) - r^-(x^*)$ and that $r^-(x^*) \leq r^+(x^*) \leq \lambda(x^*)$. In the first case, we set $p^*=r^-(x^*)$ and use Claim~\ref{clm_balancingprices} to obtain $\tilde{c}_e$ on all edges. These quantities satisfy the requirements of Lemma~\ref{lemma_eqnconditionsoptpccost} and so we can set prices as mentioned in that Lemma to obtain a Nash Equilibrium.

Next if the optimum still satisfies $Mx^*|\lambda'(x^*)| \leq \lambda(x^*) - r^-(x^*)$ but $r^-(x^*) \leq \lambda(x^*) \leq r^+(x^*)$. Then we set $p^* = \lambda(x^*)$, run the algorithm in claim~\ref{clm_balancingprices} and price each edge as mentioned in Lemma~\ref{lem_eqnconditionspccost1}. This clearly is a Nash Equilibrium.

Finally, what if $Mx^*|\lambda'(x^*)| > \lambda(x^*) - r^-(x^*)$? Then, we claim that for MPE functions, there must exist some $\tilde{x} > 0$ such that $\lambda(\tilde{x}) - r^+(\tilde{x}) \leq M\tilde{x}|\lambda'(\tilde{x})| \leq \lambda(\tilde{x}) - r^-(\tilde{x}).$ This is true because as $x \to 0$, $Mx|\ldx| < \lx - r^+(x)$ and so there must exist an intermediate $\tilde{x} > 0$ satisfying the requirements. Now, we set $p^*=\lambda(\tilde{x}) - M\tilde{x}|\lambda'(\tilde{x})|$ and run the algorithm of Claim~\ref{clm_balancingprices} to obtain quantities $\tilde{c}_e$ on all edges. These quantities clearly fulfill the requirements of Lemma~\ref{lem_eqnconditionspccost1} so we can price edges as mentioned in the lemma to obtain a Nash Equilibrium. This completes our existence proof.
\end{proof}

All our bounds for the efficiency hold since the the only property that we require about the equilibrium is given by
$$\lambda(\tilde{x}) - r^+(\tilde{x}) \leq M\tilde{x}|\lambda'(\tilde{x})| \leq \lambda(\tilde{x}) - r^-(\tilde{x}).$$

While calculating the social welfare of the Nash Equilibrium, we have so far used the fact that the integral of $\lx - r(x)$ from $0$ to $\tilde{x}$ is at least $\tilde{x}(\lambda(\tilde{x}) - r(\tilde{x}))$. Here in this case, we can use the fact that the integral is at least $\tilde{x}(\lambda(\tilde{x}) - r^-(\tilde{x}))$ since $r(x)$ is still non-decreasing so we can do the integration piecewise. We do not reprove all our theorems here as the proofs are essentially the same.

\subsection{Proofs from Section 5: Looking Beyond Graphical Markets}
Before showing the more complex proof that our results hold even when buyers demand bundles not coming from a graph, we show that in markets where all the buyers have the same monotone valuation function, efficient equilibrium is guaranteed.

\subsection*{Uniform Demand Markets with arbitrary combinatorial valuations}

\begin{thm_app}{thm_gen_uniformdemand}
In any setting where uniform buyers have arbitrary combinatorial valuations, there exists an efficient Nash Equilibrium.
\end{thm_app}

\begin{proof} Let $\vec{x^*}$ be the optimum allocation vector with the maximum possible allocation $x^*$. The following must be true for any bundle $S$ with $x^*_S > 0$ and any other bundle $S'$ with $x^*_{S'} \geq 0$.
$$v(S) - \sum_{e \in S}c_e(x^*_e) \geq v(S') - \sum_{e \in S'}c_e(x^*_e).$$

Let $\mathbb{M}$ be the set of monopolies and virtual monopolies (VM) at the optimum: these are the items that belong to all bundles with non-zero allocation at the optimum. In particular, we claim that the following simple algorithm applied for the optimal allocation $\vec{x^*}$ gives us a Nash Equilibrium. Let $B$ be the set of bundles with non-zero consumption and $B'$ be the bundles with zero consumption. Begin with a price of $p_e^0 = c_e(x^*_e)$ on each item. We now describe the ascending price process.
\begin{enumerate}
\item Pick any item $e$, which belongs to all optimal bundles $B$.
\item Increase the price on $e$ such that for all bundles $S$ in $B$, either $v(S) - \sum_{e \in S}p_e = 0$ or $\exists$ some $S' \in B'$, such that $v(S') - \sum_{e \in S'}p_e = v(S) - \sum_{e \in S}p_e$, i.e., $S'$ is now as good as any optimal bundle.
\item Repeat this process until you cannot increase the price on any VM anymore.
\end{enumerate}

We claim that the above algorithm along with the flow $\vec{x^*}$ gives us a Nash Equilibrium. The following observations are pertinent here: (i) At any step of the algorithm, all bundles in $B$ give the same utility (ii) At any step in the above algorithm, no bundle in $B'$ gives more utility than any bundle in $B$ (iii) Non-monopoly items are priced at $c_e(x^*_e)$ throughout. The first point is true because the VMs belong to all bundles in $B$, so increasing the price on such an edge leads to an equal increase in the price of any such bundle. Moreover, initially all such bundles are priced such that the utility $v(S) - \sum_{e \in S}p_e^0$ is the same. For any item $e$, we stop the increase when some bundle $S'$ in $B'$ will become better than the bundles in $B$. This means that $e \notin S'$.

Now once again, it is obvious by Lemma~\ref{cl_margcost} that no non-monopoly edge will increase or decrease its price. Suppose that the final price on any item $e$ is $p^*_e$. By point (ii) above, no bundle with zero allocation provides more utility than the bundles in $B$, so buyers are choosing best-response bundles. We only need to focus on seller behavior. We know $\forall S, \tilde{S} \in B$, $v(S) - \sum_{e \in S}p^*_e = v(\tilde{S}) - \sum_{e \in \tilde{S}}p^*_e \geq 0$. Suppose this quantity equals $0$, then no seller will increase the price on his item, because the bundles containing that item will give strictly negative utility and buyer will never purchase such a bundle. If the utility is strictly larger than zero, then it means the algorithm terminated because no VM could increase its price further without losing all its buyers to some bundle in $B',$ so in this case sellers would not increase their price, or they would lose everything.

Can sellers lower their price? We claim that this cannot happen. Suppose that the final price of all sellers is $p^*_e = p_e^0$. In this case, all items are priced at their marginal so no one would wish to lower their price as per Lemma~\ref{cl_margcost}. If even one VM is priced at $p^*_e > p_e^0$, this means that initially $v(S) - \sum_{e \in S}c_e(x^*_e) > 0$ and $x^* = T$, all the buyers have received some bundle so no buyer remains unallocated. This is true because if $v(S) - \sum_{e \in S}c_e(x^*_e) > 0$ for some bundle $S$ and there are still some buyers with no allocation, then if we allocate the bundle $S$ to an infinitesimal amount of buyers, we could receive additional non-negative utility of $v(S) - \sum_{e \in S}c_e(x^*_e)$ (Since $c_e(x)$ are all continuous). This violates the optimality of $x^*$ and therefore it must be that there are no unallocated buyers left.

So, decreasing the price has no effect. We therefore conclude that sellers have no incentive to change their price and ergo the prices given the algorithm along with the optimal flow $\vec{x^*}$ is a Nash Equilibrium optimizing social welfare.
\end{proof}

\subsection{General markets with non-graphical bundles}
We begin by formally redefining our two-stage pricing game to the case where buyers may desire arbitrary bundles and not just $s$-$t$ paths. Once again, let $E$ be any set of items. We use $\mathcal{B} \subseteq 2^E$ to represent the bundles desired by all the buyers. Buyers are indifferent between these bundles so we can use an inverse demand function $\lx$ such that at least $x$ of the population holds a value of $\lx$ or more for the bundles. We no longer assume that $c_e(0) =0$ for any of the items.

\textbf{Virtual Monopolies (VM)} A small discussion on virtual monopolies is in order here. Given any solution to our two-stage pricing game $(\vec{p}, \vec{x})$, an item $e \in E$ is called a virtual monopoly if it belongs to all the bundles consumed by the buyers, i.e., $e \in B_i$ for all $B_i\in\mathcal{B}$ with $x_{B_i} > 0$. Note that $e$ need not necessarily belong to all buyer bundles in $B$, however due to the difference in production costs and the nature of the bundles, the item belongs to all the bundles that are actually consumed by the buyers. Therefore, in that particular solution of the market game the item has the power of a monopoly: it can increase its price (up to some extent) and not lose its flow due to the absence of competition. We call such items ``Virtual Monopolies" keeping with the notation used in~\cite{chawla2008bertrand}.

Clearly, any `pure monopoly' (item belonging to all desired bundles) is also a virtual monopoly. For markets with graph structure obeying $c_e(0)=0$, the set of pure and virtual monopolies coincide. This is no longer true for general markets.

Before showing that all our techniques extend to general markets, we first redefine our pricing rule to apply for the case with virtual monopolies. However, we do not directly define a closed form expression as in Theorem~\ref{thm:existence}. Instead for any given minimum cost allocation $\vec{x}$, we define an ascending price process (Algorithm \ref{alg}) beginning with all items priced at their marginal cost. The difficulty in obtaining a closed form expression is that the price of a VM cannot be arbitrarily increased since there are usually higher priced alternative bundles that do not contain this item. The following ascending price process returns the prices on each item given by $p_e(x)$. Let $\mathcal{M}$ be the set of items monopolizing the allocation $\vec{x}$ and let $M$ be the cardinality of this set. Our complete existence proof depends on a series of properties we prove about our algorithm culminating in showing that the prices returned vary continuously with the input $x$ for every single good. 

\begin{algorithm}[tbp]
\caption{Ascending Price Algorithm for Prices}\label{alg}
\begin{algorithmic}[1]
\REQUIRE A min-cost allocation $\vec{x}$ of magnitude $x$, $\lx$, $c_e(x_e)$.
\STATE Let the initial price on each item be $p^0_e=c_e(x_e)$.
\STATE Let $M_A$ be the set of active monopolies initialized to $\mathcal{M}$ and $M_I$ be the set of inactive monopolies, initially empty.
\STATE Let $B$ be the set of bundles with non-zero allocation and $B'$ be the bundles with zero allocation.
\STATE Increase the price on all active monopolies uniformly until
\begin{enumerate}
\item One or more of the monopolies become tight, i.e., there now exists some monopoly $e \in M_A$ such that some bundle in $B'$ that does not contain $e$ now has the same price as all bundles in $B$.
\item The price of all optimal bundles is now exactly $\lambda(x)$. In this case, exit the algorithm.
\end{enumerate}
\STATE In the first case, remove all the `tight' monopolies from $M_A$ and add them to $M_I$ and repeat the above step.
\STATE The algorithm terminates when either the total price on the bundles equals $\lx$ or the set $M_A$ is empty.
\end{algorithmic}
\end{algorithm}

\begin{lemma}
\label{lem_genmarketsprice}
The following is true for our algorithm
\begin{enumerate}
\item At any time step $t$, all bundles belonging to $B$ have the same price (say $p^t$).
\item At any time step $t$, no bundle belonging to $B'$ can have a price smaller than that of any bundle in $B$.
\item For every good $e$ that belongs to $M_I$ at some time step $t$, there exists some bundle in $B'$ not containing $e$ that has the same price as all bundles in $B$ at that time step.
\end{enumerate}
\end{lemma}

\begin{proof}
We prove these claims inductively. All the claims are true at the initial time step $t=0$ when edges are priced at their marginal cost. Indeed all bundles belonging to $B$ must necessarily have the same marginal cost (optimality conditions). Moreover, no bundle in $B'$ can have a smaller marginal cost since these bundles are not consumed. If any virtual monopoly is inactive, then clearly there must exist some bundle in $B'$ not containing $e$ with the same marginal cost as all bundles in $B$.

Suppose that all the claims are true up to some time $t-1$. After this time step we increase the price uniformly for all VMs still belonging to $M_A$. Note that since these are virtual monopolies, they belong to all bundles in $B$ with non-zero allocation. Therefore, the price of all the bundles in $B$ increases uniformly and thus our first claim must be true at time $t$. Now is it possible that increasing the price for VMs in $M_A$ can make some bundle in $B'$ cheaper than the optimal bundles? This is not the case. By definition any bundle in $B'$ that has the same price as the bundles in $B$ must contain all $e \in M_A$. If not, then the VM $e$ would have become tight and we would have shifted it to $M_I$. So this means that increasing the price on the active monopolies also leads to an uniform increase in the price of bundles belonging to $B'$ that have the same price as the optimal bundles. Moreover, we stop the increase when some new bundle in $B'$ has the same price as bundles in $B$ but does not contain some $e \in M_A$. This becomes time step $t$. So claim 2 is also true.

Notice that for every monopoly $e \in M_I$ at time $t-1$, there must exist some bundle $B_1 \in B'$ such that the price of $B_1$ at $t-1$ equals the price of any bundle in $B$. By the same reasoning as before, all VMs in $M_A$ at time $t-1$ are present in $B_1$. Therefore, since these are the only goods whose prices are increased, the price of $B_1$ at time $t$ is still the same as the price of the optimal bundles so $e$ is still `tight'. Also, if some VM newly becomes tight at time $t$, then by definition, we transfer it to $M_I$ and there must exist some bundle not containing this item with the same price as the optimal bundles. So claim 3 is true as well.
\end{proof}

Notice that when we refer to a time step $t$, it is the instant after all the tight monopolies have been transferred to $M_I$. The next corollary follows from applying the lemma at the termination step.
\begin{corollary}
\label{corr_tightness}
For a given $x$, when the algorithm terminates, the following must hold.
\begin{enumerate}
\item All bundles with non-zero allocation have the same price (Let's call this $P(x)$).
\item No bundle with zero allocation has a price smaller than $P(x)$.
%\item For every inactive monopoly $e \in M_I(x)$, there exists some $B_1 \in B'$ not containing $e$ whose price is $P(x)$.
\end{enumerate}
\end{corollary}

We now show that for a given $x$, the final prices given by the algorithm our independent of which min-cost flow is actually used. We begin by establishing that the value of the differential cost function
$c_e(x_e)$ is the same for all min-cost allocations of a given magnitude $x$. 

\begin{lemma}
\label{lem_genbundle1}
As long as the production costs are convex, the set of min-cost allocations of magnitude $x$ is closed and convex.
\end{lemma}

\begin{lemma}
\label{lem_genbundle3}
For any two minimum-cost allocations of magnitude $x$, say $\vec{a}$ and $\vec{b}$, and any good $e$, $c_e(a_e) = c_e(b_e)$, where $c_e$ denotes the marginal cost.
\end{lemma}
\begin{proof}
If for some edge $e$, $a_e = b_e$, then the lemma trivially holds. So without l.o.g, we assume that $a_e < b_e$. First consider any $\vec{x} = \alpha(\vec{a}) + (1-\alpha) \vec{b}$ for some $\alpha > 0$. Clearly $\vec{x}$ is also a minimum cost allocation. Moreover, we know $C(\vec{x}) = \alpha C(\vec{a}) + (1-\alpha)C(\vec{b})$ and so we have,
$$\sum_{e \in E}C_e(\alpha a_e + (1-\alpha)b_e) = \sum_{e \in E}\alpha C_e(a_e) + (1-\alpha)C_e(b_e).$$
But each term in the left hand side is no larger than the corresponding term in the right hand side for convex $C_e$. This means that the equality can hold iff $\forall e$, $C_e(\alpha a_e + (1-\alpha) b_e) = \alpha C_e(a_e) + (1-\alpha)C_e(b_e)$. Moreover this must be true for all $\alpha \in (0,1)$. So this is only possible if the derivative of $C_e$ is constant in this region. In other words, $c_e(x_e)$ must have the same value for all $x$ in $[a_e,b_e]$. 
\end{proof}

We now show that this differential cost of a good varies continuously as we vary $x$.

\begin{claim}
\label{corr_margcontinuity}
For any $e$, the marginal cost of the item at the min-cost allocation is a continuous function of $x$, i.e., $c_e(x_e)$ is continuously changing as we vary $x$ and compute the min-cost allocation. Therefore, $c_{B_i}(x)$ is also continuous for all bundles $B_i$.
\end{claim}
The second part of the claim follows trivially from the first since the marginal cost of a bundle is simply the sum of marginal costs of the goods constituting it. Furthermore, the sum of a finite number of continuous functions is continuous. So we only prove the first part here.
\begin{proof}
We begin with a simple but fundamental result about minimum cost allocations that we will need later, namely that the limit of a sequence of minimum cost allocations of converging magnitude is also a min-cost allocation. i.e,

\begin{lemma}
\label{lem_convergingalloc}
Let $(\vec{X}_1, \vec{X}_2, \ldots, \vec{X}_n)$ be a sequence of min-cost allocations of magnitude $x_1, x_2, \ldots, x_n$ approaching $X$ as $n \to \infty$. If $\lim_{n \to \infty}\vec{X}_n = \vec{X}$, then $\vec{X}$ is a min-cost allocation of magnitude $X$.
\end{lemma}
The proof follows from the basic continuity properties of flows.

Now, we prove yet another fundamental result: for any given sequence of positive real numbers approaching $X$, there must exist a subsequence also approaching $X$ such that the min-cost allocations at each point along this subsequence approaches a min-cost allocation of magnitude $X$. As a first step towards this, we prove a simpler lemma that implies that for any single good $e$, there must be a subsequence along which the total allocation of that good converges to a finite value.

\begin{lemma}
\label{lem_allocexistence}
Let $(x_1, \ldots x_n)$ be a sequence of positive real numbers (denoting the total allocation) approaching $X$ as $n \to \infty$. Then for any given good $e$, there must exist
\begin{enumerate}
\item a value $X_e$,
\item a subsequence $(y_1, \ldots, y_n)$ also approaching $X$ as $n \to \infty$,
\item min-cost allocations $(\vec{Y_1}, \ldots, \vec{Y_n})$ where $\vec{Y_n}$ is a min-cost allocation of magnitude $y_n$
\end{enumerate}
 such that $\lim_{n \to \infty}(Y_{n})_e = X_e,$ where $(Y_{n})_e$ is the allocation of good $e$ in $\vec{Y_n}$.
\end{lemma}
We also assume that the sequence $(x_1, \ldots, x_n)$ is bounded from above by (say) the magnitude of the optimal allocation $x^*$. This is a reasonable assumption which we will not violate when we invoke this lemma.
\begin{proof}
We prove the result using a fundamental definition of limits. Suppose that $\vec{X_n}$ is any given minimum cost allocation of magnitude $x_n$. Then, the allocation of good $e$ $(X_n)_e$ must be finite, non-negative and bounded. This is because the allocation of any good cannot be larger than the overall allocation $x_n$ which in turn is no larger than $x^*$. Now, at each point $x_n$, fix one of the many possible min-cost allocations (arbitrarily) and look at the allocations of good $e$ at each of these min-cost allocations, i.e., $((X_1)_e, \ldots (X_n)_e)$.

Then, for any given $\epsilon > 0$, there must exist at least one subsequence of $(x_1, \ldots, x_n)$ also approaching $X$ and some positive value $f$ such that at all points in this subsequence, the total allocation of good e in the given min-cost allocations lies between $[f-\epsilon, f+\epsilon]$. Now take the set of all such subsequences where the allocation falls within this region and recursively apply this idea for a smaller and smaller $\epsilon > 0$. In the limit, there must exist at least one such subsequence, which converges to a single allocation value of $X_e = f$.
\end{proof}

The rest of the proof to show the existence of a converging sequence of allocations is easy to see. Apply the above lemma for any given item, to obtain a sequence approaching $X$ and some potential min-cost allocations at each point satisfying the limit conditions for that item. Now, use the obtained subsequence as input for the above lemma once again but this time with a different item to obtain a subsequence where the allocations of two different items converge. Recursively applying the above lemma, we can obtain a sequence where the allocation of every good converges to some allocation of magnitude $X$. By Lemma~\ref{lem_convergingalloc}, we know that the allocation obtained in the limit has to be a min-cost allocation.

Finally, we are in a position to prove our main claim. As we did previously, we assume that $c_e(X_e)$ is the marginal cost of good $e$ for any min-cost allocation of magnitude $X$. Assume by contradiction that $\lim_{x \to X}c_e(x_e) \neq c_e(X_e)$ for some good $e$. Then without loss of generality, there must exist some $\epsilon > 0$ such that for every small interval of size $\delta$ around $X$, there must exist at least one point where the value of $c_e(x_e) > c_e(X_e) + \epsilon$ or $c_e(x_e) < c_e(X_e) - \epsilon$. Pick some sufficiently small $\delta$ and let $K_\epsilon$ represent the set of points in a region of size $\delta$ around $X$ where $c_e(x_e)$ does not fall within a $\epsilon$ window. Clearly the set of points $K_\epsilon$ when sorted must converge to $X$.

By the result of Lemma~\ref{lem_allocexistence} and its following reasoning, there must exist a subsequence of $K_\epsilon$ where the valid min-cost allocations converge to (say) a min-cost allocation $\vec{Y}$ of magnitude $X$. By the uniqueness of the marginal cost at $X$ as per Lemma~\ref{lem_genbundle3}, we know if the allocation of $e$ according to $\vec{Y}$ is $(Y)_e$, then $c_e((Y)_e) = c_e(X_e)$. This means that along this subsequence of $K_\epsilon$, the value of the differential cost of good $e$ converges to $c_e(X_e)$. But by construction, at every point $x$ contained in $K_\epsilon$, it must either be that $c_e(x_e) > c_e(X_e) + \epsilon$ or $c_e(x_e) < c_e(X_e) - \epsilon$, which means no subsequence should be converging to $c_e(X_e)$, a contradiction. So, $\lim_{x \to X}c_e(x_e) = c_e(X_e)$. This completes the proof of the claim. $\qed$\end{proof}

\begin{claim}
\label{clm_allmincostflows}
Our pricing rule returns the same set of prices $(\vec{p})$ on every minimum cost allocation of magnitude $x$.
\end{claim}
\begin{proof}
Intuitively this is true because for all minimum cost flows of a given magnitude and any given edge $e$, $c_e(x_e)$ is uniquely defined. Suppose that $S$ is the set of all minimum cost allocations of this given magnitude. First if for all $\vec{x} \in S$, an item $e$ does not belong to $\mathcal{M}$, its price is $c_e(x_e)$ (which is a fixed for a given $x$). So for such items, the pricing rule returns the same price for all allocations.

For every min-cost allocation, the marginal cost of every bundle remains the same. This is because the marginal cost of a bundle is simply the sum total of the marginal cost of all the items belonging to that bundle. Therefore, every iteration of our algorithm must be exactly the same for all such minimum cost allocations. Namely, any good $e$ that becomes inactive at time step $0$ for any one allocation $\vec{x}$ must do so for all allocations. This is because for $\vec{x}$, there exists a bundle not containing $e$ with the same marginal cost as the bundles $e$. Moreover the marginal costs are constant across all min-cost allocations.

Notice now that the set of active monopolies for which we increase the price remains the same for all allocations. Moreover, the initial price (marginal cost) of all bundles are also the same across allocations. This means that
(inductively) at ever time step $t \leq K$,  all goods and bundles have the same price and the same set of active monopolies become inactive. Therefore, we conclude that the algorithm terminates at the exact same set of prices for all minimum cost allocations of a given magnitude.
\end{proof}

Now we know that the prices returned by our pricing rule do not really depend on which flow is input and only depends on the magnitude of $x$, i.e., we can use the term $p_e(x)$ without any ambiguity. We also define a quantity $\pe(x)$ to be the increase in the price of an item from its marginal cost, i.e., $p_e(x) = c_e(x_e) + \pe(x)$. In order to show our main existence result, we require the fact that as we vary $x$ and compute the prices returned by our algorithm, $\pe(x)$ and hence $p_e(x)$ varies continuously. In a first step towards showing this, we define some properties that a single good can obey at a given value of $x$ and show that as long as these properties are fixed in a region, the pricing rule behaves nicely.

\begin{definition}{\textbf{(Monopolies and Non-monopolies)} }
A good $e$ is said to be a virtual monopoly for a given value of $x$ if there exists at least one min-cost allocation where the total allocation of $e$, $x_e = x$.

A good $e$ is said to be a non-monopoly for a given value of $x$ if the price of the good using our pricing rule is $c_e(x_e)$, its marginal cost.
\end{definition}

Notice that the set of monopolies and non-monopolies need not be disjoint, for a given $x$, a single good {\em could belong to both}.

\begin{definition}{\textbf{(Active Monopolies and Inactive monopolies)} }
For a given $x$, a monopoly $e$ is said to be an active monopoly if upon running our pricing rule, it belongs to the active set when the algorithm terminates.

For a given $x$, a monopoly $e$ is said to be an inactive monopoly if given the prices returned by our algorithm, there exists a bundle $B_1$ not containing $e$ such that the price of this bundle equals the price of any optimal bundle containing $e$.
\end{definition}

Once again, these two sets can overlap: this can happen for example if the termination condition for the algorithm occurs simultaneously with $e$ becoming tight.

\begin{definition}{\textbf{(Rank and tight bundles)} }
For a given $x$, and a set of inactive monopolies, the position or rank of each monopoly $e$ in the set is a number greater than or equal to one that denotes when in the process of the algorithm, it becomes active. If several items become inactive at the same time, then the rank includes all of those positions. For instance if three monopolies become tight at the beginning, then they all have rank $1$, $2$, and $3$; just as $e$ can simultaneously be both active and inactive, it can also have several simultaneous ranks.

For any given inactive monopoly $e$, a bundle $B_i \not\owns e$ is said to make $e$ tight if at the point in the algorithm where $e$ first becomes inactive, the price of the bundle equals the price of any bundle with non-zero allocation containing $e$.
\end{definition}

\textbf{Profile Space} For a given value of $x$, an item could have several valid possible property sets. For instance it could be a monopoly, an inactive monopoly, have a rank $k$ and some bundle that makes it inactive. For the same value of $x$, it could be considered a monopoly, an active monopoly and have no rank or bundle. We aggregate the valid properties of all items into a single profile vector, namely a valid profile vector is a $4$-tuple $\vec{v}=(\mathcal{M}, M_A, \mathcal{R}, B)$, where $\mathcal{M}$ is the set of monopolies, $M_A\subseteq\mathcal{M}$ is the set of active monopolies, $\mathcal{R}$ is a set of monopolies and the order in which they become tight, and for every inactive monopoly, there must be one bundle listed in $B$ that makes it inactive. Here $\mathcal{R}$ consists exactly of $\mathcal{M} \setminus M_A$, the set of inactive monopolies. The set of all such profiles is the profile space.

A profile vector $\vec{v}$ is said to be {\em consistent} at a given value of $x$ if there exists a valid min-cost allocation $\vec{x}$ with a corresponding set of prices $\vec{p}$ returned by our pricing rule where,
\begin{enumerate}
\item For every $e \in \mathcal{M}$, we have $x_e = x$ and for every $e \notin \mathcal{M}$, we have $p_e = c_e(x_e)$.
\item Every $e \in M_A$ belongs to $\mathcal{M}$ and it satisfies the active monopoly property. For every $e \in \mathcal{M} \setminus M_A$, there must exist   a bundle not containing $e$ and having the same price as the allocated bundles.
\item For every $e \in \mathcal{M} \setminus M_A$, its rank in $\mathcal{R}$ is one of its valid ranks during the run of the algorithm. In other words, there is some tie-breaking for when several monopolies become tight simultaneously which will make $e$ become tight in the position given by $\mathcal{R}$.
\item For every $e \in \mathcal{M} \setminus M_A$, the certificate bundle given by $B$ really does make it tight during the run of our algorithm (even though other bundles may have become tight at the same time).
\end{enumerate}

We define $S_x$ to be the set of all consistent profile vectors for an allocation of magnitude $x$. Now, finally we can define sets of points where our pricing rule shows nice properties. Formally, define the set $I_v$ with respect to a profile vector $v$ as the set of all $x$ such that $\vec{v} \in S_x$. We also define an additional quantity $\Gamma(x)$ to be the quantity $\pe(x)$ for all the active monopolies in $M_A$ for a given interval $I_v$. Notice that this quantity is exactly the same for all $e \in M_A$. This is true because the increase in price is uniform from time step $0$ until time step $K$ for these items.

Now, in order to prove our existence result, we need to show that the prices returned by the algorithm vary continuously with $x$. As a first step towards this goal, we show some nice properties obeyed by the prices inside each $I_v$ and the fact that the $I_v$'s are closed. 

\begin{claim}\label{lem_ivcontinuous}
\textbf{(Continuity)} The price of every item is continuous within $I_v$. Formally, for a given profile vector $\vec{v}$, suppose $\exists$ an infinite sequence of points $(x_1, x_2, \ldots, x_n)$ converging to $X$ as $n \to \infty$ all belonging to $I_v$. Moreover if $X \in I_v$, then $\lim_{n \to \infty}p_e(x_n) = p_e(X)$ for all items. \end{claim}
\begin{proof} Since the set of monopolies remains consistent in this sequence for all non-monopoly goods $e$, the price on the non-monopolies is $c_e(x_e)$ which is continuous by Claim~\ref{corr_margcontinuity}. Notice also that in this sequence all inactive monopolies become tight in the exact same order and the same active monopolies remain active. Now, notice that for a fixed $x_i$, we can express the price of any good $e$ as a simple linear combination of the marginal costs of various bundles. For instance if $e_1$ is the item that becomes inactive first (position 1) and the bundle not containing $e_1$ that has the same price as the optimum bundles at time $t=1$ is $B_1$, its given certificate. Suppose that $B^*$ is a bundle with non-zero allocation. Then, the increased price of $e_1$ as a function of $x_i$ in this sequence is simply, $\bar{p}_{e_1}(x_i)= \frac{c_{B_1}(x_i) - c_{B^*}(x_i)}{|\mathcal{M}| - |\mathcal{M} \cap B_1|}$, where $|\mathcal{M}|$ is the cardinality of the set of monopolies.

Notice that since $\mathcal{M}$ is fixed in this interval, the price depends only on the marginal costs, which are all continuous. That is, for all bundles $\lim_{n \to \infty}c_B(x_n) = c_B(X)$. So, $\bar{p}_{e_1}(x)$ is continuous for this sequence. Similarly if $e_2$ is some item that becomes tight at position $2$, then we can express $\bar{p}_{e_2}(x_i)$ as $\bar{p}_{e_1}(x_i)$ plus a function of the marginal costs, which is just a more complex expression but still a linear function of the marginal costs. Continuing this way, it is not difficult to see that all prices are continuous as $x_i \to X$.
%Notice that since the order is fixed , we can define sets $\mathcal{M}_{-k}$ which is the set of initial monopolies with the first k positions removed. This way, we can extend the reasoning for the items that become tight at every instant. Suppose there are no active monopolies in this interval, then we can conclude that $\pe(x)$ is continuous for all edges.

What if there are active monopolies? Then, suppose the algorithm proceeds for $K-1$ rounds and in the last iteration, the price of the active monopolies are increased until $P(x_i)$ becomes equal to $\lambda(x_i)$. In that case, if the price of the optimal bundles at the beginning of round $K$ is $p^{K-1}_{B^*}(x_i)$, then the increase in the final iteration for any one active monopoly $e$ is simply $\frac{\lambda(x_i) - p^{K-1}_{B^*}(x_i)}{\mathcal{M}_{K-1}}$ . Note that  $\mathcal{M}_{K-1}$ is just the set of active monopolies plus the ones which become inactive at the very end. $\lambda(x_i)$ tends to $\lambda(X)$ as $i \to \infty$ by definition and the second term can be broken down into a linear function of the marginal costs and the increase in each previous iteration, which are all continuous.

The total price of any VM is simply $c_e(x_i) + \pe(x_i)$ which is a continuous function since the marginal costs are assumed to be continuous. Since the price on every good is continuous, the price on every bundle is also continuous. Finally since the set of active monopolies at termination $M_A(x)$ is the same in this interval (or at least there is a tie-breaking rule which will make it the same), then the increased price on these fixed items $\Gamma(x)$ is clearly continuous.
\end{proof}

For the rest of this section, we will use $M^k$ to denote the set of active monopolies at the end of round (time step) $k$ of our algorithm and $p^k_B(x)$ to represent the price of any bundle $B$ at the same time instant. The following corollary gives us a closed form expression for the increase in price of all monopolies during any one round. The reasoning follows from the arguments made in Claim~\ref{lem_ivcontinuous}.

\begin{corollary}
\label{corr_closedformprices}
Suppose that for a fixed $x$, a good $e$ becomes tight in round $k+1$ during the course of our algorithm. Also suppose that $B_{k+1}$ is its witness bundle and $B^*$ is any optimal bundle. Then the price increase of all active monopolies during round $k+1$ is given by
$$\frac{p^k_{B_{k+1}}(x) - p^k_{B^*}(x)}{M^k - M^k \cap B_{k+1}}.$$
\end{corollary}

\begin{claim}\label{lem_ivclosed}
For any given profile vector, the set $I_v$ is closed. Formally suppose there $\exists$ a sequence $(x_1, \ldots, x_n)$ all of which belong to $I_v$ for some profile vector $\vec{v}$. If $\lim_{n \to \infty}x_n = X$, then $X \in I_v$.
\end{claim}
\begin{proof} We break this proof down into several parts, each identifying one new boundary event that could occur at $X$, potentially causing $X$ to not belong to $I_v$. However, we show that whatever this boundary event can be, the profile vector $\vec{v}$ still has to be active at $X$ by continuity arguments. We will be using the closed form prices given by Corollary~\ref{corr_closedformprices} throughout this proof. Since each profile vector $\vec{v}$ can be characterized as a $4$-tuple, the only possible new events are
\begin{enumerate}

\item An item $e$ stops being a monopoly or a new item becomes a monopoly.

\item The order in which some (inactive) monopolies become tight changes.
\item A new witness bundle makes a given inactive monopoly tight.

\item The set of active monopolies changes.
\end{enumerate}

\begin{itemize}
\item \textbf{Monopolies} We want to show that if $e \in \mathcal{M}$ in profile $v$ which is a valid profile along the sequence, then there exists some valid flow at $X$, where $e \in \mathcal{M}$. At every point in the infinite sequence $(x_1, \ldots, x_n, \ldots)$, $e \in M$ and at each of these points there exist a set of valid allocations where $e$ is a monopoly. From Lemma~\ref{lem_genbundle1} we know that the limiting case of a minimum cost allocation is also a minimum cost allocation.  Therefore, the allocation of item $e$ at some minimum cost allocation of magnitude $X$ is given by $X_{e} = \lim_{n \to \infty}x_{n_e} = \lim_{n \to \infty}x_n = X$. So for the given sequence of min-cost allocations, the limiting min-cost allocation also has $e$ to be a monopoly.

\item \textbf{Non-Monopoly}  This half of the claim is easier. If $e \notin \mathcal{M}$ in this sequence, it means that for every $x_n$, $\exists$ some bundle whose marginal cost is equal to the marginal cost of any min-cost bundle carrying $e$. If $e$ has zero allocation in the sequence, then the marginal cost of optimal bundles can be strictly larger than that of the bundles containing $e$.

Consider the bundle $B'$ minimizing $\lim_{n \to \infty}c_{B_i}(x_n)$ over all such bundles that do not contain $e$. Let $B_e$ be some bundle with non-zero allocation at $X$ containing $e$. Clearly $\lim_{n \to \infty}c_{B_e}(x_n) \geq \lim_{n \to \infty}c_{B'}(x_n)$ since $e$ is not a monopoly at any point in the sequence. So $c_{B_e}(X) \geq c_{B'}(X)$, which implies that by our definition, $e$ is a non-monopoly at $x_1$ as well since there exists some bundle with equal price.

The following simple lemma helps us to establish that at the limit $X$, as long as the order does not change, the prices remain continuous.

\begin{lemma}
Let $\mathcal{M}$ be the set of monopolies for the sequence $(x_1, \ldots, x_n)$ converging to $X$ which become tight in the order $e_1, e_2, \ldots, e_M$. Let $k \leq M$ denote the largest index such that at $X$, the monopolies become tight in the order $e_1, e_2, \ldots, e_k$. Then for all $1 \leq i \leq k$, $\lim_{n \to \infty}\bar{p}_{e_i}(x_n)=\bar{p}_{e_i}(X)$.
\end{lemma}
We can prove this inductively. Suppose that $B_1$ is the bundle that is $e_1$'s witness at each $x_n$. Then $\bar{p}_{e_1}(x_n)= \frac{c_{B_1}(x_n) - c_{B^*}(x_n)}{|\mathcal{M}| - |\mathcal{M} \cap B_1|}$, where $B^*$ is some bundle with non-zero allocation. The denominator is the same at $X$ and the numerator is a continuous function. Therefore $\bar{p}_{e_1}(x)$ is continuous. Similarly we can show that $\bar{p}_{e_i}(x)$ for all $i \leq k$ is also continuous. Now, using this we prove that at $x$, both the order in which the monopolies become tight and the witness bundles remain the same for some tie-breaking rule.

Once again assume that at $X$,  $(e_1, \ldots, e_k)$ are the first few monopolies that become tight in the same order as in the sequence $(x_1, \ldots x_n)$ as defined in the above lemma.

\item \textbf{Order} Suppose the monopoly which becomes tight after $e_k$ is $e \neq e_{k+1}$. In the limit $x_n \to X$, $e_{k+1}$ becomes tight before $e$. Suppose $B$ is the bundle which makes $e$ tight at $X$ and $B_{k+1}$ is the witness bundle for $e_{k+1}$ in the sequence of $x_i$'s. Consider the instant at which $e$ becomes tight for a magnitude of $X$ and let $M^k$ be the set of active monopolies at this instant including $e$. $\bar{p}_e(X) = \bar{p}_{e_k}(X) + \frac{p^k_{B}(X) - p^k_{B^*}(X)}{M^k - M^k \cap B}$. We also introduce two new quantities $\Delta_e(x)$ and $\Delta_{e_{k+1}}(x)$ for $e$ and $e_{k+1}$ at every $x$, which are defined as,

$$\Delta_e(x) = \frac{p^k_{B}(x) - p^k_{B^*}(x)}{M^k - M^k \cap B},$$
$$\Delta_{e_{k+1}}(x) = \frac{p^k_{B_{k+1}}(x) - p^k_{B^*}(x)}{M^k - M^k \cap B_{k+1}}.$$

These quantities refer to the price increase of every active monopoly required after $e_k$ became inactive to make the goods $e$ and $e_{k+1}$ tight respectively. Recall that the price of every active monopoly is increased uniformly so during any one round, the price increase of the active monopolies is the same.

In particular, since for the $x_i$'s $e_{k+1}$ becomes tight before $e$, this means that $\Delta_{e_{k+1}}(x_i)$ is exactly equal to $\bar{p}_{e_{k+1}}(x_i) - \bar{p}_{e_{k}}(x_i)$. This is the price increase of $e_{k+1}$ after $e_k$ becomes tight. This quantity has to be smaller than $\Delta_e(x_i)$ since $e_{k+1}$ becomes tight first. Similarly at $X$, $\Delta_e(X)$ gives us the price increase of $e$ after $e_k$ becomes tight. However, since $\Delta_{e_{k+1}}(X)$ is the increase in price of active monopolies required to make $e_{k+1}$ tight after $e_k$ became inactive, this means that $\Delta_{e_{k+1}}(X) > \Delta_e(X)$.

We can now infer that
$$\lim_{n \to \infty}\Delta_e(x_n) = \Delta_e(X) < \Delta_{e_{k+1}}(X) = \lim_{n \to \infty}\Delta_{e_{k+1}}(x_n).$$

This is however a contradiction since we know that for all $x_n$, $\Delta_e(x_n) \geq \Delta_{e_{k+1}}(x_n)$ because $e_{k+1}$ becomes tight before $e$. This way, we can inductively prove that the order in which the monopolies become inactive at $X$ is the same as in the sequence $(x_1,\ldots, x_n)$.

\item \textbf{Bundle} For an inactive monopoly $e_k$, suppose a new bundle $B'_k$ makes it tight at $X$ as opposed to the limit where the certificate is $B_k$. Without loss of generality, we can assume that for all $e$ in the order before $e_k$, the certificate bundle remains the same (or just take the first such $k$ for which the certificate changes strictly for $e_k$). Clearly the price of the bundle $B_k$ as $x_n$ approaches $X$ at the discrete instant $k$ is smaller than or equal to that of the bundle $B'_k$. Mathematically, this means that in the limit,
$$\Delta_{B_k}(x_n) := \frac{p^{k-1}_{B_{k}}(x_n) - p^{k-1}_{B^*}(x_n)}{M^{k-1} - M^{k-1} \cap B_{k}} \leq \frac{p^{k-1}_{B'_{k}}(x_n) - p^{k-1}_{B^*}(x_n)}{M^{k-1} - M^{k-1} \cap B'_{k}}=: \Delta_{B'_k}(x_n),$$

where $\Delta_{B_k}(x)$ and $\Delta_{B'_k}(x)$ are defined as mentioned above, analogous to our previous definitions of $\Delta_e(x)$. But since $B_k$ is not a witness bundle at $X$, this means $\Delta_{B_k}(X) > \Delta_{B'_k}(X)$, which is a contradiction since these are all continuous functions. And so $B_k$ has to be a witness bundle at $X$.

%The price on the bundle cannot have any jump discontinuities because it is a simple function of the marginal costs and the $\pe$ for the items that became inactive before $e_1$, which are all continuous.

\item \textbf{Active Monopoly.} Is it possible that $e \in M_A$ as $x_n \to X$ but this is not true at $X$? We have already shown that at $X$, all the inactive monopolies become tight in the same order with the same witness bundle. Is it possible that a monopoly which was previously active becomes inactive in the final iteration at $x=X$? Let this monopoly be $e_{k+1}$ and $(e_1, \cdots, e_k)$ represent the monopolies which became inactive before $e_{k+1}$ (in that order). Let $B_{k+1}$ be the witness bundle that makes $e_{k+1}$ tight at $X$. Notice that since this bundle does not make $e_{k+1}$ tight for $x_n$, this can only mean

\begin{equation}
\label{eqn_activemonop}
(\sum_{i=1}^k \bar{p}_{e_i}(x_n)) + M_A\displaystyle\left(\bar{p}_{e_k}(x_n) + \frac{p^k_{B_{k+1}}(x_n) - p^k_{B^*}(x_n)}{M^k - M^k \cap B_{k+1}}\right) \geq \lambda(x_n) - r(x_n),
\end{equation}

where $M_A$ is the number of active monopolies at each $x_n$ as per the given profile vector. We now explain the above inequality: the first term in the LHS refers to the total increased price for monopolies $e_1$ through $e_k$. The second term inside the parenthesis is the total price increase required after $e_{k}$ becomes inactive if $e_{k+1}$ is to become tight with bundle $B_{k+1}$. Since at $x_n$, $e_{k+1}$ is a active monopoly, the term inside the parenthesis acts as an upper bound for $e_{k+1}$'s increased price. This because because for an active monopoly, the price of the monopoly cannot become tight before the algorithm terminates, by definition.

All the $M_A$ active monopolies have the same increased price and therefore, the second term in the LHS is simply an upper bound for the increased price of all active monopolies. The inequality implies that an upper bound on the total increased price of all monopolies cannot be smaller than the available slack $\lambda(x_n) - r(x_n)$.

Suppose that at $X$, $e_{k+1}$ is to be an inactive monopoly. We have already established the continuity of all the functions in the above inequality. Then taking the limit of the above equation as $x_n$ tends to $X$, would give us a jump discontinuity at $X$. This is a contradiction and therefore, by our definition of an active monopoly $e_{k+1}$ has to remain active at $x=X$.

\item \textbf{Inactive Monopoly} Suppose that an inactive monopoly $e_k$ becomes active at $X$. Without loss of generality, let $k$ be the smallest index for which an inactive monopoly becomes active and let $B_k$ be $e_k$'s certificate as $x_n$ approaches $X$. Since the monopoly is inactive for all $x_n$, we have the following inequality analogous to Inequality~\ref{eqn_activemonop} for active monopolies.

$$(\sum_{i=1}^{k-1} \bar{p}_{e_i}(x_n)) + (M - (k-1))\displaystyle \left(\bar{p}_{e_{k-1}}(x_n) + \frac{p^{k-1}_{B_{k}}(x_n) - p^{k-1}_{B^*}(x_n)}{M^{k-1} - M^{k-1} \cap B_{k}}\right) \leq \lambda(x_n) - r(x_n).$$

The first term in the left hand side of the above inequality is the increased price for the monopolies $e_1$ through $e_k$ and the second term inside the large parenthesis is the increased price of monopoly $e_k$, which acts as a lower bound for the price of the other monopolies. What the inequality represents is that the sum of the increased prices of all monopolies cannot exceed the total slack $\lambda(x) - r(x)$. This inequality should hold as in the limiting case of $x_n$ tending to $X$ and therefore at $X$ since these are all continuous functions. However, that would mean that by definition $e_k$ cannot be a strictly active monopoly at $X$. So $B_k$ still remains $e_k$'s certificate at that point.
\end{itemize}

Now, it is not hard to reason that no combination of these events can lead to $\vec{v}$ becoming an invalid profile vector at $X$.
\end{proof}

\begin{claim}
\label{corr_prices are continuouseverywhere}
For every good $e$, the increased price on that good $\pe$ is continuous for all $x$ in $[0, x^*]$. Moreover, the increased price on the active monopolies $\Gamma(x)$ is also a continuous function of $x$ in any region where $M_A$ is non-empty.
\end{claim}
\begin{proof}
Fix any value of $x$, say $x_0$. We want to show that $\lim_{x \to x^-_0}\pe(x) = \lim_{x \to x^+_0}\pe(x) = \pe(x_0)$ for every single good $e$. Let's pick any one direction, say $x \to x^+_0$.

Let $V$ represent the minimal but complete set of profile vectors that are all active as $x$ approaches $x^+_0$. Notice that $V$ is finite and non-empty since at least one profile vector has to be active at every point in the limit and there are only a finite number of profile vectors.  When $|V|=1$, then only a single profile vector $\vec{v}$ is active and using Lemma~\ref{lem_ivclosed}, $\vec{v} \in S_{x_0}$. Further, using the continuity argument in Lemma~\ref{lem_ivcontinuous}, we see that $\pe(x)$ is continuous as $x$ approaches $x_0$ from the right. The case where $|V| > 1$ is slightly more tricky and so we will resort to a more fundamental definition of continuity.
%
%Now, either $\exists$ a single profile vector $\vec{v}$ and $x_1 > x_0$ such that for all $x \in (x_0, x_1)$, $\vec{v} \in S_x$ or there is a collection (set) of profile vectors $V$ which are all active in the limit $x \to x^+_0$. Formally, we can state the second part of the above statement as: $\exists$ a minimal set of profile vectors $V$ with $|V| > 1$ and a $x_1 > x_0$, such that for every $x \in (x_0, x_1)$, at least one $\vec{v} \in V$ also belongs to $S_x$. Moreover, for every such $x$ and every $\vec{v} \in V$, $\exists$ $x_0 < x' < x$ such that $\vec{v} \in S_{x'}$. That is, these are exactly the profile vectors that belong to $\lim_{x \to x^+_0}S_x$.

For every $\vec{v} \in V$, define $X_v$ to be the set of points $x$ in $(x_0, x_1)$ such that $v \in S_x$ (and in decreasing order). By definition, $X_v$ should be an infinite set of points converging to $x_0$. Using Lemma~\ref{lem_ivclosed}, we can say that every $\vec{v} \in V$ also belongs to $S_{x_0}$. Suppose we represent $X_v$ in the form of $(x^{(1)}, \ldots, x^{(n)}, \ldots)$ so that $\lim_{n \to \infty}x^{(n)} = x_0$. Then Lemma~\ref{lem_ivcontinuous} tells us that $\lim_{n \to \infty}\bar{p}_e(x^{(n)}) = \bar{p}_e(x_0)$ for each good $e$.

This means (applying the $\epsilon - \delta$ definition of continuity), for every $\vec{v} \in V$ and every $\epsilon > 0$, there must exist a positive integer $n_v$ such that for all $n > n_v$, $\bar{p}_e(x_0) - \epsilon \leq \bar{p}_e(x^{(n)}) \leq \bar{p}_e(x_0) + \epsilon$. Define,
$$\delta = \min_{\vec{v} \in V}{(x^{(n_v)} - x_0)}.$$

Now we can say that for every $\epsilon > 0$, there exists a $\delta > 0$ such that for all $x_0 \leq x < x_0 + \delta$, we have

$$\bar{p}_e(x) - \epsilon \leq \bar{p}_e(x_0) \leq \bar{p}_e(x) + \epsilon.$$

Therefore, we can conclude that for every good $e$, $\pe(x)$ is continuous at every $x$. Moreover, this would imply that the total price $p_e(x) = c_e(x) + \pe(x)$ is also continuous since the marginal costs are also continuous. Finally, the term $\Gamma(x)$ is simply the increased price of some active monopoly at $x$, which also has to be continuous. This is indeed true, because the set active monopolies is closed and all active monopolies have the exact same increased price.
\end{proof}

\begin{lemma}
\label{lem_activeprices}
If the set of active monopolies $M_A(x)$ at some $x$ is non-empty, then for $e \in M_A(x)$, $\bar{p_e}(x) \geq \frac{\lambda(x) - r(x)}{M}$, where $M$ is the number of virtual monopolies at the minimum cost allocation of magnitude $x$.
\end{lemma}
\begin{proof}
This is not hard to see since we begin by pricing all edges at the marginal cost and distribute the slack among the remaining virtual monopolies. Further, if there are still active monopolies at termination, then it means that the final price of $\lambda(x)$ has been reached. The total slack must therefore equal $\lambda(x) - r(x)$. Moreover, for some items, we stop increasing in the middle, so among all goods, the active monopolies must have the largest value of $\bar{p_e}(x)$ which is at least $\frac{\lambda(x) - r(x)}{M}$ since the prices are increased uniformly.
\end{proof}

Now we can show that for $MPE$ functions, $\exists$ at least one non-trivial equilibrium point via the following theorem.
\begin{theorem}
\label{thm_existencegenmodel}
As long as the demand function is continuously differentiable and belongs to the class MPE, there exists at least one $\tilde{x} \leq x^*$ which is a Nash Equilibrium.
\end{theorem}
\begin{proof}
First consider the case when running the algorithm at the optimum point $x^*$ returns a non-empty set $M_A(x^*)$. Define $x_0$ to be the smallest $x \leq x^*$ such that for all $x \geq x_0$, the set of active monopolies is non-empty. Since the region is closed, we know that $x_0$, all the monopolies belonging to $M_A(x_0)$ are both active and in-active and therefore the total price of all optimal bundles at $x_0$ is $P(x_0) = \lambda(x_0)$. Recall that $\Gamma(x)$ is the quantity $\bar{p}_e(x)$, the increased price for all active monopolies.

Since Lemmas~\ref{lem_eqnconditionsincr} and~\ref{lem_eqnconditionsdecr} did not depend on flows anywhere, the same argument applies for general markets as well. First suppose that at $x^*$, $\Gamma(x^*) \geq x^*|\lambda'(x^*)|$, then we set $\tilde{x}=x^*$ and we are done. This is because both the non-monopolies and the inactive monopolies are tight. The active monopolies satisfy the condition that $\pe \geq x|\ldx|$ at the given value. Now what if $\Gamma(x^*) < x^*\lambda'(x^*)$?

Suppose that at $x_0$, $\Gamma(x) \geq x_0|\lambda'(x_0)|$. Then we claim that there exists a $\tilde{x} \geq x_0$ satisfying the equilibrium condition, namely $\Gamma(\tilde{x}) = \tilde{x}\lambda'(\tilde{x})$. This follows because $\Gamma(x) - x|\ldx|$ is a continuous function so if it is positive at one point ($x_0$) and negative at the other ($x^*$), then clearly there must exist a root somewhere in that interval. It is easy to see that if such a point does indeed exist, it satisfies all our equilibrium conditions. First no active monopoly can increase or decrease the price due to Lemmas~\ref{lem_eqnconditionsincr} and~\ref{lem_eqnconditionsdecr}. Look at any inactive monopoly $e$. Clearly for this item, we know $\pe(\tilde{x}) \leq \Gamma(x)$ because active monopolies have the largest increase in price. So this satisfies condition (3) of Lemma~\ref{lem_eqnconditionsdecr}. Moreover, inactive monopolies cannot increase their price because by definition, there exists some other bundle not containing this and having the same price as all the used bundles.

Finally, what about $\Gamma(x_0) < x|\lambda'(x_0)|$? In this case, we set $\tilde{x} =x_0$ and claim that it is a Nash Equilibrium. First, it is not too hard to see that if this condition is true, then $x_0 > 0$. Let $N$ be the size (cardinality) of the largest bundle. Clearly at any value of $x$, $\Gamma(x) \geq \frac{\lambda(x) - r(x)}{N}$ according to Lemma~\ref{lem_activeprices}. For $x_0$, it is true that $\frac{\lambda(x_0) - r(x_0)}{N} \leq \Gamma(x_0) < x_0|\lambda'(x_0)|$. If it is the case that $x_0 = 0$, then the function $\frac{x|\ldx|}{\lx}$ tends to a non-zero number as $x \to 0$ which violates the requirement for MPE functions. And so, $x_0 > 0$.

We now claim that if at $x_0$, $\Gamma(x_0) < x|\lambda'(x_0)|$, then this point along with the prices returned by our algorithm form a Nash Equilibrium. First for every monopoly edge $e$, $\pe(x_0) \leq \Gamma(x_0) \leq x_0|\lambda'(x_0)|$, so no edge would wish to decrease its price as per Condition (3) of Lemma~\ref{lem_eqnconditionsdecr}. The monopolies cannot increase their price either because they are all inactive. And so we are done. Finally, the case when at $x^*$, there are no active monopolies is actually quite simple. It is not hard to show that our pricing rule leads to a Nash Equilibrium because all sellers are tight.
\end{proof}

Recall that the main tool that enabled us to show all the efficiency results was the Equilibrium conditions in~\ref{corr_eqconditionssuff}. We show that similar conditions apply here when the equilibrium is not optimal.

\begin{corollary}
There exists some $M' \leq M(\tilde{x})$ such that for any demand function $\lambda$ in the class MPE, there exists a non-trivial equilibrium obeying
\begin{enumerate}
\item Either $\lambda(\tilde{x}) - r(\tilde{x}) = M'\tilde{x}|\lambda'(\tilde{x})|$ or $\tilde{x}$ is optimal.
\item Every non-monopoly and non-VM edge is priced at its marginal cost.
\end{enumerate}
\end{corollary}
\begin{proof}
This statement is not hard to see. Notice from Lemma~\ref{lem_activeprices}, that at the equilibrium, any active monopoly $e$ must have $\bar{p_e}(\tilde{x}) = \tilde{x}|\lambda'(\tilde{x})|$ and $\bar{p_e}(\tilde{x})$ lies between $\lambda(\tilde{x}) - r(\tilde{x})$ and $\frac{\lambda(\tilde{x}) - r(\tilde{x})}{M}$. If there are no active monopolies at equilibrium, then let $e$ be an edge which becomes inactive at termination. Clearly, this edge must satisfy  $\frac{\lambda(\tilde{x}) - r(\tilde{x})}{M} \leq \bar{p_e}(\tilde{x}) \leq \tilde{x}|\lambda'(\tilde{x})|$, so for some $M'$, the above equality must hold.
\end{proof}

All our previous efficiency bounds therefore hold, with $M$ being the number of virtual monopolies. It is important to recall here that even in the absence of `pure' monopolies, equilibria are no longer efficient. This is because the equilibrium may be dominated by virtual monopolies. 

\begin{theorem}
For buyers with identical sets of valid bundles $B_i$, all the results from Section~\ref{sec:efficiency} hold, but with $M$ being the number of {\em virtual monopolies} of the equilibrium solution $\tilde{x}$, instead of the number of monopolies.
\end{theorem}

While in the worst case, the number of virtual monopolies can be as large as the number of sellers, for markets with reasonable (but not perfect) competition it is likely to be a lot less.

\section{Proofs from Section 6: Multiple Source Networks and Efficient Equilibrium}
\label{appsec:pos1}
\begin{thm_app}{thm_specposcap}
For any single-source single-sink instance with edge capacities, $M$ monopolies, and demand function $\lx = ax^{1/r}$, the efficiency is one as long as $r > M$.
\end{thm_app}

\begin{proof} Consider any optimal flow $(x^*_e)$ of size $x^*$. Clearly, since the demand function is defined for all $x > 0$, at the optimal flow $x^*$,  the minimum cut of the network is saturated. First assume that at least one monopoly edge is saturated. That is, for some monopoly $c_e = x^*$. For ease of exposition, we prove the claim when $\lx = x^{-1/r}$, the same proof applies when the function is scaled by a constant. Observe that $x\ldx = -\frac{1}{r}x^{-1/r} = -\frac{1}{r}\lx$. We now price each edge as follows. If it is a non-monopoly edge, then set $p^*_e = 0$; if it is a monopoly edge which is unsaturated, set $p^*_e = \frac{1}{r}\lambda(x^*)$, and distribute the remaining price equally on the saturated monopolies. Let $U$ denote the set of unsaturated monopolies and $S$, the set of saturated monopolies.

Note that since $|S| + |U| = M$ and $r > M$, then the price of each saturated monopoly $p^*_e > \frac{1}{M}\lambda(x^*)$. We now claim that these prices along with the flow of $(x^*_e)$ form a Nash equilibrium for the given instance. First observe that the cost of all paths exactly equals $\lambda(x^*)$, so this is indeed a best-response flow to the prices.

No non-monopoly edge has any incentive to raise or lower its price since their price is zero. Now, consider the saturated monopolies. Clearly, these monopolies do not wish to decrease their price as they are already saturated and cannot sustain a larger flow. Suppose one of these monopolies $e$ increases its price to $p$, resulting in a new flow of $x \leq x^*$. Consider the profit of this monopoly at this new price and flow: $\pi(x)=px=[p_e^*+\lambda(x)-\lambda(x^*)]x$. To show that this deviation is not beneficial, we will show that $\pi(x)$ is maximized at $x=x^*$.

The derivative of $\pi(x)$ is
$$\pi'(x)=p^*_e-\lambda(x^*)+\lambda(x)+x\lambda'(x) > \frac{1}{M}\lambda(x^*)-\lambda(x^*)+\lambda(x)-\frac{1}{r}\lx,$$
due to our observations about $p^*_e$ and $x\lambda'(x)$ above. Thus, to show that $\pi(x^*)\geq \pi(x)$ for all $x\leq x^*$, we only need to show that the above is at least 0, i.e., $$\frac{M-1}{M}\lambda(x^*)\leq\frac{r-1}{r}\lambda(x).$$
This is true since $r>M$ and $\lambda$ is non-increasing, as desired. Thus, all saturated monopolies have no incentive to change their price.

 %with its new profit $\pi(x) = px$. Then, we have
%$$\frac{d}{dx}pro = p + x\frac{dp}{dx} = p + x\ldx = p - \frac{x}{r}\lx.$$ Suppose that $p_0$ is the price of all the other monopolies, we know that $p_0 + p^*_e = \lambda(x^*)$. Since $p^*_e \geq \frac{1}{M}\lambda(x^*)$, $p_0 \leq (1 - \frac{1}{M})\lambda(x^*)$. At the new price $p$, we know that $p + p_0 = \lx$, so
%$$p = \lx - p_0 \geq \lx - \frac{M-1}{M}\lambda(x^*).$$ So, we can conclude that
%$$\frac{d}{dx}pro = p - \frac{x}{r}\lx \geq \lx - \frac{M-1}{M}\lambda(x^*) - \frac{x}{M}\lx \geq \frac{M-1}{M}(\lx - \lambda(x^*)).$$
%
%Since $x \leq x^*$, $\lx \geq \lambda(x^*)$, so $\frac{d}{dx}pro \geq 0$ for all $x \leq x^*$ or equivalently $\frac{d}{dp}pro \leq 0$ for all $p \geq p^*_e$. In other words, whatever be the monopoly's profit upon increasing its price to $p$, it could have made  better profit by increasing it by a smaller amount. By this recursive reasoning, its profit is maximized at $p=p^*_e$, so saturated monopolies are at equilibrium.

Consider any unsaturated monopoly $e$, with a price of $p^*_e=\frac{1}{r}\lambda(x^*)$. Note that if $p_0$ is the total price of all the other monopolies, then $p_0 = \frac{r-1}{r}\lambda(x^*)$. Suppose this monopoly changes its price to $p$ and the resulting flow is $x$. Then, by the same argument as above, we know that its profit is $\pi(x)=px=[p_e^*+\lambda(x)-\lambda(x^*)]x$, and so
$$\pi'(x)=p^*_e-\lambda(x^*)+\lambda(x)+x\lambda'(x) = \frac{r-1}{r}(\lx - \lambda(x^*)).$$
Since $\lx$ is non-increasing, we know that for all $x \leq x^*$, $\pi'(x) \geq 0$, and for all $x\geq x^*$, $\pi'(x)\leq 0$. Thus, $\pi(x)$ is maximized at $x^*$, and so there is no incentive for $e$ to change its price.

The above argument of stability was only for the case when at least one monopoly is saturated. Now, suppose that no monopoly is saturated. Then we price all monopolies at $p^*_e = \frac{1}{r}\lambda(x^*)$. Since the minimum cut must be saturated in an optimum flow, this cut consists only of non-monopoly edges. Distribute the remaining surplus of $p=\lambda(x^*) - \frac{M}{r}\lambda(x^*)$ among any minimum cut. That is for any edge in the minimum cut, set the price equal to $p$. Clearly every $s-t$ path with flow is now priced at $\lambda(x^*)$, while all other paths are more expensive. Thus, $x^*$ is a best-response flow for these prices. It is not hard to reason that this is an equilibrium.
\end{proof}

\subsection{Bad Examples for Multiple-Source Networks.}
\label{app:multsourcesink}

We show two bad examples that illustrate the difficulty in extending our results for the single-source case to multiple-sources, i.e., when different buyers have different desired bundles. 
\begin{claim}
There exists an instance of our two-stage pricing game with two sources and a single sink that does not admit a Nash Equilibrium even when the buyers have uniform demand functions.
\end{claim}
\begin{example}
\label{ex:multsourcenoeq}
Let the two source nodes be $s_1$ and $s_2$. There are two paths between $s_1$ and $t$, a direct edge $e_1 = (s_1,t)$, and a simple path $s_1 \rightarrow s_2 \rightarrow t$. In other words, $s_1$ has two preferred bundles $\left\lbrace (e_1), (e_2,e_3) \right\rbrace$, and $s_2$ has exactly one bundle $e_3$ that it desires. Source $s_1$ has a demand for one unit of flow which it values at $\lambda_1 = 100$, and $s_2$ also has demand for one unit of flow, which it values at $\lambda_2 = 25$. We assume that $e_1$ has a cost function $C_1(x) = 3x$ and $e_3$ has a cost $C_3(x) = 2x$. $e_2$ has no production cost. The same example can actually be shown to hold for more convex cost functions, although for simplicity we stick to these linear cost functions.
\end{example}
\begin{proof}
The following points establish that no Nash Equilibrium exists. We will let $p_i$ denote the price of edge $e_i$.
\begin{enumerate}
\item In any Nash equilibrium, $e_3$ makes a profit of at least $23$, else it would change its price to $25-\epsilon$ and receive at least one unit of flow, obtaining a profit of at least $25-\epsilon - 2(1)$. Thus, since the total flow possible on $e_3$ is at most size 2, it must be that $p_3\geq 13.5$ in every Nash equilibrium (since then $2p_3-2(2)=23$).

\item In any Nash equilibrium, all of $s_1$'s flow must be on edge $e_1$. If not, then $e_1$ can always lower its price for more flow since $p_3\geq 13.5$. 

\item Since $e_3$ is only receiving flow from $s_2$ in a Nash equilibrium, it must be that $p_3=25$. Otherwise $e_3$ would benefit from changing its price: if $p_3<25$ it could raise the price by some tiny amount and still receive the same flow; if $p_3>25$ then its profit is 0 since it receives no flow. Therefore $p_1\leq 25+p_2$: if it were larger then $e_1$ would receive 0 flow because it would take the $e_2,e_3$ path.

\item Now consider the prices $p_1\leq 25+p_2$ and $p_3=25$. If $p_1>25$, then $p_2$ could lower its price to be small enough until $25+p_2<p_1$, and then $e_2$ would receive positive profit because it would receive flow from $s_1$. Thus, $p_1\leq 25$. $e_1$ is receiving all of $s_1$'s flow, so it could raise its price unless $p_1=25+p_2$; since $p_1\leq 25$ this implies that $p_1=p_3=25$ and $p_2=0$ at every Nash equilibrium.

\item $p_1=p_3=25$ and $p_2=0$ is not a Nash equilibrium, however: $e_3$ has incentive to lower its price to $25-\epsilon$, and receive utility $2(25-\epsilon)-4$ instead of $25-2$. \end{enumerate}\end{proof}

\begin{clm_app}{ex:nomonopolyeq_main}
There exists an instance with two sources, one sink, and no monopoly edges for either source, where all equilibria are inefficient.
\end{clm_app}
\begin{proof} Consider the instance shown in the Figure~\ref{fig_monopolyposgeq1_main} with two sources $s_1$ and $s_2$. Let the demand of $s_1$ be $\lx = 1-x$ for $x\leq 1$ and 0 afterward, and the demand of $s_2$ be $\lx = 4-x$ for $x\leq 4$ and 0 afterward. At the unique optimum point $s_1$ sends a flow of $1/3$ on its direct link (edge $e_1$) to the sink and $s_2$ sends $2$ units of flow, one on each of its paths.

First, we claim that there exist no set of prices stabilizing a flow of $\frac{1}{3}$ for $s_1$ and $2$ units for $s_2$ divided equally on both the paths. Since $s_1$ is sending a total of a third of a unit of flow, the price on any flow path must be $1-\frac{1}{3} = \frac{2}{3}$. Similarly, the price on any flow-containing path for $s_2$ must be $2$. Any set of prices stabilizing this flow must therefore have $e_1$ being priced at $p=\frac{2}{3}$. Consider the flow carrying path for $s_2$ consisting of the edges $(s_2,i_1)$ and $(i_1,t)$. We claim that in any Nash Equilibrium with this flow, the edge $e_2=(i_1,t)$ must be priced at $p_2 \geq 1$. Indeed, if the price on this edge is less than $1$, then its profit is $p_2*1 - 1 < 0$, which means that the edge is no longer stable.

At the current prices, look at the two paths available for $s_1$ to send flow on. It's direct edge $e_1$ is priced at $p_1=\frac{2}{3}$ and the other path has a price no smaller than $p_2 \geq 1$. So the best-response is to always send flow on the edge $e_1$ and so $e_1$ is a virtual monopoly for this source (refer Section~\ref{appsec:generalizations} for a formal definition). Can these prices be a Nash Equilibrium? We claim that this is not the case. Suppose $e_1$ increases his price to $p'_1=\frac{3}{4}$. Note that for $s_1$, the $e_1$ path is still strictly cheaper than the other path at the current prices, so its best-response for this price would be to send a flow of $\frac{1}{4}$ on this edge. The seller's new profit is $\frac{3}{4}\times \frac{1}{4} - \frac{1}{4^2} = \frac{1}{8}$ which is greater than its old profit of $\frac{1}{9}$. The original solution is therefore not stable for the sellers.\end{proof}

\begin{thm_app}{thm_seriesparallelpos}
A multiple-source single-sink series-parallel network with no monopolies admits a Nash Equilibrium that has the same social welfare as the optimum solution for any given instance where all edges have $c_e(0)=0$.
\end{thm_app}

\begin{proof}
We begin by showing two simple structural properties on MS series-parallel graphs. If $P$ is a given path on a graph, we use the notation $P_{uv}$ to refer to the sub-path of $P$ between nodes $u$ and $v$.

\begin{lemma}
\label{sublem_sublem_serparbridge}
Let $E_1$ and $E_2$ be two node-disjoint paths in the graph between nodes $u$ and $v$ and $w \in E_2$. Then any $w$-$t$ path must go through $v$.
\end{lemma}
\begin{proof}
The series-parallel nature of the directed graph precludes the existence of bridges. Formally, if $P_1$ and $P_2$ are node-disjoint paths between two given nodes, then for any $i \in P_1$ and $j \in P_2$, there cannot exist a path between $i$ and $j$ (and vice-versa)~\cite{chang1981characterization}. For our given lemma, assume by contradiction that there exists a $w$-$t$ path $P'$ not containing $v$. Let $P$ be any $v$-$t$ path and suppose that $P'$ first intersects $P$ at some node $z$.

Consider the following two paths between the nodes $u$ and $z$: $E_1 \rightarrow P_{vz}$ and $E_{2_{uw}} \rightarrow
P'_{wz}$. Clearly, these two paths are mutually disjoint. But there exists a bridge between $w$ and $v$, namely $E_{2_{wv}}$, which contradicts the series-parallel property.
\end{proof}

\begin{lemma}
\label{lem_nosubs}
Let $E_1$ and $E_2$ be any two node-disjoint paths between two nodes $u$ and $v$ and $w \in E_2$. Then any $s_j$-$w$ path must go through $u$, where $s_j$ is any source node.
\end{lemma}
\begin{proof}
The series-parallel nature of the directed graph precludes the existence of bridges~\cite{chang1981characterization}. Suppose that the premise is false and there is some $s_j$-$w$ path $P'$ not containing $u$. Let $P$ be some $s$-$u$ path in the graph. Among all the nodes on $P$ that the path $P'$ intersects, let $z$ be the one the one closest to $u$ (along path $P$). In other words, there exist a path from $z$ to $u$ along $P$, and $z$ to $w$ along $P'$ that do not have any nodes in common. Now consider the following two mutually disjoint paths, $P_{zu} \rightarrow E_1$ and $P'_{zw} \rightarrow E_{2_{wv}}$. There exists a bridge between $u$ and $w$ along $E_2$, which contradicts the series-parallel assumption. Therefore, any $s_j$-$w$ path must contain $u$.
\end{proof}

Now consider any optimal solution $\vec{x^*}$. Suppose that we price each edge $e$ at its marginal cost $c_e(x^*_e)$. We claim that these prices along with the optimal flow give us a Nash Equilibrium. One half of the proof is fairly easy: Lemma~\ref{cl_margcost} tells us that no edge would ever decrease its price from the marginal cost irrespective of the amount of flow it gains. So we only need to show that sellers cannot increase their price. We show something stronger namely that if any edge $e$ increases its price from $p_e = c_e(x^*_e)$, then it would lose all of its flow. This would imply that for every edge $e$ and every source $s_i$ sending some non-zero flow on that edge, there exists a $s_i$-$t$ path not containing $e$ with some flow on it (sent by $s_i$).

Consider any edge $e$ with some flow on it (or else its price is zero), and some source $s_i$ which has flow on that edge. The total price of any $s_i$-$t$ path with $s_i$-flow on it is $\lambda_i(x^*)$, where $x^*$ is the amount being sent by source $s_i$, and so all we need to show is that there exists some $s_i$-$t$ path not containing $e$ and having exactly this price. Let $P$ be any flow-carrying $s_i$-$t$ path that contains the edge $e$. Since no single edge monopolizes any source, this means that there must exist a $s_i$-$t$ path not containing $e$. More specifically, there is some predecessor $u$ and successor $v$ of $e$ along the path $P$, such that there exists an alternative $P'$ that has no edge in common with $P_{uv}$. That is $P'$ is a shortcut as we used for no monopolies case with just one source and sink (See Claim~\ref{clm_monopolypos1}).

If we show that (upon marginal pricing), $P_{uv}$ and $P'$ have the same cost, then we are done since the new path $P_{s_iu} \rightarrow P' \rightarrow P_{vt}$ and the original flow carrying path $P$ would have the same price, and thus $e$ would lose everything if it were to increase its price. Suppose that $s_i$ sends non-zero flow on some edge $e_0=(w_1,w_2)$ in $P'$, then we claim that $s_i$ must send some flow on a $u$-$v$ path that has no edge in common with $P_{uv}$. Since $e \in P_{uv}$ and $e_0 \in P'$ and the two paths are disjoint, then there cannot be any bridge connecting these two edges. More specifically, source $s_i$ is sending flow on $e_0$ through some path that does not contain $e$. All flow carrying paths must have the same marginal cost and therefore some price of $\lambda_i(x^*)$ and so for this case, $e$ cannot increase its price as it would lose its flow to the $s_i$-$t$ path not containing $e$. What if $s_i$ does not send flow on any edge belonging to $P'$?

In this case, we would like to show that $P_{uv}$ and $P'$ have exactly the same marginal cost or there exists some other path $P_j$ that is also mutually disjoint from $P_{uv}$ with the same price. Assume that $P_{uv}$ has a  strictly smaller price $\sum_{e \in P_{uv}}c_e(x^*_e)$ than $P'$'s price of $\sum_{e \in P'}c_e(x^*_e)$. Note that $P'$ cannot have the smaller price because $s_i$ has access to both these sub-paths and if $P'$ has a smaller marginal cost, we could have easily shifted some flow on $P'$ and reduced the cost of the optimal solution.

Since $\sum_{e \in P'}c_e(x^*_e) > \sum_{e \in P_{uv}}c_e(x^*_e)$, there must be some flow carrying edge $e_1=(w_3,w_4)$ on $P'$ that contains flow from some source $s_j$. By Lemma~\ref{lem_nosubs}, all $s_j$-$w_3$ paths must pass through $u$. And By Lemma~\ref{sublem_sublem_serparbridge}, all $w_3$-$t$ paths must go through $v$. Consider some path on which source $j$ sends non-zero flow throughout which also contains $w_3$ (say $P_j$). $P_j$ passes through both $u$ and $v$ but is disjoint from $P_{uv}$ because $P_j$ contains $w_3$ which also contained in $P'$. This means that if $P_{uv}$ and $P_{j_{uv}}$ are not disjoint, then there exists two internal nodes $u$ and $v$ in $P_{uv}$ and $P_{j_{uv}}$ (or vice-versa) that are connected which cause a bridge between $P_{uv}$ and $P'$. Therefore, it must be the case that $P_{uv}$ and $P_{j_{uv}}$ have no edge in common.

Now we know that source $i$ sends flow on $P_{uv}$ and source $j$ on $P_{j_{uv}}$. We claim now that $\sum_{e \in P_{uv}} c_e(x^*_e) = \sum_{e \in P_{j_{uv}}}c_e(x^*_e)$. If this is not the case, then one of $s_i$ or $s_j$ could shift their flow between these two sub-paths and reduce the overall cost, which is a contradiction. Therefore, for any edge $e$, there exists an alternative path not containing that edge. Therefore, marginal cost pricing results in  a Nash Equilibrium.
\end{proof}

\begin{clm_app}{clm_multsourceundemand}
There exists efficient equilibrium in multiple-source multiple-sink networks with uniform demand buyers at each source if one of the following is true
\begin{enumerate}

\item Buyers have a large demand and production costs are strictly convex.

\item Every source node is a leaf in the network.
\end{enumerate}
\end{clm_app}
We define these two cases formally before proving the result. 
\begin{enumerate}
\item We consider an arbitrary multiple-source multiple-sink network such that for every source $s_i$, 

$$\lambda_i(x) = \lambda_i \quad 0 \leq x \leq L_i \quad \text{$L_i$ is large}.$$

We also assume that the cost functions are strictly convex, i.e. for every given edge $C_e(x)$ is strictly convex or $\frac{d}{dx}C_e(x) = c_e(x)$ is increasing. 

\item For every source $s_i$, there exists a personal (last-mile) monopoly $e_i$ such that every $s_i$-$t_i$ path contains $e_i$ and for any other $j$, $e_i$ is not reachable from $j$.
\end{enumerate}
\begin{proof} \textbf{(Statement 1)} The proof is relatively straightforward. If the costs are convex, then there is a unique optimum flow $(x^*_e)$. Price each edge at $p_e = c_e(x^*_e)$, where $x^*_e$ is the total flow on that edge at the optimum. We claim that this marginal pricing results in a Nash Equilibrium. First, no source will send more than $L_i$ units of flow. Now, observe that by the properties of optimal flows, for any buyer $i$ with a non-zero flow $x^*(i)$ at the optimum and for any $s_i$-$t_i$ paths $P_i$ with non-zero flow and $P_j$ with no flow, we have
\begin{align}
\label{eqn_shortestpm} \lambda_i &= \sum_{e \in P_i} c_e(x^*_e)  = \sum_{e \in P_i}p_e\\
\label{eqn_shortestpm2} \lambda_i & \leq \sum_{e \in P_j} c_e(x^*)  = \sum_{e \in P_i}p_e
\end{align}
%Similarly for any buyer $i$ with zero flow at the optimum and any path $P_i$ for that buyer
%$$\lambda_i \leq \sum_{e \in P_i} c_e(x^*_e)  = \sum_{e \in P_i}p_e.$$
It is not hard to see that the buyer behavior is a best-response to these prices. Next, consider any seller $e$. By Claim~\ref{cl_margcost}, if they reduce their price, whatever be the resulting flow, the deviation will not be profitable (this claim still holds even in the presence of multiple buyer types, i.e., multiple sources). Suppose that such a seller increases its price. First assume that there is a non-zero flow $x^*_e$ on the edge. Then, from Equations~\ref{eqn_shortestpm}, \ref{eqn_shortestpm2}, for every source $i$ for whom there is a bundle ($s$-$t$ path) containing $e$, and every such bundle $P_i$, we know $\sum_{e \in P_i}p_e \geq \lambda_i$. This means that if the seller increase his price, its flow would drop to zero as it is not profitable for any buyer to send flow on a path costing strictly more than $\lambda_i$. The same applies to edges with no flow on them since they are priced at $p_e = c_e(0)$. This concludes our proof. $\blacksquare$

\textbf{(Statement 2)}The proof is almost identical to the previous case so we only focus on the differences. Once again consider the optimal flow $x^*$ and price each edge at its marginal cost. If the conditions in equations~\ref{eqn_shortestpm}, \ref{eqn_shortestpm2} hold, then marginal pricing gives us a Nash Equilibrium. Suppose that is not the case, then the condition of Equation~\ref{eqn_shortestpm} may not hold. That is for some source $i$ and every $s_i-t_i$ path with non-zero flow $P_i$, it may be the case that
\begin{equation}
\label{eqn_sourcenottight}
\lambda_i > \sum_{e \in P_i}c_e(x^*_e)
\end{equation}

In this case, consider the personal monopoly $e_i$ for that source. Suppose the current price of this edge is $p_{e_i}$, increase the price to $p'_{e_i} = p_{e_i} + \lambda_i - \sum_{e \in P_i}c_e(x^*_e)$, where $P_i$ is any $s_i-t_i$ path with non-zero flow. Since the flow is a min-cost flow all such paths have equal marginal cost, so the increase in price is unambiguous. Repeat this for all sources which satisfy equation~\ref{eqn_sourcenottight}. Note that with the new prices $(x^*_e)$ is still a best-response flow for each buyer since the price of each $s_i-t_i$ path for buyer $i$ is now exactly $\lambda_i$. Also note that we have increased prices on only the leaf edges, so the flow of other sources is not affected by this change.

Now, no node in the network can increase its price as it would lose all the flow. Look at the leaf monopoly edges with increased price: if they decrease their price, then their flow cannot increase since that source node is already sending its full quota of $l_i$ units of flow. Leaf monopolies which are still priced at the marginal have no incentive to decrease their price either. So this set of prices forms a Nash Equilibrium supporting the optimal flow.
\end{proof}
\end{document}